\documentclass[10pt]{article}

\usepackage[utf8]{inputenc}
\usepackage{amssymb,amsmath,booktabs}		
\usepackage{hyperref} 				
\usepackage{float,wrapfig} 		    

\usepackage{mathtools,verbatim}     
\usepackage{amsthm}
\usepackage{caption}
\usepackage{subcaption}
\usepackage{color}
\usepackage{multirow}
\usepackage{ifthen}
\newboolean{usefigures}
\setboolean{usefigures}{true} 

\usepackage{enumerate}

\usepackage[ruled, vlined,linesnumbered]{algorithm2e}
\def\BibTeX{{\rm B\kern-.05em{\sc i\kern-.025em b}\kern-.08em
    T\kern-.1667em\lower.7ex\hbox{E}\kern-.125emX}}

\usepackage{subfiles} 
\usepackage{cleveref} 

\usepackage{yhmath}

\usepackage{tikz}\usetikzlibrary{shapes.geometric, positioning,arrows.meta}
\newcommand{\ray}{3.2cm}
\newcommand{\rnode}[4]{\node[shape=circle,minimum size=0.1,draw=black,fill=#1, label={[label distance=0.1cm] #2:#3}] () at (#4)  {};}

\newcommand{\vertx}[4]{\node[shape=star,minimum size=0.1,draw=black,fill=#1, label={[label distance=0.1cm] #2:#3}] () at (#4)  {};}

\newcommand\INITpoints{0/1/1, 47/0.3/0, 98/1/1, 123/1/1, 182/0.4/0, 220/1/1, 300/0.2/0, 340/0.8/0}
\newcommand\CIRCLEpoints{0/1/1, 47/0.85/0, 98/1/1, 123/1/1, 182/0.75/0, 220/1/1, 300/0.56/0, 340/0.8/0}

\newcommand{\collead}{lightgray}
\newcommand{\colblock}{black}
\newcommand{\coldefault}{black}
\newcommand{\colneutral}{lime}
\newcommand{\colangle}{green}
\newcommand{\colanglem}{green!30}
\newcommand{\colmedian}{magenta}
\newcommand{\colvertx}{pink}
\newcommand{\colbeacon}{orange}
\newcommand{\colregular}{red}
\newcommand{\coleast}{yellow}
\newcommand{\coleastdiameter}{yellow!50}
\newcommand{\colwest}{blue}
\newcommand{\colwestdiameter}{blue!50}
\newcommand{\colblockL}{cyan}
\newcommand{\colblockR}{lime}

\newcommand{\colinchord}{gray!30}
\newcommand{\coloutchord}{gray}
\newcommand\colpadding{teal}

\title{Optimal Uniform Circle Formation by Asynchronous Luminous Robots}

\newcommand{\keywordstring}{Uniform Circle Formation, Robots with Lights, Autonomous Robots, Rank Encoding, Time and Color Complexities, Computational SEC}
\newcommand{\etal}[1]{#1 \textit{et al.}} 

\newcommand{\colr}[1]{\texttt{#1}} 	
\newcommand{\conf}[1]{$C\kern-0.1em on\!f_{#1}$}
\newcommand{\config}{C\kern-0.1em on\!f}
\newcommand{\ucf}{{\sc Uniform Circle Formation}}
\newcommand{\UCF}{{\sc UCF}}
\newcommand{\fsynch}{$\mathcal {FSYNC}$}
\newcommand{\ssynch}{$\mathcal {SSYNC}$}
\newcommand{\asynch}{$\mathcal {ASYNC}$}

\newcommand{\proc}[1]{Procedure {\sc #1}}

\newcommand{\bdcp}{{\sc Beacon Directed Curve Positioning}}
\newcommand{\BDCP}{{\sc BDCP}}

\newcommand{\cR}{\mathcal{R}}

\newcommand{\SC}{\mathcal{C}} 
\newcommand{\chord}{\mathcal{L}} 
\newcommand{\diam}{\rho} 
\newcommand{\barc}{\Gamma} 

\DeclarePairedDelimiter{\floor}{\lfloor}{\rfloor}
\newcommand{\reals}{\mathbb{R}^2}
\newcommand{\enclosed}{\sqsubseteq}
\usepackage{textcomp}
\newcommand{\mycircled}[1]{\textcircled{\raisebox{-0.9pt}{\bf #1}}}

\usepackage[margin=1in]{geometry}

\author{Caterina Feletti\footnote{Department of Computer Science, Università degli Studi di Milano, Italy,
{\it caterina.feletti@unimi.it}} \and Debasish Pattanayak\footnote{School of Computer Science, Carleton University, Canada,
{\it drdebmath@gmail.com}} \and Gokarna Sharma\footnote{Department of Computer Science, Kent State University, USA,
{\it gsharma2@kent.edu}}}

\newtheorem{theorem}{Theorem}
\newtheorem{lemma}{Lemma}
\newtheorem{corollary}{Corollary}
\newtheorem{definition}{Definition}
\newtheorem{observation}{Observation}

\begin{document}
\date{}
\maketitle

\begin{abstract}
    We study the {\sc Uniform Circle Formation} ({\sc UCF}) problem for a swarm of $n$ autonomous mobile robots operating in {\em Look-Compute-Move} (LCM) cycles on the Euclidean plane. 
    We assume our robots are \emph{luminous}, i.e. embedded with a persistent light that can assume a color chosen from a fixed palette, and \emph{opaque}, i.e. not able to see beyond a collinear robot.
    Robots are said to {\em collide} if they share positions or their paths intersect within concurrent LCM cycles.
    To solve {\UCF}, a swarm of $n$ robots must autonomously arrange themselves so that each robot occupies a vertex of the same regular $n$-gon not fixed in advance. 
    In terms of efficiency, the goal is to design an algorithm that optimizes (or provides a tradeoff between) two fundamental performance metrics: \emph{(i)} the execution time and \emph{(ii)} the size of the color palette.
    There exists an $O(1)$-time $O(1)$-color algorithm for this problem under the fully synchronous and semi-synchronous schedulers and a $O(\log\log n)$-time $O(1)$-color or $O(1)$-time $O(\sqrt{n})$-color algorithm under the asynchronous scheduler, avoiding collisions.   
    In this paper, we develop a deterministic algorithm solving {\sc UCF} avoiding collisions in $O(1)$-time with $O(1)$ colors under the asynchronous scheduler, which is asymptotically optimal with respect to both time and number of colors used, the first such result.  
    Furthermore, the algorithm proposed here minimizes for the first time what we call the \emph{computational SEC}, i.e. the smallest circular area where robots operate throughout the whole algorithm.
    \end{abstract}


\section{Introduction}
\subsection{Background and Motivation}
In the field of robotics, a new approach has been established over the past few years: instead of developing \emph{ad hoc} powerful and monolithic robots for specific tasks, \emph{swarm robotics} involves solving different problems by swarms of small, even tiny, simple and autonomous agents, which can be programmed to cooperate.
Besides the practical and engineering factors that robotics must consider (hardware, control, data processing, etc.), other theoretical aspects must be investigated: distributed computability, design of (self-stabilizing) algorithms, correctness and complexity analysis.
According to this, research has proposed various theoretical models to formalize such distributed systems and to develop distributed algorithms on them \cite{Prencipe2013}.

In the classical macro-model \cite{Flocchini2012}, robots are assumed to be {\em autonomous} (no external control), {\em anonymous} (no unique identifiers), {\em indistinguishable} (no external identifiers), and \emph{homogeneous} (execute the same algorithm).
Robots are modeled as pointlike agents that act on the plane through a sequence of \emph{Look-Compute-Move} (LCM) cycles: when a robot is activated, it first obtains a snapshot of its surroundings ({\em Look}), then computes a destination based on the snapshot ({\em Compute}), and finally moves to the destination ({\em Move}).
Most of the literature considers {\em disoriented} robots: each robot has its local coordinate system without any assumption of global orientation.
The sensory capability of a robot, generally called {\em vision}, allows a robot to determine the positions of other robots in its own local coordinate system.
Robots are traditionally assumed to be {\em transparent} (so that each robot has complete visibility of the swarm), {\em oblivious} (they cannot store any information about past LCM cycles), and {\em silent} (they do not have any direct communication means) \cite{Flocchini2012}.

In this paper, we consider a more recent model, the {\em robots with lights} (aka {\em luminous}) model \cite{Das2012,Luna2014,Flocchini2012,Peleg2005}, where robots are embedded with a light whose color can be chosen from a fixed palette and persists cycle by cycle until its next update.
Since such a light is visible to both the robot itself and the other robots, the luminous model grants robots both a persistent internal state (memory) and a direct communication means.
Additionally, we assume our robots are \emph{opaque}, thus they experience {\em obstructed visibility} in case of collinearities (if robots $a,b,c$ are collinear, $a$ and $c$ cannot see each other).
Note that, since by default the size of the swarm $n$ is not known in advance by the robots, they may not realize if they are experiencing complete or obstructed visibility in the {\em Look} phase.

We say that two robots {\em collide} if either \emph{(i)} they share the same position at a given time or \emph{(ii)} their paths towards their destinations intersect within concurrent LCM cycles.
We assume robot movements are {\em rigid}, i.e., in each \emph{Move}, the robot stops only after reaching its destination.    

In this paper, we study the fundamental {\ucf} (\UCF) problem for \emph{luminous-opaque robots}.
Starting from an arbitrary configuration where $n$ robots occupy distinct points on a plane, the problem aims at relocating them autonomously on the vertices of a regular $n$-gon not fixed in advance. 
Namely, {\UCF} is solved if the $n$ robots are equally distributed on the same circle.
Thanks to the well-known properties of such a geometric configuration, {\UCF} is deemed an important special case of the {\sc Geometric Pattern Formation} problem \cite{ColemanKMOUV20,Flocchini2012,Flocchini19,SuzukiY99,VaidyanathanST22}.
A {\em geometric pattern} (or, simply, \emph{pattern}) is defined as a set of points on a plane.
The {\sc Geometric Pattern Formation} problem for a pattern $P$ is said to be solved by a swam of robots if they form a stable configuration (i.e. where robots no longer move) where the set of their positions is similar to $P$ up to scaling, rotation, translation, and reflection.

Related to {\UCF}, it is noteworthy to mention its precursor, the {\sc Circle Formation} problem \cite{Bhagat2018,Datta2013}: robots are required to relocate on the boundaries of the same circle, without any other requirement. 
The uniform version of this problem, namely {\UCF}, has been investigated by multiple works \cite{ChatzigiannakisMN04,DefagoK02,DefagoS08,Dieudonne2006,DieudonneP07,DieudonneP08,FlocchiniPSV17,Katreniak05,MaminoV16,MondalC18,MondalC20} in the classic \emph{oblivious-transparent} robot model (see \cite{Viglietta19} for an excellent survey). 
Many of the cited works focused mostly on the solvability of the problem in finite time, considering different model assumptions (e.g. specific initial configurations \cite{DieudonneP07,Katreniak05}, only 3 or 4 robots \cite{DieudonneP08, MaminoV16}), without examining the time complexity of the proposed algorithm \cite{FlocchiniPSV17,Viglietta19}. 

Time complexity is measured based on the underlying scheduler under which robots operate. 
Generally, literature considers three main schedulers. 
Under the {\em fully synchronous} scheduler ({\fsynch}), the robots perform their LCM cycles in perfect synchrony and hence time is measured in rounds of LCM cycles. 
Under the {\em semi-synchronous} ({\ssynch}) and {\em asynchronous} ({\asynch}) schedulers, a robot may stay inactive for a finite but indeterminate time and hence time is measured using the idea of epoch. 
An {\em epoch} is the smallest time interval within which each robot performs its LCM cycle at least once \cite{Cord-Landwehr11}. 
We will use the term “time” generically to mean rounds for the {\fsynch} and epochs for the {\ssynch} and {\asynch} schedulers.

\etal{Feletti} \cite{feletti2018,FelettiOPODIS23,feletti2023journal,feletti2023} are the first to consider {\UCF} in the \emph{luminous-opaque} robot model. 
They presented a deterministic $O(1)$-time $O(1)$-color solution under the {\fsynch} and {\ssynch} schedulers, and a deterministic $O(\log n)$-time $O(1)$-color solution under the {\asynch} scheduler, avoiding collisions.  
Recently, Pattanayak and Sharma \cite{PattanayakS2024} developed a framework that provides a deterministic $O(x)$-time solution using $O(n^{1/2^x})$ colors under the asynchronous scheduler.  
Setting $x$ to some constant, it gives a $O(1)$-time solution (i.e. asymptotically optimal) with $O(\sqrt{n})$ colors. 
Setting $x=O(\log\log n)$, it gives a $O(\log\log n)$-time solution which uses $O(1)$ colors.

\begin{table}[!t]
    \centering
\resizebox{\textwidth}{!}{
    \def\arraystretch{1.2}
    \begin{tabular}{ccccc}
    \toprule
    {\bf Algorithm} & {\bf Time (in epochs)} & {\bf Number of Colors} & {\bf Computational SEC} & {\bf Scheduler} \\
    \toprule\hline
      \cite{feletti2018} & $O(1)$  & $O(1)$ & Not minimized     & \fsynch  \\ 
    \hline\hline
    \cite{feletti2023} & $O(1)$  & $O(1)$  & Not minimized   & \ssynch  \\ 
    \hline\hline
    \cite{feletti2023} & $O(n)$  & $O(1)$  & Not minimized    &  \multirow{6}{*}{\asynch} \\
    \cline{1-4}
    \cite{FelettiOPODIS23} & $O(\log n)$  & $O(1)$   & Not minimized   &  \\ 
    \cline{1-4}
    {\bf Generic} \cite{PattanayakS2024} & $O(x)$ & $O\left(n^{1/2^x}\right)$ & Not minimized & \\ 
    \cline{1-4}
    {\bf OptTime} \cite{PattanayakS2024} & $O(1)$ & $O(\sqrt{n})$ & Not minimized &  \\
    \cline{1-4}
    {\bf OptColor} \cite{PattanayakS2024} & $O(\log\log n)$ & $O(1)$ & Not minimized &  \\
    \cline{1-4}
    {\bf OptTime\&Color} (this paper) & $O(1)$ & $O(1)$ &  Minimized & \\
    \bottomrule
    \end{tabular}
}
\caption{Existing {\UCF} deterministic solutions for $n\geq 1$ luminous-opaque robots on the plane, avoiding collisions.
    $x\in [1,O(\log\log n)]$.} 
    \label{table:summary}
\end{table}

\subsection{Our Result}
In this paper, we propose a deterministic algorithm for solving {\UCF} in the luminous-opaque model under {\asynch}, avoiding collisions. 
Our algorithm runs in $O(1)$ time with $O(1)$ colors. 
Compared to the state-of-the-art result of Pattanayak and Sharma \cite{PattanayakS2024} which is either optimal in time with $O(\sqrt{n})$ colors or optimal in colors with $O(\log\log n)$ time, our algorithm is simultaneously optimal on both time and number of colors. 
Another important aspect of our proposed algorithm is that, in addition to optimizing both runtime and color complexity, it minimizes a spatial metric that we call  \emph{computational SEC}, i.e. the smallest circle containing all the points the robots touch during the execution of the algorithm.
This objective is driven by a recent research direction on robot computing which aims to solve \emph{space-constrained} versions of classic distributed problems.
Given a problem $P$ for a swarm of robots, a \emph{space-constrained version of} $P$ requires optimizing some spatial metrics (e.g. the maximum distance traveled by a robot \cite{Bhagat2018, ColemanKMOUV20}) while solving $P$. 
Although, in most cases of our algorithm, the computational SEC minimization turns out to be a byproduct of the technique used to optimize time and color complexity, a particular configuration needed an \textit{ad hoc} technique to be designed to minimize the computational SEC.
Note that the existing {\UCF} solutions do not minimize the computational SEC. 
\Cref{table:summary} summarizes our result by comparing it with the existing {\UCF} solutions.

\subsection{Challenges}
Performing in {\asynch}, our algorithm had to face the intrinsic issues that asynchrony entails.
Such issues derive from the discrepancy in the duration of the LCM cycles each robot follows. 
This discrepancy can affect the correctness of algorithmic strategies (e.g. a robot may work on an outdated snapshot of the system it takes on its {\em Look} portion of the cycle), the integrity of the system (e.g. two robots may collide if their motions are not synchronized), and the algorithmic runtime (e.g. robots may be forced to perform sequentially to avoid erroneous computations or collisions).
Moreover, \emph{opaqueness} (and so \emph{obstructed visibility}) posed an additional challenge under {\asynch}, since moving robots can hide or be hidden by other robots.

However, moving in parallel is essential in reaching a $O(1)$-time algorithm.
The main challenge of this work was to make robots exploit parallelism even in conditions of asynchrony and obstructed visibility, always keeping the size of the color palette constant.
Indeed, this paper shows that robots can move in parallel to solve {\UCF} without colliding, using $O(1)$ colors.
As in \cite{feletti2018,FelettiOPODIS23,feletti2023journal,feletti2023,PattanayakS2024}, we factorize {\UCF} into three sub-problems (see \Cref{tab:steps_ucf}): \emph{(i)} {\sc Complete Visibility}, \emph{(ii)} {\sc Circle Formation}, and lastly \emph{(iii)} {\sc Uniform Transformation}.
Given an arbitrary initial configuration of robots on distinct points on the plane, we \emph{(i)} exploit the deterministic {\sc Complete Visibility} result of \etal{Sharma} \cite{SharmaVT17} under {\asynch} to arrange robots on the vertices of a convex polygon in $O(1)$ time using $O(1)$ colors, avoiding collisions. 
After that, \emph{(ii)} a simple procedure will safely move the robots to the perimeter of a circle, say $Cir$, in $O(1)$ epochs using $O(1)$ colors.
The remaining step \emph{(iii)} aims to equally distribute the robots on the perimeter of $Cir$, thus solving {\UCF}.
This last step represents the real challenge for the time-color optimization of {\UCF}.

\begin{table}[!h]
\begin{center}
\begin{tabular}{c|c|c|c}

	\begin{tikzpicture}[scale=0.4, transform shape, font = {\LARGE}]
		\def\r{\ray*0.9}
		\draw [very thin, white] (0,0) circle (\r);
		\foreach \i/\x in \INITpoints 
			{\rnode{\colneutral}{left}{}{\i:\x*\r}; 
			}
	\end{tikzpicture}
	&
	\begin{tikzpicture}[scale=0.4, transform shape, font = {\LARGE}]
		\def\r{\ray*0.9}
		\def\n{7}
		\draw [very thin] (0,0) circle (\r);

            \draw[dashed] (0:\r) -- (47:0.85*\r);
		\draw[dashed] (47:0.85*\r) -- (98:\r);
		\draw[dashed] (98:\r) -- (123:\r);
		\draw[dashed] (123:\r) -- (182:0.75*\r);
		\draw[dashed] (220:\r) -- (182:0.75*\r);
		\draw[dashed] (220:\r) -- (300:0.56*\r);
		\draw[dashed] (300:0.56*\r) -- (340:0.8*\r);
		\draw[dashed] (340:0.8*\r) -- (0:\r);
	
		\foreach \a/\p/\s in \CIRCLEpoints{
			\ifthenelse{\s = 0}
				{\draw[dotted, ->] (\a:\p*\r) -- (\a:\r)}
				{};
			\rnode{\colneutral}{left}{}{\a:\p*\r}; 
		}

	\end{tikzpicture}
	&
	\begin{tikzpicture}[scale=0.4, transform shape, font = {\LARGE}]
		\def\r{\ray*0.9}
		\draw [very thin] (0,0) circle (\r);
		\foreach \a/\p/\s in \CIRCLEpoints{
			\rnode{\colneutral}{left}{}{\a:\r}; 
		}
	\end{tikzpicture}
	&
	\begin{tikzpicture}[scale=0.4, transform shape, font = {\LARGE}]
		\def\r{\ray*0.9}
		\def\n{7}
		\def\a{360/\n}
		\draw [very thin] (0,0) circle (\r);
	
		\foreach \i in {1,...,\n}
			{\draw[dashed] (\i*\a:\r) -- ({(\i+1)*\a}:\r);
		\rnode{\colneutral}{left}{}{\i*\a:\r}; 
			}
		
	\end{tikzpicture}
	\\ 
	\footnotesize{Initial Configuration}&\footnotesize{\sc Complete Visibility}&\footnotesize{\sc Circle Formation}&
	\footnotesize{\sc Uniform Transformation}
\end{tabular}
\end{center}
\caption{Sub-problems composing {\ucf}.}
\label{tab:steps_ucf}
\end{table}

Up to now, step \emph{(iii)} was solved in a binary tournament fashion in \etal{Feletti} \cite{FelettiOPODIS23} achieving an $O(\log n)$-time solution using $O(1)$ colors under {\asynch}. 
Afterward, the step was expedited by Pattanayak and Sharma \cite{PattanayakS2024} by dividing the robots in $\sqrt{x}$ groups in each step for $x$ robots in a group achieving an $O(\log\log n)$-time solution using $O(1)$ colors and $O(1)$-time solution using $O(\sqrt{n})$ colors.   
In this paper, we optimize step \emph{(iii)} by developing a $O(1)$-time and $O(1)$-color algorithm, forgoing the use of both binary and square root tournament.

The other main challenge of this work has been minimizing the computational SEC, i.e. the circular area in which robots operate.
Note that, to minimize the computational SEC, robots have to always stay within the smallest enclosing circle of the initial configuration.
The solutions used to solve \emph{(i)} {\sc Complete Visibility} \cite{SharmaVT17} and \emph{(ii)} {\sc Circle Formation} guarantee to minimize the computational SEC.
Thus, to minimize the computational SEC throughout the whole {\UCF} solution, our solution for step \emph{(iii)} must guarantee robots operate within the initial circle.
Note that the $O(\log n)$-time {\sc Uniform Transformation} solution provided in \cite{FelettiOPODIS23} makes robots move out from the initial circle $Cir$, before having them positioned uniformly on it.
Indeed, requiring robots to move internally to $Cir$ complicated the investigation for a $O(1)$-time solution, since internal moving robots can obstruct the visibility of the robots on the perimeter of $Cir$. 
However, our {\sc Uniform Transformation} solution proposed in this paper guarantees robots always stay within $Cir$ (perimeter included), thus minimizing its computational SEC.

\subsection{Techniques} 
Let us briefly introduce our algorithmic strategies.
Our algorithm for step \emph{(iii)} receives in input a circle configuration, where all $n$ robots lay on distinct points of a circle, $Cir$. 
The target $n$-gon will be inscribed in $Cir$.
Identifying the type of symmetry of the configuration turns out to be necessary to address the input robot disposition.
In particular, we classify the configuration into three categories: regular, biangular, or periodic. 
The regular configuration immediately solves {\UCF}.  
A biangular configuration can be converted into a regular configuration through a similar approach to the strategy introduced in \cite{DieudonneP07} in order to minimize the computational SEC.
The most challenging is the periodic configuration, for which we developed a sequence of multi-step procedures to form a regular polygon by the robots.
At first, we divide $Cir$ into $k\geq 2$ uniform sectors (depending on periodicities) with each arc containing $n/k$ vertices of the target $n$-gon to be formed. 
The algorithm then works in parallel in each of the $k$ uniform sectors. 
Within each sector, two robots are elected as \emph{guards} to fix the chirality of a structure called \emph{odd-block}.
Let $\chord$ be the chord joining the left guard with the right guard.
The other robots on the sector arc are now moved to the chord $\chord$. 
One robot is elected as the \emph{median} robot and reaches the midpoint of the arc cut by $\chord$.
A (inner) circle $\SC$ is now drawn such that it passes through the median robot and such that the chord becomes its tangent. 
All the robots on the chord now move to be positioned on the perimeter of $\SC$. 

We now use a procedure to provide a rank to the robots on $\SC$ so that a robot with rank $j$ moves to the target $n$-gon vertex $j$.
Indeed, \emph{rank encoding} is an important task involved in {\UCF} and it belongs to that part of literature that attempted to provide efficient strategies to encode \emph{non-constant} data (e.g. swarm cardinality, robot ranks, configuration properties, etc.) through the swarm configuration, in a sort of \emph{stigmergic behavior} \cite{BoseAKS19,BramasLT23,FelettiOPODIS23,LunaFSVY20a}.
In \cite{BramasLT23}, \etal{Bramas} develop the \emph{level-slicing technique} for solving the fault-tolerant {\sc Gathering} problem using the mutual distance of another robot. 
In particular, robots operate according to their \emph{level} previously encoded in their distance from another robot.
In \cite{FelettiOPODIS23}, authors use the \emph{pairing technique} to solve {\UCF} in that pairs of robots exploit their mutual distance to encode their rank, guaranteeing that half of the robots can reach the target position in a binary tournament fashion.
To achieve a parallel fashion, one can suggest making robots encode their rank through their mutual distance from a selected robot.
However, this general strategy is not sufficient to guarantee robots can reach such a mutual distance in $O(1)$ time without colliding or creating dangerous collinearities (in fact, robots need to count themselves to compute their rank).

To cope with these complications, we split the inner circle $\SC$ into slice units (equal to the minimum angular distance between any two robots) and make each robot encode its rank by properly positioning \emph{within} its slice, thus ensuring complete visibility during collision-free motions in a parallel fashion.
Besides the closeness with the level-slicing technique \cite{BramasLT23}, our rank encoding presents some peculiar differences: robots operate on a circle instead of a line, the slices are fixed for the whole procedure, and they have the same length instead of exponential lengths. 
Such differences have been necessary to cope with opaque and collision-intolerant robots and provide a robust rank encoding that retains effectiveness in complete asynchrony so that a robot must be able to recompute its rank $j$ even when all the other robots are moving.
Specifically, let $\diam$ be the diameter of $\SC$ passing through the median robot.
The robots on the inner circle are equally partitioned into two groups so that one group is positioned on the left half-perimeter of $\SC$ and the other group on its right half-perimeter. 
The robots in the left group then move along the perimeter to encode their rank using the angular distance with a fixed robot.
Then, the right group can obtain its rank using the left group as a reference.
Then, the robots of each group (first the right one, then the left one) migrate on $\diam$ on their projections, then they recompute their rank and reach their target vertices on the $Cir$-arc belonging to their uniform sector.
After each uniform sector completes the algorithm, the $n$ robots are equally distributed on $Cir$, thus solving {\UCF}. 
Care should be taken in synchronizing all these steps to ensure robots never collide even in {\asynch}. 
We show that all the steps can be executed correctly and safely (avoiding collisions) in $O(1)$ epochs.

\subsection{Related Works}
{\UCF} has been investigated in the luminous-opaque robot model by \etal{Feletti} \cite{feletti2018,FelettiOPODIS23,feletti2023journal,feletti2023} and Pattanayak and Sharma \cite{PattanayakS2024}. 
\etal{Feletti} \cite{feletti2018,FelettiOPODIS23,feletti2023journal,feletti2023} exhibited a $O(1)$-time algorithm in {\fsynch} and {\ssynch}, and a $O(\log n)$-time algorithm under {\asynch}, both using $O(1)$ colors and avoiding collisions.  
Pattanayak and Sharma \cite{PattanayakS2024} showed that {\UCF} can be solved in $O(1)$ time with $O(\sqrt{n})$ colors, or in $O(\log\log n)$ time with $O(1)$ colors in {\asynch}. 

Previous works have studied {\UCF} in the classic oblivious-transparent robot model.
The state-of-the-art result by \etal{Flocchini} \cite{FlocchiniPSV17} shows that {\UCF} can be solved under {\asynch} starting from any arbitrary configuration of $n> 5$ robots initially being on distinct locations, without any additional assumption (i.e., chirality, rigidity, etc.). 
However, their paper focuses on the computability of {\UCF} and not on its complexity in terms of execution time.
Prior to \cite{FlocchiniPSV17}, other different assumptions \cite{ChatzigiannakisMN04,DefagoK02,DefagoS08,Dieudonne2006,DieudonneP07,DieudonneP08,Katreniak05} had been considered in the study of {\UCF}. An excellent discussion can be found in \cite{Viglietta19}.
The case of $n\leq 5$ can be solved using ad-hoc algorithms.  
The special case of the {\sc Square Formation} has been solved in \cite{MaminoV16}. 
Recently, {\UCF} has been studied with additional objectives or features: minimizing the maximum distance traveled by a robot \cite{Bhagat2018} and considering transparent/opaque fat robots (of diameter 1) \cite{Datta2013,MondalC18,MondalC20}. 

{\sc Geometric} or {\sc Arbitrary Pattern Formation} is the macro-problem which encompasses {\UCF}.
Due to its importance, this problem has been heavily studied under different models (oblivious and luminous) \cite{Flocchini2012,Flocchini19,VaidyanathanST22}. 
Starting from arbitrary configurations of robots initially on distinct points on a plane, robots must autonomously arrange themselves to form the target pattern given to robots as input. 
\etal{Vaidyanathan} \cite{VaidyanathanST22} showed that {\sc Arbitrary Pattern Formation} can be solved in $O(\log n)$ epochs using $O(1)$ colors under {\asynch}. According to that result, {\UCF} can also be solved in $O(\log n)$ epochs using $O(1)$ colors, provided that the target $n$-gon positions are given to robots as input a priori. 
However, since the target circle is not known in advance in this paper, the strategies in \cite{VaidyanathanST22} are not compliant with our assumptions.

\subsection{Paper Organization}
We introduce the model and some preliminaries in \Cref{section:model}. 
We present a high-level overview of the three components of our algorithm in \Cref{section:highlevel}. 
The first two components are described in \Cref{section:componentI,section:componentII}. The discussion on the second component focuses on minimizing the computational SEC.
Then, we describe the third component in \Cref{section:componentIII}, which develops novel ideas allowing us to achieve complete parallelism using $O(1)$ colors. \Cref{section:alltogether} describes how the techniques developed in the three components collectively solve {\ucf} and prove its time- and color-optimization properties as well as computational SEC minimization property. 
Finally, we conclude with a short discussion in \Cref{section:conclusion}.  
The pseudo-code and correctness proofs are provided in \Cref{appendix:preudocode,appendix:correctness_proofs}: complexity analysis directly follows from the correctness proofs.

\section{Model and Preliminaries}\label{section:model}
\noindent{\bf Robots.}
We consider a swarm of $n$ dimensionless robots (agents) $\cR=\{r_1,r_2,\dots,r_{n}\}$ that can move in an Euclidean plane $\reals$. 
They enjoy the classical features of these models: they are {\em autonomous}, {\em anonymous}, {\em indistinguishable}, and \emph{homogeneous}.
They are completely disoriented (no global agreement on coordinate systems, chirality, unit distance, or origin).
We consider \emph{collision-intolerant} robots, i.e. robots must avoid collisions at any time.
Two robots $r_i,r_j$ are \emph{mutual visible} to each other iff there does not exist a third robot $r_k$ in the line segment joining $r_i$ and $r_j$, such that $r_k\notin\{r_i,r_j\}$. 
Clearly, mutual visibility is reflexive and symmetric.
The cardinality of the swarm $n$ is not a given parameter of the algorithm robots must execute.
We will often use $r_i$ to denote the robot as well as its position.

\vspace{2mm}
\noindent{\bf Lights.} Each robot is embedded with a light that can assume one color at a time from a fixed $O(1)$-size palette.

\vspace{2mm}
\noindent{\bf Look-Compute-Move.}
By default, a robot $r_i$ is inactive.
A robot can be activated by an adversarial scheduler.
When activated, it performs a ``Look-Compute-Move'' cycle as described below.
After that, the robot becomes inactive until its next activation.
Let us describe each LCM phase:

\begin{itemize}
	\item  {\em Look:} $r_i$ takes an instantaneous snapshot of the current configuration, i.e. it takes the positions and the colors of all the robots \emph{visible} to it.
	Robot positions are taken according to the coordinate system of $r_i$.
	
	\item 
	{\em Compute:}
	Robot~$r_i$ performs a deterministic algorithm with infinite precision using the taken snapshot as the sole input.
	The algorithm produces in output a (possibly) new position and light color for $r_i$.
	
	\item 
	{\em Move:} 
	Robot~$r_i$ updates its light to the new color and moves to its new position along a straight trajectory. 
	The move is rigid and hence $r_i$ always reaches its calculated position.   
\end{itemize}

\vspace{1mm}
\noindent{\bf Robot Activation and Synchronization.}
Generally, three schedulers are considered in the literature.
Under the synchronous schedulers, time is divided in atomic \emph{rounds}: at each round, a subset of robots is activated and they perform each phase of one LCM cycle in perfect synchrony.
Specifically, under the fully synchronous scheduler ({\fsynch}), all robots are activated at each round, while under the semi-synchronous scheduler ({\ssynch}), an arbitrary non-empty subset of robots is activated at each round.
Under the asynchronous scheduler ({\asynch}), robots act without any activation/synchronization assumption except for the fact that the Look phase is always instantaneous. The duration of each cycle is finite but unpredictable.
Indeed, an arbitrary amount of time may elapse between Look/Compute and Compute/Move.
Both in {\ssynch} and in {\asynch}, we assume the \emph{fairness assumption} which states that each robot is activated infinitely often.
An activated robot cannot distinguish between active and inactive robots from the snapshot.

\vspace{2mm}
\noindent{\bf Runtime.}
In {\fsynch}, time is measured in rounds of LCM cycles.
However, since in {\ssynch} and {\asynch} a robot could stay inactive for an indeterminate time, the concept of \emph{epoch} \cite{Cord-Landwehr11} has been introduced to have a proper measure to analyze algorithm runtimes.
We will use the term ``time'' generically to mean rounds for the {\fsynch} and epochs for the {\ssynch} and {\asynch} schedulers.

\vspace{2mm}
\noindent{\bf Configurations.}
Let the \emph{initial} configuration be denoted as \conf{init}. In \conf{init}, the $n$ robots in $\cR$ are positioned on $n$ distinct points on the plane, all with the same light color (w.l.o.g. \colr{off}).
We denote with \conf{convex} a configuration where the $n$ robots lay on the vertices of a convex $n$-gon.
If this $n$-gon is inscribed into a circle, we denote this configuration as \conf{circle}.
If this $n$-gon is regular, we denote the configuration as \conf{regular}.

\begin{definition}[\ucf]\label{definition:ucf}
Given any \conf{init} with $n$ luminous-opaque robots on the Euclidean plane, the \ucf\ (\UCF) problem asks the robots to autonomously reposition themselves to reach \conf{regular} without colliding, and terminate.
\end{definition}

\vspace{1mm}
\noindent\textbf{Terminology.} 
Given two circles $\Omega_1$, $\Omega_2$ on the Euclidean plane, we say that $\Omega_1$ \emph{is enclosed in} $\Omega_2$, denoted as $\Omega_1 \enclosed \Omega_2$, if all the points of $\Omega_1$ belong to the enclosed region delimited by $\Omega_2$ (perimeter included).
Indeed, $\enclosed$ is reflexive and transitive.
Given two points $A,B$, $\overline{AB}$ denotes the line segment connecting such points.
If $A$ and $B$ lay on a circle, $\wideparen{AB}$ denotes the smaller arc between the points $A$ and $B$. 
If points $A$ and $B$ are antipodal (i.e., diametrically opposite), we identify the arc by $\wideparen{ACB}$ where $C$ is a point on the circle that is contained in the arc $\wideparen{AB}$.
We say that an arc is \emph{chiral} if the robots on it are arranged in an asymmetrical pattern.
Given a point $D$ and a segment $\overline{EF}$, we define the \emph{mirror projection} (or, simply, \emph{projection}) of $D$ on $\overline{EF}$ the point $D'\in\overline{EF}$ (if it exists) such that $\overline{DD'}\perp \overline{EF}$.
The \emph{smallest enclosing circle} ($SEC$) of a given configuration $\config$, denoted as $SEC(\config)$, is the (unique) smallest circle containing all the robot positions in $\config$.
We denote with $Cir$ the circle $SEC(\config_{convex})$. 
$Cir$ can be divided into $n$ arcs of equal length. 
We call these arcs \emph{uniform segments} and the corresponding endpoints are \emph{uniform positions}. 
\Cref{tab:terminology} summarizes the essential notations here or later introduced and mainly used throughout the whole paper.

\begin{table}[H]
\centering
	{\small
	\renewcommand{\arraystretch}{1.1}
		\begin{tabular}{|c|l|}\hline
                $\Omega_1 \enclosed \Omega_2$ & Circle $\Omega_1$ is enclosed in circle $\Omega_2$\\\hline
			$\overline{AB}$ & The line segment joining points $A$ and $B$\\\hline
			$\wideparen{AB}$ & The smallest arc joining two points on the circle\\\hline
			$SEC(\config)$ & SEC of the configuration $\config$ \\\hline
			$Cir$ & $SEC(\config_{convex})$ \\\hline
			$\alpha, \beta, \mu$ & Angular Sequences \\\hline
			\BDCP & \bdcp \\\hline
			$\Phi$ & The set of robots with the smallest angular sequence \\\hline
                $\Upsilon_0,\dots, \Upsilon_{k-1}$ & Uniform sectors \\\hline
			$B_i, B_{i+1}$ & Boundaries of $\Upsilon_i$\\\hline
			$q$ & The total number of robots inside each uniform sector\\\hline
                $U_1,\dots,U_q$ & Uniform positions inside a block\\\hline
			$\chord$ & The \emph{block chord} joining the left guard with the right guard\\\hline
                $\barc$ & The Block Arc \\\hline
			$\SC$ & The small circle inside an odd block \\\hline
			$\diam$ & The median diameter of $\SC$ \\\hline
			$m$ & Number of \colr{west} (and so \colr{east}) robots on $\SC$ \\\hline
			$\delta$ & The small angle for \proc{Slice} \\\hline
		\end{tabular}
	}
\caption{Terminology.}
\label{tab:terminology}
\end{table}

\ifthenelse{\boolean{usefigures}}{
\begin{figure}[h] 
  	\centering
   \resizebox{\textwidth}{!}{
	\subfloat[\conf{regular}]{
            \centering
            \begin{tikzpicture}[scale=0.5, transform shape, font = {\LARGE}]
                \def\r{\ray*0.8}
                \def\n{12}
                \def\a{360/\n}
                \draw[draw=none] (0,0) circle (0.95*\ray);
                \draw [very thick] (0,0) circle (\r);
                \foreach \i in {1,..., \n} 
                {
                   \draw [thin, dashed] (0,0) -- (\i*\a:\r);
                   \rnode{\coldefault}{below right}{}{\i*\a:\r};
                    \node (a) at (\i*\a + \a*0.5:\r*0.7) {$a$};

                }
        
            \end{tikzpicture}
       \label{subfig:regular}}
	\hfill
	\subfloat[\conf{biangular}]{
            \centering
            \begin{tikzpicture}[scale=0.5, transform shape, font = {\LARGE}]
                \def\r{\ray*0.8}
                \def\n{360/12}
                \def\a{60}
                \def\b{20}
                \draw[draw=none] (0,0) circle (0.95*\ray);
                \draw [very thick] (0,0) circle (\r);
                \foreach \i in {1,...,10} 
                {  
                   \draw [thin, dashed] (0,0) -- (\i*\a:\r);
                   \draw [thin, dashed] (0,0) -- (\i*\a+\b:\r);
                   \rnode{\coldefault}{below right}{}{\i*\a:\r};
                    \rnode{\coldefault!30}{below right}{}{\i*\a+\b:\r};
                    \node (a) at ({\i*\a - (\a-\b)*0.5}:\r*0.7) {$a$};
                    \node (b) at (\i*\a + \b*0.5:\r*0.7) {$b$};
                    
                }
        
            \end{tikzpicture}
       \label{subfig:conf_biangular}}
	\hfill
	\subfloat[\conf{uniperiodic}]{
            \centering
            \begin{tikzpicture}[scale=0.5, transform shape, font = {\LARGE}]
                \def\r{\ray*0.8}
                \def\a{120}
                \def\b{30}
                \def\c{20}
                \draw[draw=none] (0,0) circle (0.95*\ray);
                \draw [very thick] (0,0) circle (\r);
                \foreach \i in {1,2,3} 
                {  
                   \draw [thin, dashed] (0,0) -- (\i*\a:\r);
                   \draw [thin, dashed] (0,0) -- (\i*\a+\b:\r);
                    \draw [thin, dashed] (0,0) -- (\i*\a+\b+\c:\r);

                   \rnode{\coldefault}{below right}{}{\i*\a:\r};
                    \rnode{\coldefault!50}{below right}{}{\i*\a+\b:\r};
                    \rnode{\coldefault!20}{below right}{}{\i*\a+\b+\c:\r};
                    \node (a) at ({\i*\a + \a - (\a-\b-\c)*0.5}:\r*0.7) {$c$};
                    \node (b) at (\i*\a + \b*0.5:\r*0.7) {$a$};
                    \node (c) at (\i*\a + \b + \c*0.5:\r*0.7) {$b$};  
                }
        
            \end{tikzpicture}
       \label{subfig:uniperiodic}}
		\hfill
	\subfloat[\conf{biperiodic}]{
            \begin{tikzpicture}[scale=0.5, transform shape, font = {\LARGE}]
                \def\r{\ray*0.8}
                \def\a{20}
                \def\b{50}
                \def\c{120}
                \def\d{150}
                \draw[draw=none] (0,0) circle (0.95*\ray);
                \draw [very thick] (0,0) circle (\r);
                \foreach \i in {1,-1} 
                {  
                   \draw [thin, dashed] (0,0) -- (90+\i*\a:\r);
                   \draw [thin, dashed] (0,0) -- (90+\i*\b:\r);
                    \draw [thin, dashed] (0,0) -- (90+\i*\c:\r);
                    \draw [thin, dashed] (0,0) -- (90+\i*\d:\r);
        
                   \rnode{\coldefault}{below right}{}{90+\i*\a:\r};
                   \rnode{\coldefault!60}{below right}{}{90+\i*\b:\r};
                   \rnode{\coldefault!40}{below right}{}{90+\i*\c:\r};
                    \rnode{\coldefault!20}{below right}{}{90+\i*\d:\r};

                    \node (b) at ({90+\i*(\b-(\b-\a)*0.5)}:\r*0.7) {$b$};
                    \node (c) at ({90+\i*(\c-(\c-\b)*0.5)}:\r*0.7) {$c$};
                    \node (d) at ({90+\i*(\d-(\d-\c)*0.5)}:\r*0.7) {$d$};)

                }
                \node (a) at (90:\r*0.7) {$a$};)
                \node (e) at (270:\r*0.7) {$e$};)

            \end{tikzpicture}
    \label{subfig:biperiodic_norobot}}
    		\hfill
	\subfloat[\conf{biperiodic}]{
            \begin{tikzpicture}[scale=0.5, transform shape, font = {\LARGE}]
                \def\r{\ray*0.8}
                \def\a{20}
                \def\b{50}
                \def\c{120}
                \def\d{150}
                \draw[draw=none] (0,0) circle (0.95*\ray);
                \draw [very thick] (0,0) circle (\r);
                \foreach \i in {1,-1} 
                {  
                   \draw [thin, dashed] (0,0) -- (90+\i*\a:\r);
                   \draw [thin, dashed] (0,0) -- (90+\i*\b:\r);
                    \draw [thin, dashed] (0,0) -- (90+\i*\c:\r);
                    \draw [thin, dashed] (0,0) -- (90+\i*\d:\r);
        
                   \rnode{\coldefault!80}{below right}{}{90+\i*\a:\r};
                   \rnode{\coldefault!60}{below right}{}{90+\i*\b:\r};
                   \rnode{\coldefault!40}{below right}{}{90+\i*\c:\r};
                    \rnode{\coldefault!20}{below right}{}{90+\i*\d:\r};
                    
                    \node (a) at ({90+\i*(\a)*0.5)}:\r*0.7) {$a$};
                    \node (b) at ({90+\i*(\b-(\b-\a)*0.5)}:\r*0.7) {$b$};
                    \node (c) at ({90+\i*(\c-(\c-\b)*0.5)}:\r*0.7) {$c$};
                    \node (d) at ({90+\i*(\d-(\d-\c)*0.5)}:\r*0.7) {$d$};)
                    \node (e) at ({270+\i*((180-\d)*0.5)}:\r*0.7) {$e$};)

                }
                 \draw [thin, dashed] (90:\r) -- (270:\r);
                \rnode{\coldefault}{below right}{}{90:\r};
                \rnode{\coldefault!10}{below right}{}{270:\r};
            
            \end{tikzpicture}
    \label{subfig:biperiodic_endpointrobot}}
}
\caption{Different types of \conf{circle}: (a) \conf{regular}, (b) \conf{biangular}, (c) \conf{uniperiodic}, (d) \conf{biperiodic} without robots on the axis of symmetry, (e) \conf{biperiodic} with robots on the axis of symmetry.}
\label{fig:configuration_symmetries}
\end{figure}
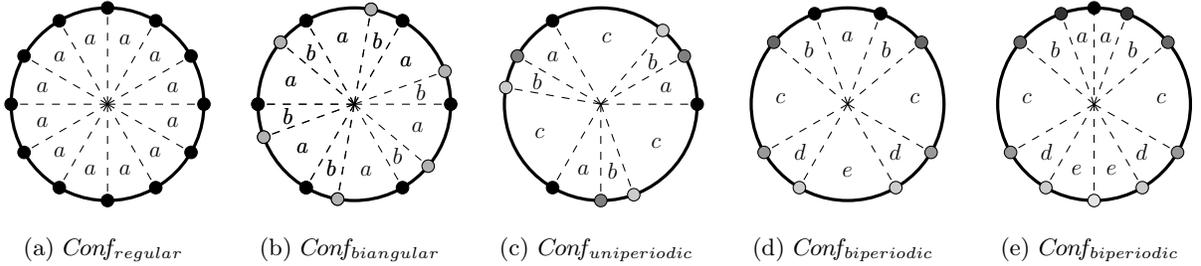}{}

\subsection{Classes of Circle Configurations}
\label{subsection:circleclasses}
Consider \conf{circle} where $n$ robots colored \colr{onSEC} are located on $Cir$. 
Let $r_0, \dots, r_{n-1}$ be their positions. 
We define the \emph{angular sequence} for a robot $r_j$ as the sequence of angles in the clockwise order represented by $\alpha(j) = a_ja_{j+1}\dots a_{j-1}$, where $a_i = \angle r_iOr_{i+1}$, $O$ is the center of $Cir$, and the indices are considered modulo $n$.
Similarly, we define $\beta(j) = b_jb_{j-1} \dots b_{j+1}$, where $b_i = \angle r_{i}Or_{i-1}$ considered in the counterclockwise direction. 
Let $\mu(j) = \min\{\alpha(j),\beta(j)\}$ be the lexicographical minimum sequence for $r_j$. 
Let $\hat{\mu} = \min\{\mu(0), \dots, \mu(n-1)\}$.
Note that, there may be multiple robots corresponding to the \emph{smallest angular sequence} $\hat{\mu}$. 
We define $\Phi = \{r_i: \mu(i) = \hat{\mu}\}$.

A configuration is \emph{reflective} (i.e., a mirror symmetry) if there exists a pair of robots $r_p$ and $r_{p'}$ in $\Phi$ such that $\alpha(p) = \beta(p')$.
Now, we distinguish between two cases, according to the cardinality of $\Phi$:

\subsubsection{Case $|\Phi| = n$ -- Regular and Biangular Configurations} 
In this case, all the robots correspond to the smallest angular sequence. 
There are two types of configurations. 
\begin{enumerate}
    \item\textbf{Regular Configuration:} In \conf{regular}, we have $\alpha(0) = \alpha(i)$, for all $i = 0,\dots,n-1$ (\Cref{subfig:regular}).
     \item\textbf{Biangular Configuration:} In \conf{biangular}, we have $\alpha(0)=\alpha(2i) = (a_0a_1)^{\frac{n}{2}}$ where $a_0\neq a_1$ and $n\geq 2$, for any $i = 0,\dots,\frac{n}{2}-1$ (\Cref{subfig:conf_biangular}).
\end{enumerate}
\begin{observation}\label{obs:uniformsymmetry}
If $|\Phi|=n$ and $\alpha(i) =  \beta(i)$ for some $0\leq i \leq n-1$, then the configuration is \conf{regular}.
\end{observation}
In \conf{biangular}, an axis of symmetry always passes through the midpoint between any pair of adjacent robots on $Cir$.
Also, no robot lies on this axis of symmetry (\Cref{subfig:conf_biangular}). 

\subsubsection{Case $1\leq |\Phi|< n$ -- Periodic Configurations} 
In a periodic configuration \conf{periodic}, there are $|\Phi|$ groups of robots creating some form of symmetry (reflective or rotational). 
\conf{periodic} can be classified into two subclasses as follows. 
If \conf{periodic} is \emph{not reflective}, we call it \emph{uniperiodic} and denote as \conf{uniperiodic} (\Cref{subfig:uniperiodic}). 
The special case $|\Phi|=1$ gives us an \emph{asymmetric} configuration, denoted as \conf{asymmetric}.
If the configuration is \emph{reflective}, we call it \emph{biperiodic} and denote as \conf{biperiodic}.
Given a diameter $d$ lying on an axis of symmetry, there are three cases: \emph{(i)} no robot lies on the endpoints of $d$ (\Cref{subfig:biperiodic_norobot}), \emph{(ii)} each endpoint of $d$ is covered by one robot (\Cref{subfig:biperiodic_endpointrobot}), and \emph{(iii)} a robot lies on one endpoint of $d$ whereas the other endpoint is robot-free.

 \begin{algorithm}
	\caption{\bdcp~\cite{SharmaVT17}}
	\label{alg:bdcp}
	\footnotesize	
        \SetKwInput{Input}{Input Assumptions}        
        \SetKw{Strategy}{Strategy:}

        \Input{
            \begin{enumerate}
                \item$\Omega$, a $k$-point curve, with $k\in\{2,3\}$;
                \item $2k$ \emph{beacons} lay on $\Omega$ ($k$ \emph{left} beacons and $k$ \emph{right} beacons);
                \item$m$ \emph{waiting} robots $r_1,\dots, r_m$ s.t.
                    \begin{enumerate}
                        \item they are external to $\Omega$ belonging to the same half-plane (if $k=2$) or the convex side (if $k=3$);
                        \item they can see all the $2k$ original beacons;
                        \item they have to reach $\Omega$
                        \begin{itemize}
                            \item{} in some points between the original \emph{left} beacons and the \emph{right} beacons;
                            \item{} traveling along disjoint paths $\pi_1, \dots, \pi_m$;
                            \item{} so that their paths intersect $\Omega$ only once;    
                            \item{} so that if a waiting robot $r_i$ sees at least $k$ beacons, then it can compute its path $\pi_i$;
                        \end{itemize}
                    \item a waiting robot becomes a new beacon once it reaches $\Omega$.
                    \end{enumerate}
            \end{enumerate}
            }
		\KwResult{All the waiting robots reach $\Omega$ without colliding in $O(\log{k})$ epochs.}	
\end{algorithm}

\subsection{Beacon Directed Curve Positioning}
{\bdcp} (\BDCP) is a procedure developed by \etal{Sharma} in \cite{SharmaVT17} that uses some robots as \emph{beacons} to define a $k$-point curve, and then guide other \emph{waiting} robots (external to the curve all on one side) from their initial positions to their final positions on the curve in {\asynch}. 
A $k$-point curve $\Omega$ is a curve on $\reals$ where $k$ is the minimal number of points necessary to define $\Omega$.
Here, we consider only straight lines and circular arcs (so $k\in\{2,3\}$).
Beacons and waiting robots have different colors, w.l.o.g. \colr{beacon} and \colr{waiting}.
{\BDCP} starts placing $2k$ \colr{beacon} robots on $\Omega$ (divided in $k$ left and $k$ right beacons) to define the curve. 
Specifically, the beacons are positioned at their final positions on the curve. 
After that, all the $m$ \colr{waiting} robots move to be positioned on $\Omega$ between left and right beacons (\Cref{fig:beacon}). 
Once a robot reaches $\Omega$, it changes its color to \colr{beacon} becoming a new beacon.
It was shown in \cite{SharmaVT17} that if the first epoch started with $m$ \colr{waiting} robots wanting to move to the curve (i.e., not already in transit to the curve) and an epoch $e\geq 1$ starts with $v\geq 2k$ beacons, then epoch $e+1$ starts with $\min\{m+2k,3v/2\}$ beacons, hence relocating all the robots to the curve in $O(\log k)$ epochs, using $O(1)$ colors.
In our framework, we borrow the idea of \etal{Sharma} \cite{SharmaVT17} and use {\BDCP} to relocate robots on line segments ($k=2$) and circular arcs ($k=3$) using $2k$ beacons in $O(\log{k}) = O(1)$ epochs in {\asynch}.
\Cref{alg:bdcp} summarizes the input assumptions and the resulted output of the {\BDCP} procedure.

\ifthenelse{\boolean{usefigures}}{
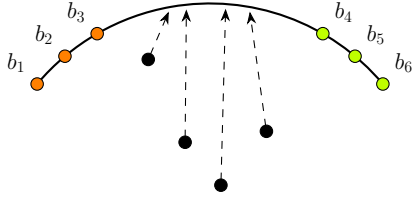
\begin{figure}
\centering
\begin{tikzpicture}[scale=0.5, transform shape, font = {\LARGE}]
    \def\radius{6cm}

    \draw[thick] (40:\radius) arc[start angle=40, end angle=140, radius=\radius];
    
    \node (A1) at (40:\radius) {};
    \node (A2) at (50:\radius) {};
    \node (A3) at (60:\radius) {};
    \node (A4) at (120:\radius) {};
    \node (A5) at (130:\radius) {};
    \node (A6) at (140:\radius) {};
    \node (P1) at (60:0.5*\radius) {};
    \node (P2) at (110:0.8*\radius) {};
    \node (P3) at (106:0.4*\radius) {};
    \node (P4) at (76:0.2*\radius) {};
    \node (D1) at (80:\radius) {};
    \node (D2) at (100:\radius) {};
    \node (D3) at (96:\radius) {};
    \node (D4) at (86:\radius) {};

    \foreach \point/\i in {A1/6, A2/5, A3/4}
    	\rnode{lime}{above right}{$b_\i$}{\point};
    \foreach \point/\i in {A4/3, A5/2, A6/1}
    	\rnode{orange}{above left}{$b_\i$}{\point};
    \foreach \point in {P1, P2, P3, P4}
    	\rnode{\coldefault}{below right}{}{\point};
    
    \draw[dashed, -{Stealth}] (P1) -- (D1);
    \draw[dashed, -{Stealth}] (P2) -- (D2);
    \draw[dashed, -{Stealth}] (P3) -- (D3);
    \draw[dashed, -{Stealth}] (P4) -- (D4);
\end{tikzpicture}
\caption{\BDCP ~on a circular arc, where $b_1,b_2,b_3$ ($b_4,b_5,b_6$, resp.) are the \emph{left} (\emph{right}, resp.) beacons.}
\label{fig:beacon}
\end{figure}}{}

\section{Overview of the Algorithm}\label{section:highlevel}
In this section, we provide an overview of our $O(1)$-time $O(1)$-color algorithm for {\UCF} under {\asynch}, avoiding collisions. 
Starting from an arbitrary configuration \conf{init} where $n$ \colr{off}-colored robots lay on distinct points of the plane, the proposed algorithm includes three major components: 
\begin{itemize}
    \item{} (I) circle formation;
    \item{} (II) transforming biangular configurations to regular configurations;
    \item{} (III) transforming periodic configurations to regular configurations.
\end{itemize}
In the next Sections, we will provide the details on each of these components.
The proposed algorithm runs Component I and, subsequently, either Component II or Component III. 
The diagram in \Cref{fig:transition} depicts the configuration transitions induced by the execution of each step of our algorithm.
\ifthenelse{\boolean{usefigures}}{
\begin{figure*}[!t]
    \centering
\resizebox{\textwidth}{!}{
    \begin{tikzpicture}
  \tikzstyle{confbox} = [rectangle, draw, text centered, minimum height=2em, rounded corners]
  \def\height{4.5}
  \def\length{23}
  
  \node [confbox] (init) at (-0.2*\length,0) {\LARGE \conf{init}};
  \node [confbox] (convex) at (-0.1*\length,\height) {\LARGE \conf{convex}};
  \node [confbox] (circle) at (0,0){\LARGE \conf{circle}};
  \node [confbox] (periodic) at (0.1*\length,\height) {\LARGE \conf{periodic}};
  \node [confbox] (usect) at (0.33*\length,\height) {\LARGE \conf{unisect}};
  \node [confbox] (oddblock) at (0.63*\length,\height) {\LARGE \conf{oddblock}};
  \node [confbox] (smallcircle) at (0.91*\length,\height) {\LARGE \conf{smallcircle}};
  \node [confbox] (biangular) at (0.5*\length,\height/3){\LARGE \conf{biangular}};
  \node [confbox] (regular) at (\length,0) {\LARGE \conf{regular}};

  \draw [-{Latex[length=3mm]}] (init) -- node[sloped,above] {\sc Complete Visibility} (convex);
  \draw [-{Latex[length=3mm]}] (convex) -- node[sloped,above] {  \sc Circle Formation} (circle);
  \draw [-{Latex[length=3mm]}] (circle) -- node[sloped,above] {  \sc } (periodic);
  \draw [-{Latex[length=3mm]}] (periodic) -- node[sloped,above] {  \sc Split} (usect);
  \draw [-{Latex[length=3mm]}] (usect) -- node[sloped,above] {  \sc Odd Block ($q \geq 12$)} (oddblock);
  \draw [-{Latex[length=3mm]}] (usect) -- node[sloped,above] {  {\sc Sequential Match} ($q < 12$)} (regular);
  \draw [-{Latex[length=3mm]}] (oddblock) -- node[sloped,above] {  \sc Small Circle} (smallcircle);
  \draw [-{Latex[length=3mm]}] (biangular) -- node[sloped,above] {  \sc Biangular} (regular);
  \draw [-{Latex[length=3mm]}] (init) -- (circle);
  \draw [-{Latex[length=3mm]}] (circle) -- (biangular);
  \draw [-{Latex[length=3mm]}] (smallcircle) -- node[sloped,above] {  \sc Slice} (regular);

  \draw [-{Latex[length=3mm]}] (init) -- node[sloped, below] {\mycircled{1}} (convex);
  \draw [-{Latex[length=3mm]}] (convex) -- node[sloped, below] {\mycircled{2}} (circle);
  \draw [-{Latex[length=3mm]}] (biangular) -- node[sloped, below] {\mycircled{3}} (regular);
  \draw [-{Latex[length=3mm]}] (periodic) -- node[sloped, below] {\mycircled{4}} (usect);
  \draw [-{Latex[length=3mm]}] (usect) -- node[sloped, below] {\mycircled{5}} (oddblock);
  \draw [-{Latex[length=3mm]}] (oddblock) -- node[sloped, below] {\mycircled{6}} (smallcircle);
  \draw [-{Latex[length=3mm]}] (smallcircle) -- node[sloped, below] {\mycircled{7}} (regular);
  \draw [-{Latex[length=3mm]}] (circle) -- (regular);
  \draw [-{Latex[length=3mm]}] (usect) -- node[sloped, below] {\mycircled{8}} (regular);
\end{tikzpicture}
}
    \caption{Transition diagram among configurations while solving {\UCF}.
    The arrows without numbering denote a transition with only color change (no robot moves). 
    The parameter $q$ is the number of robots in each uniform sector of \conf{unisect}.}  
    \label{fig:transition}
\end{figure*}
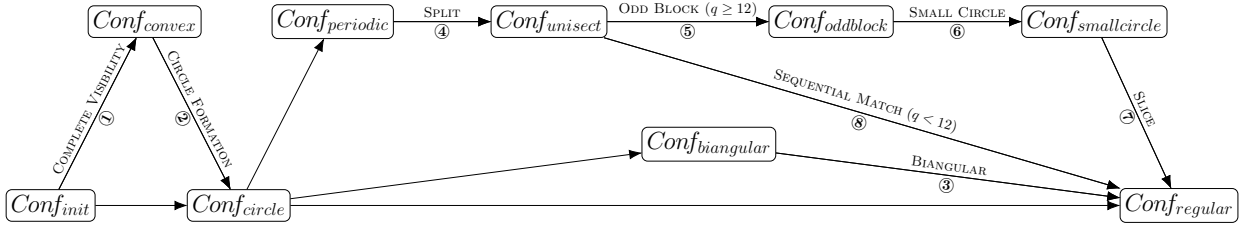
}{}

Starting from \conf{init}, robots execute \proc{Complete Visibility} \cite{SharmaVT17} to form a convex configuration \conf{convex}, ({\bf transition} \mycircled{1}).
Afterward, robots execute \proc{Circle Formation} to achieve \conf{circle} ({\bf transition} \mycircled{2}). 
Pattanayak and Sharma \cite{PattanayakS2024}  showed that \conf{init} can be transformed to \conf{circle} in $O(1)$ epochs using $O(1)$ colors under {\asynch}, avoiding collisions.

Once in \conf{circle}, if robots satisfy \conf{biangular}, they use \proc{Biangular} to arrange themselves to satisfy \conf{regular}, thus solving {\UCF} ({\bf transition} \mycircled{3}).  

Indeed, the most challenging case results when \conf{circle} satisfies \conf{periodic}. 
Our contribution in this paper will establish that \conf{periodic} can be transformed into \conf{regular} in $O(1)$ epochs under {\asynch}, avoiding collisions, using $O(1)$ colors ({\bf transitions} \mycircled{4}--\mycircled{8}).
Specifically, \conf{periodic} is first transformed into \conf{unisect} (a configuration containing uniform sectors as defined later), using \proc{Split} ({\bf transition} \mycircled{4}). 
Then \conf{unisect} is transformed into \conf{oddblock} (an odd block configuration defined later) using \proc{Odd Block} ({\bf transition} \mycircled{5}).
Then, \conf{oddblock} is transformed into \conf{smallcircle} using \proc{Small Circle} ({\bf transition} \mycircled{6}).
Finally \conf{smallcircle} is transformed into \conf{regular} using \proc{Slice} ({\bf transition} \mycircled{7}), thus solving {\UCF}. 
When \conf{unisect} has less than 12 robots in each uniform sector, we directly solve {\UCF} using \proc{Sequential Match} ({\bf transition} \mycircled{8}).
We will show that each of these procedures takes $O(1)$ time and colors.   
Light colors are also used to synchronize robots to decide on the subsequent configuration to achieve, avoiding ambiguous situations.
Moreover, we will show that our solution minimizes the computational SEC by proving that $SEC(\config)\enclosed SEC(\config_{init})$ for any intermediate configuration $\config$ achieved throughout the execution of the whole algorithm.
\sloppy

We now provide details of the algorithm. 
We give particular focus on Components II and III since Component I can be run using the techniques of \etal{Feletti} \cite{FelettiOPODIS23,feletti2023journal,feletti2023} and Pattanayak and Sharma \cite{PattanayakS2024}. Therefore, we describe Component I in brief in \Cref{section:componentI} and move on to describing techniques for Components II and III in \Cref{section:componentII,section:componentIII}, respectively. Notice that the techniques of \etal{Feletti} \cite{FelettiOPODIS23,feletti2023journal,feletti2023} and Pattanayak and Sharma \cite{PattanayakS2024} can be used in Component II but those techniques only optimize time and color complexities but not minimize the computational SEC. 
Our new idea both optimizes time/color complexities and minimizes the computational SEC.

\section{Component I - Circle Formation (Transitions \mycircled{1} and \mycircled{2})}
\label{section:componentI}
Given \conf{init}, the robots use the {\sc Complete Visibility} algorithm of \etal{Sharma} \cite{SharmaVT17} to reach \conf{convex} where robots occupy the vertices of a convex $n$-gon, thus guaranteeing complete visibility to the whole swarm.
Let $Cir$ be $SEC(\config_{convex})$. 
The {\sc Complete Visibility} algorithm in \cite{SharmaVT17} works using $O(1)$ time and colors, and guarantees that any robot in the swarm operates within $SEC(\config_{init})$ (thus minimizing the computational SEC).
W.l.o.g. we assume that each robot is colored as \colr{convex} in \conf{convex}. 

Starting from \conf{convex}, robots execute \proc{Circle Formation}. 
Let $r$ be a \colr{convex} robot.
If $r$ already lays on the perimeter of $Cir$, it sets its color as \colr{onSEC}. 
Note that at least 2 (if antipodal) or 3 robots (if not antipodal) on $Cir$ set their color as \colr{onSEC} without moving, while the remainder lies in the interior of $Cir$. 
Otherwise (i.e. if $r$ is internal to $Cir$), $r$ moves radially towards $Cir$ until reaching the perimeter of the circle.
After reaching $Cir$, $r$ assumes color \colr{onSEC}.
Once every robot turns into \colr{onSEC}, {\sc Circle Formation} is solved.   
Pattanayak and Sharma \cite{PattanayakS2024} showed that this process avoids collisions and robots correctly position themselves on $Cir$ in $O(1)$ epochs with $O(1)$ colors. 
\begin{lemma}[\textbf{Circle Formation} \cite{PattanayakS2024}]\label{lemma:circleformation}
Given \conf{init} with $n$ \colr{off}-colored robots on distinct points on a plane, robots reposition themselves in \conf{circle} using $O(1)$ epochs and $O(1)$ colors under {\asynch}, avoiding collisions, guaranteeing that robots operate within $SEC(\config_{init})$.
\end{lemma}

\section{Component II - Transforming Biangular Configurations into Regular Configurations (Transition \mycircled{3})}
\label{section:componentII}
Each \colr{onSEC} robot on \conf{circle} can check whether the current configuration satisfies either $\config_{regular}, \config_{biangular},$ or \conf{periodic}. 
Thus, each robot changes its color accordingly (without moving) to either \colr{regular}, \colr{biangular}, or \colr{periodic}, respectively.  
Colors are used to make robots agree on the procedure to be executed.
Indeed, \conf{regular} solves {\UCF}, so robots will no longer update their color or position.
\ifthenelse{\boolean{usefigures}}{
\begin{figure}[H]
    \centering  
    \resizebox{\textwidth}{!}{
    \subfloat[\conf{biangular}]{
        \begin{tikzpicture}[scale=0.5, transform shape, font = {\LARGE}]
        \def\r{3}
        \def\n{360/4}
        \def\a{15}
        \draw (0,0) circle (\r);
        \draw[draw=none] (0,0) circle (4.2);
        \foreach \i in {1,...,4}
        {
            \draw [thin, dashed] (0,0) -- (\i*\n - \a:\r);
            \draw [thin, dashed] (0,0) -- (\i*\n + \a:\r);
            \node (a) at ({\i*\n}:\r*0.7) {$a$};
            \node (b) at ({\i*\n + \n*0.5}:\r*0.7) {$b$};
            \rnode{\coldefault}{below right}{}{\i*\n - \a:\r};
            \rnode{\coldefault}{below right}{}{\i*\n + \a:\r};
        }
        \end{tikzpicture}
   \label{subfig:biangular}}\hfill
    \subfloat[Regular $n$-gon $P'$ \cite{DieudonneP07}]{
        \begin{tikzpicture}[scale=0.5, transform shape, font = {\LARGE}]
        \def\r{3}
        \def\y{3.67}
        \def\z{4.098} 
        \def\x{3.136} 
        \def\n{360/4}
        \def\a{15}
        \def\b{22.5}
        \draw[dashdotted] (0,0) circle (\r);
        \draw[draw=none] (0,0) circle (4.2);
        \foreach \i in {1,...,4}
        {
            \rnode{\coldefault}{below right}{}{\i*\n - \a:\r};
            \rnode{\coldefault}{below right}{}{\i*\n + \a:\r};
            \draw (\i*\n - \b:\x) -- (\i*\n + \b:\x);
            \draw (\i*\n + \b:\x) -- (\i*\n + \n - \b:\x);
        }
        \end{tikzpicture} 
   \label{subfig:SP}}\hfill
    \subfloat[Exogenous Polygon $EP$]{
        \begin{tikzpicture}[scale=0.5, transform shape, font = {\LARGE}]
        \def\r{3}
        \def\y{3.67}
        \def\n{360/4}
        \def\a{15}
        \draw[draw=none] (0,0) circle (4.2);
        \foreach \i in {1,...,4}
        {
            \rnode{\coldefault}{below right}{}{\i*\n - \a:\r};
            \rnode{\coldefault}{below right}{}{\i*\n + \a:\r};
            \draw (\i*\n - \a:\r) -- (\i*\n + \a:\r);
            \draw (\i*\n + \a:\r) -- (\i*\n + \n - \a:\r);
            \draw[dashed] (\i*\n :\y) -- (\i*\n + \n :\y);
        }
        \node (z) at (45:3.3) {$z$};
        \draw[<->] (\y+0.3, 0.3) -- (0.3,\y+0.3);
        \end{tikzpicture}
   \label{subfig:EP}}\hfill
    \subfloat[Regular $n$-gon $P''$]{
        \begin{tikzpicture}[scale=0.5, transform shape, font = {\LARGE}]
        \def\r{3}
        \def\y{3.67}
        \def\z{2.812}
        \def\n{360/4}
        \def\a{15}
        \draw[draw=none] (0,0) circle (4.2);
        \foreach \i in {1,...,4}
        {
            \rnode{\coldefault}{below right}{}{\i*\n - \a:\r};
            \rnode{\coldefault}{below right}{}{\i*\n + \a:\r};
            \draw (\i*\n - 22.5:\z) -- (\i*\n + 22.5:\z);
            \draw (\i*\n + 22.5:\z) -- (\i*\n + \n - 22.5:\z);
            \draw[dashed] (\i*\n :\y) -- (\i*\n + \n :\y);
        }
        \node (y) at (18:4) {$y$};
        \draw[<->] (\y+0.3, 0.3) -- (\y+0.3 - \y*0.29,\y*0.29+0.3);
        \end{tikzpicture}
   \label{subfig:inscribed}}}
    \caption{Arrangement of \conf{biangular} in a regular $n$-gon.}
    \label{fig:biangular}
\end{figure}
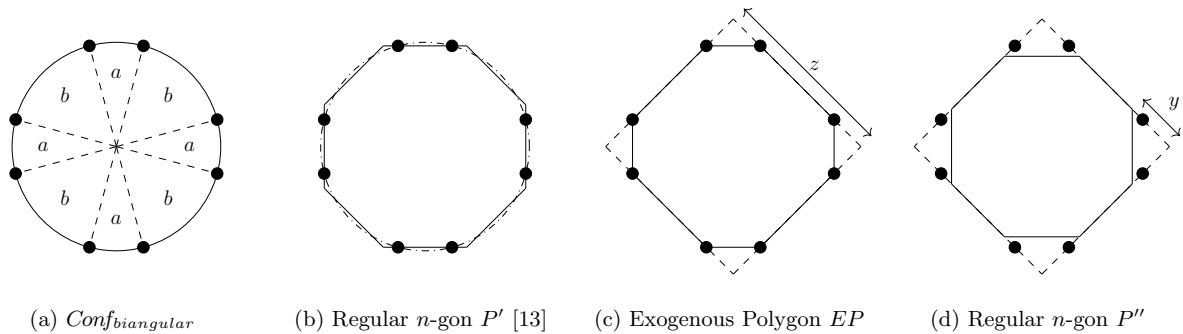
}{}
We discuss here \conf{biangular}.
\conf{periodic} is treated in Component III and discussed in the next section.
To solve this special configuration, we take inspiration from an idea of \etal{Dieudonné} \cite{DieudonneP07}, although we implement a new approach to transform \conf{biangular} into \conf{regular} which suits our objective.
In \conf{biangular}, $n$ is even and the configuration has $n/2$ axes of symmetry. 
Let $n=2k$, and let $P$ be the $n$-gon formed by the robots in  \conf{biangular} (i.e. each vertex of $P$ is covered by a robot).
The strategy developed by \etal{Dieudonné} \cite{DieudonneP07} was to spot the target regular $n$-gon $P'$ which encloses $P$, such that robots lay on alternative edges of $P'$ (\Cref{subfig:SP}), and make robots slide on the edges of $P'$ until they stop to the vertices of $P'$.
If on the one hand this strategy solves {\UCF} even under {\asynch}, on the other hand it leads the robots to move outside $SEC(\config_{biangular})$ to reach the vertices of $P'$. 
Seeing that our objective is also to minimize the computational SEC, we propose a target $n$-gon $P''$ that is inscribed inside $P$. 
Specifically, the inscribed polygon $P''$ has its vertices on the larger edge of $P$ (being biangular, $P$ has two different edge lengths). 
After detecting the target polygon $P''$, the robots are made to move along the larger edges of $P$ towards the vertices of $P''$, thus achieving \conf{regular}. 
We now show that, as in \cite{DieudonneP07}, our strategy works also considering asynchronous robots. 
Given $P$, we form a regular $k$-gon by extending the larger edges of $P$ (see \Cref{subfig:EP}).
We call this extended $k$-gon as the \emph{exogenous polygon} ($EP$ for short).
Notice that, as long as the robots move along the larger edges of $P$ such that the larger edge remains the larger edge, the $EP$ also remains the same.
For an $EP$, there exists a unique regular $2k$-gon inscribed in it (see \Cref{subfig:inscribed}).
This property guarantees that the target polygon $P''$ (i.e. the regular $2k$-gon inscribed in $EP$) is maintained even when some robots are in transit, and so it can be correctly computed even when robots no longer form $P$.
In particular, the location of the vertices of $P''$ is located at a distance $y$ from the vertex of the $EP$, where $z$ is the edge length of the $EP$. 
The value of $y$ is given by
\begin{align*}
    y = \frac{z}{2+ \sqrt{2- 2\cos\left(\frac{(k-2)\pi}{k}\right)}}
\end{align*}
After reaching a vertex of $P''$, the robots assume color \colr{regular}.
In summary, robots form \conf{regular} from \conf{biangular} in three epochs: \emph{(i)} robots change color to \colr{biangular}, \emph{(ii)} then they head to the vertices of $P''$, traveling along the edges of $EP$, and \emph{(iii)} lastly they color as \colr{regular}.

\begin{lemma}[\textbf{Biangular to Regular Configuration}]\label{lemma:biangular}
Given \conf{biangular} of $n$ \colr{biangular}-colored robots, the robots reposition to \conf{regular} in $O(1)$ epochs using $O(1)$ colors under {\asynch}, avoiding collisions, guaranteeing that robots always operate within $SEC(\config_{biangular})$. 
\end{lemma}

\section{Component III - Transforming Periodic Configurations into Regular Configurations (Transitions \mycircled{4}--\mycircled{8})}
\label{section:componentIII}
Before diving into the deep of our algorithm, we mention two strategic usages of robot lights to ensure a correct evolution of the algorithm under {\asynch}.
For the sake of conciseness, we will often omit the details of these two strategies while explaining the algorithm, assuming their rationale is understood.
\begin{itemize}
    \item{}\textbf{Elections under {\asynch}.} In some points of our algorithm, a group of robots is elected to perform a movement, in parallel. 
    Being under {\asynch}, we need to ensure that all robots can correctly elect themselves even if part of the elected group is already moving.
    For this reason, our algorithm requires \emph{all} the elected robots in the group to synchronize themselves setting their colors as \colr{pre\_<color>} before starting moving, where \colr{<color>} will be the color they have to set once stopped and reactivated.
    \item{}\textbf{Movements under {\asynch}.} In the remainder, our algorithm makes a robot color itself as \colr{to\_<color>} before starting moving, where \colr{<color>} will be the color it has to set once stopped and reactivated.
    This strategy guarantees robots recognize moving robots in {\asynch} (those not hidden by other robots), thus avoiding ambiguous situations.
\end{itemize}

\subsection{\proc{Split} ({\bf Transition} \mycircled{4})}
\label{sec:procedure_split}
\proc{Split} takes in input \conf{periodic} where all $n$ robots are \colr{periodic}-colored and lay on $Cir$ in a periodic pattern.
This procedure partitions \conf{periodic} into $k\geq 2$ circular sectors $\Upsilon_0,\dots, \Upsilon_{k-1}$ such that \emph{(i)} they have the same arc length and \emph{(ii)} they are size-balanced (i.e. containing the same number of robots), and \emph{(iii)} they are chiral (i.e. the robots are arranged in an asymmetric pattern along the arc of each $\Upsilon_i$).
We call such sectors as \emph{uniform sectors}.
In this procedure, some robots will be elected as \emph{leaders} to fix the boundaries of each $\Upsilon_i$, and some robots will be elected and made to move to fix the chirality of the sector.
Leaders will turn their color into \colr{regular} if they lay on the boundaries of uniform sectors, or, otherwise, into \colr{leader} (i.e. if the boundary of a uniform sector lies in the middle point between two consecutive \colr{leader} robots).
We proceed as follows, distinguishing two cases based on $|\Phi|$.

\ifthenelse{\boolean{usefigures}}{
\begin{figure}[h] 
  	\centering
   \resizebox{\textwidth}{!}{
	\subfloat[$|\Phi| = 1$, $n$ odd.]{
			\centering
			\begin{tikzpicture}[scale=0.5, transform shape, font = {\LARGE}]
				\def\r{\ray}
				\def\a{360/13}
                \draw[draw=none] (0,0) circle (1.15*\r);
				\draw [very thick] (0,0) circle (\r);
				\draw[dotted] (90:\r) -- (270:\r); 
				\draw[] (90+\a:\r) -- (0:0); 
				\draw[] (90-\a:\r) -- (0:0); 
				\draw[] (270+0.5*\a:\r) -- (0:0); 
				\draw[] (270-0.5*\a:\r) -- (0:0); 
	
				\node (a) at (270:\r*0.6) {$\frac{2\pi}{n}$};
				\node (a) at (90-0.5*\a:\r*0.6) {$\frac{2\pi}{n}$};
				\node (a) at (90+0.5*\a:\r*0.6) {$\frac{2\pi}{n}$};
	
				\rnode{\colregular}{above}{}{90:\r};
				\rnode{\colblockL}{above left}{}{90+\a:\r};
				\rnode{\colblockL}{above right}{}{90-\a:\r};
	
				\rnode{\colblockR}{below}{}{270+0.5*\a:\r};
				\rnode{\colblockR}{below}{}{270-0.5*\a:\r};
		
			\end{tikzpicture}
		\label{subfig:unisplit_odd}}
	\hspace{0.5em}
	\subfloat[$|\Phi| = 3$, uniperiodic.]{

			\begin{tikzpicture}[scale=0.5, transform shape, font = {\LARGE}]
				\def\r{\ray}
				\def\a{360/16}
				\def\s{360/3}
		
                \draw[draw=none] (0,0) circle (1.15*\r);
				\draw [very thick] (0,0) circle (\r);

				\foreach \i in {0,1,2} 
				{	\draw[] (90+\i*\s-\a:\r) -- (0:0); 
					\draw[] (90+\i*\s+\a:\r) -- (0:0); 
					\draw[dotted] (90+\i*\s:\r) -- (0:0); 
				}
				\node (a) at (90-0.5*\a:\r*0.6) {$\frac{2\pi}{n}$};
				\node (a) at (90+0.5*\a:\r*0.6) {$\frac{2\pi}{n}$};

				\rnode{\colregular}{below right}{}{90-\s:\r};
				\rnode{\colblockL}{below right}{}{90-\s+\a:\r};
				\rnode{\colblockR}{below right}{}{90-\s-\a:\r};
	
				\rnode{\colregular}{left}{}{90-\s*2:\r};
				\rnode{\colblockL}{below left}{}{90-\s*2+\a:\r};
				\rnode{\colblockR}{above left}{}{90-\s*2-\a:\r};
	
				\rnode{\colregular}{above}{}{90-\s*3:\r};
				\rnode{\colblockL}{above left}{}{90-\s*3+\a:\r};
				\rnode{\colblockR}{above right}{}{90-\s*3-\a:\r};
	
			\end{tikzpicture}
		\label{subfig:unisplit_uniper}}
	\hspace{0.5em}
	\subfloat[$|\Phi| = 4$, biperiodic without \colr{regular} robots.]{

			\begin{tikzpicture}[scale=0.5, transform shape, font = {\LARGE}]
				\def\r{\ray}
				\def\p{2}
				\def\a{360/24}
                \draw[draw=none] (0,0) circle (1.15*\r);
				\draw [very thick] (0,0) circle (\r);
				\draw[dotted] (90:\r) -- (270:\r); 
	
				\foreach \i in {0,1} 
				{
					\draw[] (90*\i+\a*0.5:\r) -- (180 + 90*\i+\a*0.5:\r); 
					\draw[] (90*\i-\a*0.5:\r) -- (180 + 90*\i-\a*0.5:\r); 
				}
	
				\node (a) at (90:\r*0.7) {$\frac{2\pi}{n}$};
	
				\rnode{\collead}{above left}{}{90+\p*\a:\r};
				\rnode{\collead}{above right}{}{90-\p*\a:\r};
	
				\rnode{\collead}{above left}{}{270+\p*\a:\r};
				\rnode{\collead}{above right}{}{270-\p*\a:\r};
				
				\foreach \i in {0,2} 
				{
					\rnode{\colblockR}{below right}{}{90*\i+0.5*\a:\r};
					\rnode{\colblockR}{below left}{}{90*\i-0.5*\a:\r};
					\draw[dotted] (90*\i:\r) -- (0:0); 
				}
				
				\foreach \i in {1,3} 
				{
					\rnode{\colblockL}{below right}{}{90*\i+0.5*\a:\r};
					\rnode{\colblockL}{below left}{}{90*\i-0.5*\a:\r};
					\draw[dotted] (90*\i:\r) -- (0:0); 
				}
	
			\end{tikzpicture}
		\label{subfig:unisplit_biper}}
	\hspace{0.5em}
	\subfloat[$|\Phi| = 6$, biperiodic with \colr{regular} robots.]{

			\begin{tikzpicture}[scale=0.5, transform shape, font = {\LARGE}]
				\def\r{\ray}
                    \def\a{360/21} 	
				\def\L{10}		
                \draw[draw=none] (0,0) circle (1.15*\r);
				\draw [very thick] (0,0) circle (\r);

				\foreach \i in {0,2,4} 
				{
					\def\S{60*\i-30}          
                        \def\E{60*\i+60-30}       
					\draw[dotted] (\S:\r) -- (\S+180:\r); 
					\draw[] ({\E+\a}:\r) -- ({\E+\a+180}:\r); 
					\draw[] ({\E+\a}:\r) -- ({\E+\a+180}:\r); 
					
					\draw[] (\S+\a:\r) -- (\S+\a+180:\r); 
					\draw[] (\S-\a:\r) -- (\S-\a+180:\r); 
					
					\rnode{\collead}{below}{}{\S+\L:\r};
					\rnode{\collead}{below}{}{\S-\L:\r};
					\rnode{\colregular}{below}{}{\S:\r};
					\rnode{\colblockL}{below}{}{\S+\a:\r};
					\rnode{\colblockL}{below}{}{{\S -\a}:\r};
					
					\rnode{\colblockR}{below}{}{{\E+\a}:\r};
					\rnode{\colblockR}{below}{}{{\E-\a}:\r};

				}
				\node (a) at (90-\a*0.5:\r*0.7) {$\frac{2\pi}{n}$};
                \node (a) at (90+\a*0.5:\r*0.7) {$\frac{2\pi}{n}$};
                \node (a) at (270-\a*0.5:\r*0.7) {$\frac{2\pi}{n}$};
                \node (a) at (270+\a*0.5:\r*0.7) {$\frac{2\pi}{n}$};

			\end{tikzpicture}
		\label{subfig:unisplit_biper_odd}}}
\caption{\conf{unisect} where $Cir$ is split into uniform sectors (here delimited by dotted lines) using the light color: \colr{regular} (here \textit{\colregular}), \colr{left} (here \textit{\colblockL}),  \colr{right} (here \textit{\colblockR}), and \colr{leader} (here \textit{\collead}). All the other robots are omitted.}
\label{fig:uni_split}
\end{figure}
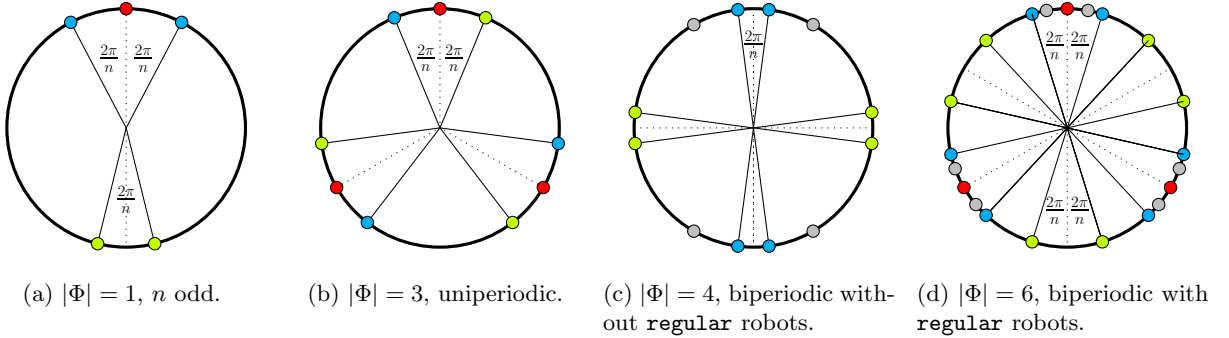 
}{}
\paragraph*{Case $|\Phi| = 1$:} In this case, we form two uniform sectors $\Upsilon_0,\Upsilon_1$. There are two subcases.
    \begin{itemize}
    \item{}\textbf{$n$ odd:}
    	It is possible to unambiguously select a diameter $d$ passing through only one robot $r$ and splitting $Cir$ into two size-balanced halves.
        These two $Cir$-halves correspond to $\Upsilon_0,\Upsilon_1$.
        The robot $r$ changes its color to \colr{regular}. 
        Then, for each $\Upsilon_i$, let $U_1$ be the uniform position on the arc $\Upsilon_i$ such that $\angle rOU_1 = {2\pi }/{n}$ and $U_{{(n-1)}/{2}}$ be the uniform position such that $\angle rOU_{{(n-1)}/{2}} = {(n-1)\pi}/{n}$. 
        Two robots from each $\Upsilon_i$ are elected and made to move to $U_1$ and $U_{{(n-1)}/{2}}$, setting their colors as \colr{left} and \colr{right}, respectively. See \Cref{subfig:unisplit_odd}.
    \item{}\textbf{$n$ even:}
     	It is possible to unambiguously select a diameter $d$ either \emph{(i)} passing through two antipodal robots $r$ and $r'$, and splitting $Cir$ into two size-balanced halves, or \emph{(ii)} passing through only one robot $r$ and splitting $Cir$ into two halves with $n/2$ and $n/2-1$ robots. 
        This diameter acts as the delimiter for the uniform sectors.
        The robot $r$ turns its color into \colr{regular}.
        When only one robot is located on $d$, then a robot $r'$ from the half containing $n/2$ robots reaches the empty endpoint of $d$.
        Similar to the odd case, two robots are elected in each sector, and reach the uniform positions $U_1$ (such that $\angle rOU_1 = 2\pi/n$) and $U_{(n-2)/2}$ (such that $\angle r'OU_{(n-2)/2} = 2\pi/n$), setting their colors as \colr{left} and \colr{right}, respectively.
        Finally, $r'$ sets its color to \colr{regular}.
    \end{itemize}

\paragraph*{Case $|\Phi| > 1$:}
        In this case, we form $|\Phi|$ uniform sectors.
        Let $(r_0, r_1, \dots, r_{|\Phi|-1})$ be the robots from $\Phi$ in a cyclic order.
        In the following, we consider indices $i\in\{0,\dots,|\Phi|-1\}$ modulo $n$.
        We distinguish between uniperiodic and biperiodic configurations as defined in \Cref{subsection:circleclasses}. 
        \begin{itemize}
    	\item{}\textbf{Uniperiodic:} 
        All the robots in $\Phi$ set their color as \colr{regular}. They define the boundaries of the uniform sectors so that each $\Upsilon_i$ is delimited by $\wideparen{r_ir_{i+1}}$.
        In a uniperiodic configuration, all the uniform sectors have the same chirality, w.l.o.g. clockwise.
        Now, inside each $\Upsilon_i$, two robots have to move and update their color to fix this chirality.
        Let $q$ be the number of \colr{periodic} robots in the uniform sectors.
        Let $U_1$ ($U_q$, resp.) be the uniform positions on the arc of $\Upsilon_i$ such that $\angle r_iOU_1 = 2\pi/n$ ($\angle U_qOr_{i+1} = 2\pi/n$, resp.).
        Within each $\Upsilon_i$, two \colr{periodic} robots are elected to be intended for the points $U_1$ and $U_q$.
        In at most two epochs, the elected \colr{periodic} robots reach the uniform positions $U_1$ and $U_q$, setting their colors as \colr{left} and \colr{right}, respectively.
        Note that, in this case, robots always enjoy complete visibility of the swarm, even if some robots are moving.
        See \Cref{subfig:unisplit_uniper}.
        
    	 \item{}\textbf{Biperiodic:}
            All the robots in $\Phi$ change their color to \colr{leader}.
            In this case, $|\Phi|$ is always even, and leaders never lay on the axes of reflective symmetry of the configuration.
    	 Let $B_i$ be the midpoint of the arc $\wideparen{r_{i-1}r_{i}}$ for each $i$.
        We define $\Upsilon_i$ as the sector whose arc is $\wideparen{B_iB_{i+1}}$.
        If a robot lies on a boundary $B_i$ for some $i$, it sets its color as \colr{regular}.
        Notice that each uniform sector $\Upsilon_i$ has a central angle $2\pi/|\Phi|$ and contains the \colr{leader}-colored robot $r_i$ (not necessary on the middle point of the arc of $\Upsilon_i$).
         In a biperiodic configuration, two adjacent sectors have opposite chirality. 
         Let us consider a uniform sector $\Upsilon_i$, and suppose that $B_i$ is the left boundary according to the chirality of $\Upsilon_i$.
        Let $U_1$ ($U_q$, resp.) be the point on the arc $\Upsilon_i$ such that it forms a central angle with $B_i$ ($B_{i+1}$, resp.) equal to \emph{(i)} $\frac{2\pi}{n}$ if a robot lies on $B_i$ ($B_{i+1}$, resp.), \emph{(ii)} $\frac{\pi}{n}$ otherwise.
         In the first epoch, two robots are elected from each sector $\Upsilon_i$ and they set their color as \colr{pre\_left} and \colr{pre\_right} so that the \colr{pre\_left} (\colr{pre\_right}, resp.) robot is intended to travel towards $U_1$ ($U_q$, resp.).
        If the elected \colr{pre\_left} (\colr{pre\_right}, resp.) robot already lies on $U_1$ ($U_q$, resp.), then it can directly assume the color \colr{left} (\colr{right}, resp.).
	Such two robots are unambiguously elected within each $\Upsilon_i$ among the closest adjacent \colr{periodic} robots to the points $U_1$ and $U_q$.
         In the second epoch, all the \colr{pre\_left} robots reach the $U_1$ point on the corresponding sector and change their color to \colr{left}.
           In the third epoch, all the \colr{pre\_right} robots reach the $U_q$ point on the corresponding sector and set their color as \colr{right}.
          See \Cref{subfig:unisplit_biper,subfig:unisplit_biper_odd}.
          \end{itemize}

Let \conf{unisect} be the configuration obtained at the end of \proc{Split} (see \Cref{fig:uni_split}).
In \conf{unisect}, all the $n$ robots lay on $Cir$, and they are partitioned in $k\geq 2$ uniform sectors $\Upsilon_0,\dots, \Upsilon_{k-1}$. 
All the robots laying on the boundaries of uniform sectors are \colr{regular}-colored.
Each $\Upsilon_i$ contains a \colr{left} robot and a \colr{right} robot, which set the chirality of $\Upsilon_i$.
All the robots that are not \{\colr{regular}, \colr{left}, \colr{right}\}-colored change their color to \colr{split}.
Note that all the \{\colr{regular}, \colr{left}, \colr{right}\}-colored robots will not move anymore: they already lay on uniform positions.
From now on, no robot will move out from its original uniform sector, and each of the subsequent procedures will run in parallel inside each uniform sector.
See \Cref{sec:correctness_procedure_split} for the correctness proofs and complexity analysis of \proc{Split}.
\begin{lemma}[\textbf{Split}]
\label{lemma:split}
Starting from \conf{periodic}, \conf{unisect} is reached in $O(1)$ epochs using $O(1)$ colors under {\asynch}, avoiding collisions, guaranteeing that robots always operate within $SEC(\config_{periodic})$. 
Robots are evenly partitioned in uniform sectors, whose chirality is set by \colr{right}- and \colr{left}-colored robots.
\end{lemma}

\subsection{\proc{Odd Block} ({\bf Transition} \mycircled{5})} 
\label{sec:procedure_oddblock}
Given \conf{unisect}, \proc{Odd Block} acts on each uniform sector independently.
As the name suggests, we build an \emph{odd block} inside each uniform sector $\Upsilon_i$.
Before describing the details of \proc{Odd Block}, we depict the target configuration that we want to achieve at the end of this procedure. 
\subparagraph*{Odd block.}
An odd block of size $2l+1$, for an integer $l$, is a circular sector of $Cir$ whose central angle measures $\frac{4l\pi}{n}$.
When $l$ is known, we will refer to the odd block as $l$-block.
An $l$-block contains $2l+1$ uniform positions along its arc, which we denote as $\barc$.
Let $U_1, U_2, \dots, U_{2l+1}$ be the uniform positions of this block, where $U_1$ and $U_{2l+1}$ lay on the endpoints of $\barc$.
The block contains two robots on $U_1$ and $U_{2l+1}$.
We will refer to such two robots as the \emph{left} and the \emph{right guard}.
A robot, called the \emph{median robot}, is located at position $U_{l+1}$.
Let $\chord$ be the block chord $\overline{U_1U_{2l+1}}$.
At the end of \proc{Odd Block}, the guards and the median robot are colored as \colr{blockL} (left), \colr{blockR} (right), and \colr{median}, respectively. 
These robots do not move anymore.
Instead, all the other $2l-2$ robots will have moved to distinct points of $\chord$, with the \colr{chord} color. 
Through the subsequent procedures (\proc{Small Circle} and \proc{Slice}), the \colr{chord} robots will uniformly arrange themselves on $\barc$, in order to cover the unoccupied uniform positions $U_2,\dots , U_{l}, U_{l+2}, \dots, U_{2l}$.
An odd block is built \emph{inside} each uniform sector $\Upsilon_i$.
The guards of the block fix the \emph{block chirality}, which is local and is the same as the chirality of the related uniform sector.

\ifthenelse{\boolean{usefigures}}{
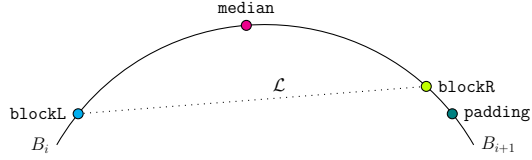
\begin{figure}[!h]
\begin{center}
		\begin{tikzpicture}[scale=0.4, transform shape, font = {\LARGE}]
			\def\r{8cm}		
			\def\a{360/40}
			\draw[]  ({\r * cos(30)}, {\r*sin(30)}) arc[start angle=30, end angle=150,radius=\r] ;
			\draw[dotted] (30+2*\a:\r) -- (150-\a:\r); 

			\node[label=180:$B_i$] (Bi) at (150:\r) {};
			\node[label=0:$B_{i+1}$] (Bi1) at (30:\r) {};
			\rnode{\colmedian}{above}{\colr{median}}{90+\a*0.5:\r};
			\rnode{\colblockL}{left}{\colr{blockL}}{150-\a:\r};
			\rnode{\colblockR}{right}{\colr{blockR}}{30+2*\a:\r};
			\rnode{\colpadding}{right}{\colr{padding}}{30+\a:\r};

   			\node[label=0:$\chord$] (Bi1) at (90:\r*0.75) {};

		\end{tikzpicture}
\end{center}
\caption{Odd block formation if $q$ is even.}
\label{fig:oddblock_padding}
\end{figure}}{}

\ifthenelse{\boolean{usefigures}}{
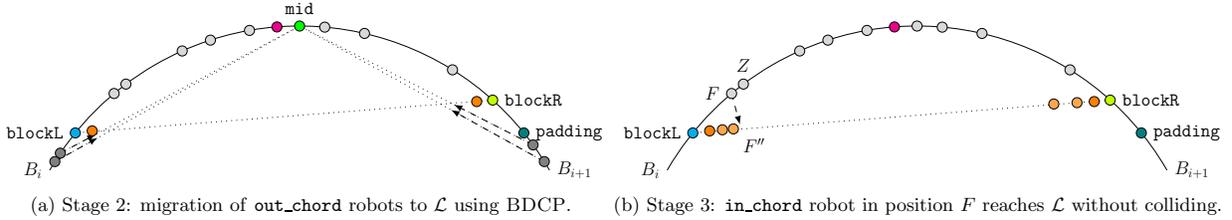
\begin{figure}[!h]
\resizebox{\textwidth}{!}{
\subfloat[Stage 2: migration of \colr{out\_chord} robots to $\chord$ using {\BDCP}.]{
    \centering
    \begin{tikzpicture}[scale=0.5, transform shape, font = {\LARGE}]
        \def\r{9.5cm}		
        \def\a{360/40}
        \draw[]  ({\r * cos(30)}, {\r*sin(30)}) arc[start angle=30, end angle=150,radius=\r] ;
        \draw[dotted] (30+2*\a:\r) -- (150-\a:\r); 

        \node[label=180:$B_i$] (Bi) at (150:\r) {};
        \node[label=0:$B_{i+1}$] (Bi1) at (30:\r) {};
        \rnode{\colmedian}{above left}{}{90+\a*0.5:\r};
        \rnode{\colpadding}{right}{\colr{padding}}{30+\a:\r};

        \rnode{\colblockL}{left}{\colr{blockL}}{150-\a:\r};
        \rnode{\colblockR}{right}{\colr{blockR}}{30+2*\a:\r};

        \foreach \i/\s in {2/0.3,6/0.3,116/0.4,118/0.4}
            {\draw[thin, dotted] (30+\i:\r) -- (90:\r);
            \draw [dashed, -latex,shorten >= \r*\s]  (30+\i:\r) -- (90:\r);
            \rnode{\coloutchord}{right}{}{30+\i:\r};
            }

        \foreach \i in {28,47,70, 55,78, 84, 97,100}
            {\rnode{\colinchord}{right}{}{30+\i:\r};}
        \rnode{\colbeacon}{right}{}{150-\a*1.28:\r*0.96};
        \rnode{\colbeacon}{right}{}{30+2.25*\a:\r*0.96};
        
            \rnode{green}{above}{\colr{mid}}{90:\r};

    \end{tikzpicture}
\label{fig:oddblock_stage2}}
\hfill
\subfloat[Stage 3: \colr{in\_chord} robot in position $F$ reaches $\chord$ without colliding.]{
		\begin{tikzpicture}[scale=0.5, transform shape, font = {\LARGE}]
			\def\r{9.5cm}		
			\def\a{360/40}
			\draw[]  ({\r * cos(30)}, {\r*sin(30)}) arc[start angle=30, end angle=150,radius=\r] ;
			\draw[dotted] (30+2*\a:\r) -- (150-\a:\r); 

   			\rnode{\colpadding}{right}{\colr{padding}}{30+\a:\r};

			\node[label=180:$B_i$] (Bi) at (150:\r) {};
			\node[label=0:$B_{i+1}$] (Bi1) at (30:\r) {};
			\rnode{\colmedian}{above left}{}{90+\a*0.5:\r};

			\rnode{\colblockL}{left}{\colr{blockL}}{150-\a:\r};
			\rnode{\colblockR}{right}{\colr{blockR}}{30+2*\a:\r};

                \rnode{\colbeacon}{right}{}{150-\a*1.28:\r*0.96};
                \rnode{\colbeacon}{right}{}{30+2.25*\a:\r*0.96};
                
                \rnode{\colbeacon!70}{right}{}{150-\a*1.5:\r*0.93};
                \rnode{\colbeacon!70}{below right}{$F''$}{150-\a*1.7:\r*0.905};

                \rnode{\colbeacon!70}{right}{}{30+2.55*\a:\r*0.92};
                \rnode{\colbeacon!70}{right}{}{30+3*\a:\r*0.87};

                \draw[thin, dashed, -latex] (30+100:\r) -- (30+100+3:\r*0.9);

                \foreach \i in {28,47,70, 55,60,78, 84}
                    {\rnode{\colinchord}{above}{}{30+\i:\r};}

                \rnode{\colinchord}{above}{$Z$}{30+97:\r};
                \rnode{\colinchord}{left}{$F$}{30+100:\r};

		\end{tikzpicture}
\label{fig:oddblock_stage3}}
}
\caption{Stages in \proc{Odd Block}.}
\label{fig:oddblock_stages}
\end{figure}
}{}
Given a uniform sector $\Upsilon_i$ in \conf{unisect}, \proc{Odd Block} works in three stages.  
\begin{itemize}
    \item{}\textbf{Stage 1.} Defining the guards of the odd block within $\Upsilon_i$ by adding a \colr{padding} robot if necessary.
    \item{}\textbf{Stage 2.} Moving the robots located between the boundaries of $\Upsilon_i$ and $U_1$ (resp. $U_{2l+1}$) to the block chord $\chord$.
    \item{}\textbf{Stage 3.} Moving the robots from $\barc$ to the block chord $\chord$.
\end{itemize}

Let $\Upsilon_i$ be a uniform sector in \conf{unisect}, and let $q$ be the number of robots in $\Upsilon_i$ (except for those on its boundaries).
Indeed, all the uniform sectors contain $q$ robots.
Let $U_1,\dots, U_q$ be the uniform positions along the arc of $\Upsilon_i$, such that the \colr{left} robot lies on $U_1$ whereas the \colr{right} robot lies on $U_q$.
We describe separately below how we perform the movement of robots in each of the five stages.
All the robots involved in the following stages belong to $\Upsilon_i$.
Note that some robots use the \colr{pre\_<color>} colors to synchronize themselves in the various uniform sectors before starting to move (as explained at the beginning of \Cref{section:componentIII}).

\subparagraph*{Stage 1.}  
If $q$ is even, the \colr{split} robot closest to $U_{q-1}$ sets its color as \colr{pre\_blockR} and, after the synchronization with all the other \colr{pre\_blockR} in the other uniform sectors, it reaches $U_{q-1}$ where $\angle U_{q-1}OU_q = 2\pi/n$, and it changes its color to \colr{blockR}.
The \colr{right} robot changes its color to \colr{padding}. The \colr{padding} robot will no longer move.
The \colr{left} robot changes its color to \colr{blockL}.  
If $q$ is odd, then the robots with color \colr{left} and \colr{right} change their colors to \colr{blockL} and \colr{blockR}, respectively. 
In both cases, \colr{blockL} and \colr{blockR} robots act as the left and right guards for each odd block. 
The block chord $\chord$ is defined as the chord joining the guards, and the block arc $\barc$ is the arc cut by $\chord$.
Then, the \colr{split} robot closest to the midpoint of $\barc$ elects itself as \colr{pre\_median} and then moves to the midpoint, changing its color to \colr{median} (see \Cref{fig:oddblock_padding}).

\subparagraph*{Stage 2.}
Let us consider the $l$-odd block built within $\Upsilon_i$.
Let $B_i$ and $B_{i+1}$ be the boundaries of $\Upsilon_i$, and let $U_1$ and $U_{2l+1}$  be the boundaries of $\Gamma$ (i.e. where the guards lay on).
All the \colr{split} robots which lay on $\Gamma$ change their color to \colr{in\_chord}, otherwise they change into \colr{out\_chord} (i.e. if they lay on $\wideparen{B_iU_1}$ and on $\wideparen{U_{2l+1}B_{i+1}}$).
This stage aims to make all the \colr{out\_chord} robots migrate to $\chord$.
If no robot with color \colr{out\_chord} exists in $\Upsilon_i$, then Stage 3 begins directly.
Let $M$ be the middle point of the arc of $\Upsilon_i$.
If the \colr{median} robot does not lay on $M$, then the \colr{split} robot closest to $M$ elects itself as \colr{pre\_mid} and then it moves to $M$ changing its color to \colr{mid}.
Otherwise, the \colr{median} robot sets its color to \colr{mid}.
Notice that, for $q \geq 12$, $M$ always lies on $\Gamma$.
Now, we use {\BDCP} to make \colr{out\_chord} robots migrate to $\chord$.
Firstly, we set two beacons on $\chord$: one left beacon next to $U_1$ and one right beacon next to $U_{2l+1}$. Such two beacons together with the guards on $U_1$ and $U_{2l+1}$ form the 4 beacons needed to implement {\BDCP}.
Let us show how to select the left beacon $b$ (a complement strategy holds for the selection of the right beacon).
If there exists an \colr{out\_chord} robot on $\wideparen{B_iU_1}$, then the closest \colr{out\_chord} robot to $U_1$ is elected as $b$. In this case, $b$ changes its color to \colr{pre\_beacon} and moves to the intersection of the line segments $\overline{DM}$ and $\chord$, where $D$ was the position of $b$.
Otherwise, $b$ is the farthest robot to $U_1$ selected from the other \colr{out\_chord} robots on the arc of $\Upsilon_i$. 

In this case, $b$ reaches a position on $\chord$ so that it can play the role of the left beacon and it does not create collinearities with the \colr{mid} robot. Once on $\chord$, $b$ sets its color as \colr{beacon}.
An analogous procedure is repeated to choose the right beacon close to $U_{2l+1}$. 
The two \colr{beacon} robots together with the \colr{blockL} and \colr{blockR} robots act as the four beacons for the {\BDCP} procedure (see \Cref{fig:oddblock_stage2}).
By {\BDCP}, each \colr{out\_chord} robot $r$ located at $E$ heads to $\chord$ to the point intersecting $\overline{EM}$ and $\chord$. 
Once $r$ reaches $\chord$, it changes its color to \colr{chord}.
Note that the setting of the \colr{mid} robot at the midpoint of the arc $\Upsilon_i$ was necessary to make $r$ select the \colr{mid} robot belonging to $\Upsilon_i$ (in fact, no ambiguity arises since the \colr{mid} robot of $\Upsilon_i$ is always closer to $r$ than the other \colr{mid} robots).
Stage 2 ends when no more \colr{out\_chord} robots are located in $\Upsilon_i$.

\subparagraph*{Stage 3.}
In this stage, we apply {\BDCP} to make the \colr{in\_chord} robots migrate from $\Gamma$ to $\chord$. 
From now on, the \colr{mid} robot is no longer necessary, so it changes its color to \colr{in\_chord} if the \colr{median} exists on the block arc.
Otherwise, the \colr{mid} robot resets its color to \colr{median}.
At the end of Stage 2, there are \colr{beacon}- and \colr{chord}-colored robots on $\chord$. 
Such robots can act as beacons for fixing the chord $\chord$ for the \colr{in\_chord} robots.
Let $r$ be a \colr{in\_chord} robot located at $F$ on $\Gamma$. 
Let $F'$ be the projection of $F$ on $\chord$, i.e., $\overline{FF'}$ is perpendicular to $\chord$. 
Let $Z$ be the robot location such that the projection $Z'$ on $\chord$ is the closest to $F'$.
If $F'$ is already occupied by a robot, then the revised destination is computed as $F''$ where $|\overline{F'F''}| =  |\overline{F'Z'}|/3$ (see \Cref{fig:oddblock_stage3}).
Once $r$ reaches $\chord$, it changes its color to \colr{chord}.
Note that, if there were no robots colored \colr{out\_chord} in Stage 2, then the closest \colr{in\_chord} colored robots to $U_1$ and $U_{2l+1}$ move to their projection to act as beacons at the beginning of Stage 3.
This ensures that all robots reach $\chord$ in $O(1)$ time by implementing {\BDCP}.

See \Cref{sec:correctness_procedure_oddblock} for the correctness proofs of \proc{Odd Block}.

\begin{lemma}[\textbf{Odd Block}]\label{lemma:blockformation}
Starting from \conf{unisect}, \conf{oddblock} is reached in $O(1)$ epochs using $O(1)$ colors under {\asynch}, avoiding collisions, guaranteeing robots always operate within $SEC(\config_{unisect})$. 
In each odd block, the guards are colored as \colr{blockL} and \colr{blockR}, the median robot is colored as \colr{median}, whereas all the other robots are located on the block chord $\chord$ with color \colr{chord}.
\end{lemma}

\subsection{\proc{Small Circle} ({Transition} \mycircled{6})}\label{sec:procedure_smallcircle}
Consider an $l$-block in \conf{oddblock}. 
W.l.o.g., let the block chirality be clockwise.
First, the left guard, the median robot, and the right guard change their color to \colr{scL}, \colr{scMedian}, and \colr{scR}, respectively.
Let $\diam$ be the line segment connecting $U_{l+1}$ to the projection of $U_{l+1}$ on the block chord $\chord$.
We define the \emph{small circle} $\SC$ as the circle such that $\diam$ is its diameter (i.e. $U_{l+1}$ lays on $\SC$ and $\chord$ is its tangent). 
We say that $\diam$ splits $\SC$ into two halves $\SC_w$ (left) and $\SC_e$ (right). 
We will refer to $\diam$ as the \emph{median diameter}. 
See \Cref{fig:smallcircle_init_config}.

\ifthenelse{\boolean{usefigures}}{
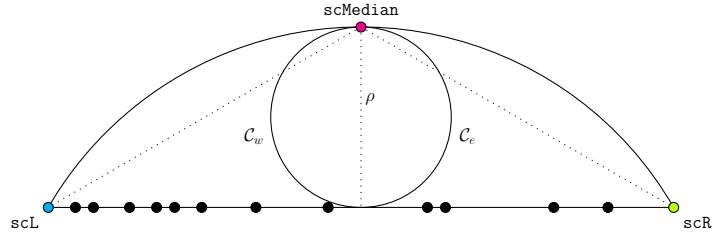
\begin{figure}[!h]
\begin{center}

		\begin{tikzpicture}[scale=0.4, transform shape, font = {\LARGE}]
			\def\r{3cm} 
			\draw[fill=yellow!0] (0,-\r) -- ({sqrt(12)*\r}, -\r) arc[start angle={30}, end angle={150},radius=\r*4] -- (0, -\r);
			\draw [] (0,0) circle (\r);
			\draw [thin, dotted] (90:\r) -- (270:\r);
			\node (a) at (\r*0.1,\r*0.2) {$\diam$};
			
			\node (a) at (190:\r*1.2) {$\SC_w$};
			\node (a) at (-10:\r*1.2) {$\SC_e$};

			\draw [thin, dotted] (90:\r) -- ({sqrt(12)*\r}, -\r);
			\draw [thin, dotted] (90:\r) -- (-{sqrt(12)*\r}, -\r);
							
			\rnode{\colblockL}{below left}{\colr{scL}}{-{sqrt(12)*\r}, -\r};
			\rnode{\colblockR}{below right}{\colr{scR}}{{sqrt(12)*\r}, -\r};
			
			\rnode{\colmedian}{above}{\colr{scMedian}}{90:\r};
			
			\foreach \a in {0.3, 0.5, 0.9, 1.2, 1.4, 1.7, 2.3, 3.1, 4.2,4.4, 5.6, 6.2}
				{\rnode{\colblock}{below}{}{-{sqrt(12)*\r + \r*\a}, -\r};
			}
			
		\end{tikzpicture}
	
\end{center}
\caption{An odd-block at the beginning of \proc{Small Circle}, containing the small circle $\SC$. All the \colr{chord} robots (here \emph{\colblock}) lay on $\chord$.}
\label{fig:smallcircle_init_config}
\end{figure}
}{}

Now, \proc{Small Circle} proceeds in three stages as described in the following.

\begin{itemize}
	\item\textbf{Stage 1 -- Moving robots to the small circle:} 
Each \colr{chord} robot moves along the line joining its position and the \colr{scMedian} robot, until it reaches $\SC$.

	\item\textbf{Stage 2 -- Moving all robots to the right half:}
The robots on $\SC_w$ migrate to $\SC_e$.

	\item\textbf{Stage 3 -- Balancing robots on the small circle:}
The robots on $\SC_e$ split into two equal groups, and one of the groups comes back to $\SC_w$ forming a reflective symmetric configuration on $\SC$.
\end{itemize}

\subparagraph*{Stage 1.}
Let $r$ be a \colr{chord} robot on the block chord $\chord$.
It always sees the \colr{scMedian} robot of its block.
Thanks to the presence of the other robots on $\chord$ (\colr{chord}-, \colr{scR}- or \colr{scL}-colored), $r$ reconstructs the supporting line of $\chord$, and so $\diam$ and $\SC$.
Then, $r$ reaches $\SC$ traveling along the trajectory connecting its own position to the \colr{scMedian} robot (see \Cref{fig:smallcircle_stage1}), and it updates its color to \colr{smallcircle}.
Consider the line connecting the left (right, resp.) guard of the block with the median robot such that it splits $\SC_w$ ($\SC_e$, resp.) into two arcs.
Indeed, all the \colr{smallcircle} robots on $\SC$ lay on the lower arcs of $\SC_w$ and $\SC_e$. We call such arcs the \emph{safe arcs} of $\SC$ (see \Cref{fig:smallcircle_stage1}).
When the \colr{scR} (resp. \colr{scL}) robot does not see any \colr{chord} or \colr{to\_smallcircle} (i.e the moving robots towards $\SC$ before becoming \colr{smallcircle}) robots on its half-block, then it sets its color as \colr{scR\_complete} (resp. \colr{scL\_complete}).
When a \colr{smallcircle} robot on $\SC_e$ (resp. $\SC_w$) sees the \colr{scR\_complete} (resp. \colr{scL\_complete}) robot on its half-block, it sets its colors as \colr{smallcircle\_complete}.
These color updates are needed to synchronize all robots on $\SC$ before proceeding with the next stage.

See \Cref{lemma:smallcircle_stage1_proof} in \Cref{appendix:smallcircle} for the correctness proof of Stage 1.

\ifthenelse{\boolean{usefigures}}{
\begin{figure}[!h]
\begin{center}

		\begin{tikzpicture}[scale=0.4, transform shape, font = {\LARGE}]
			\def\r{3cm} 
			\draw[fill=yellow!0] (0,-\r) -- ({sqrt(12)*\r}, -\r) arc[start angle={30}, end angle={150},radius=\r*4] -- (0, -\r);
			\draw [] (0,0) circle (\r);
			\draw [thin, dotted] (90:\r) -- (270:\r);

			\draw [thin, dashed] (90:\r) -- ({sqrt(12)*\r}, -\r);
			\draw [thin, dashed] (90:\r) -- (-{sqrt(12)*\r}, -\r);
							
			\rnode{\colblockL}{below left}{\colr{scL}}{-{sqrt(12)*\r}, -\r};
			\rnode{\colblockR}{below right}{\colr{scR}}{{sqrt(12)*\r}, -\r};
			
			\rnode{\colmedian}{above}{}{90:\r};
			
			\foreach \a in {0.3, 0.5, 0.9, 1.2, 1.4, 1.7, 2.3, 3.1, 4.2,4.4, 5.6, 6.2}
				{\rnode{\colblock}{below}{}{-{sqrt(12)*\r + \r*\a}, -\r};
				\draw [thin, dotted] (90:\r) -- (-{sqrt(12)*\r + \r*\a}, -\r);
			}
			
		\end{tikzpicture}
	
\end{center}
\caption{Stage 1 of \proc{Small Circle}: trajectories of \colr{chord} robots to reach $\SC$. The two dashed lines delimit the \emph{safe arcs} on $\SC_e$ and $\SC_w$.}
\label{fig:smallcircle_stage1}
\end{figure}
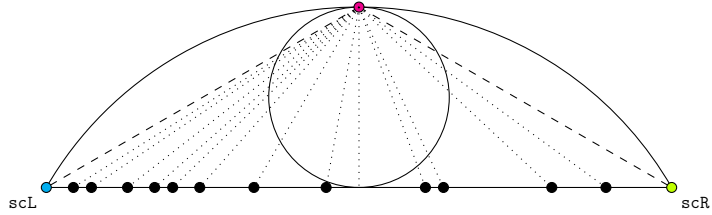
}{}

\subparagraph*{Stage 2.}
At the beginning of this stage, all the robots on $\chord$ have reached the safe arcs of $\SC$. 
Specifically, we have two groups of \colr{smallcircle\_complete} robots on $\SC_e$ and $\SC_w$, possibly with different cardinalities.
We assume all the \colr{smallcircle\_complete} robots on $\SC_e$ ($\SC_w$, resp.) set their color as \colr{smallcircle\_east} (\colr{smallcircle\_west}, resp.).
If the robots on $\SC_e$ and $\SC_w$ lay in a reflective symmetry, they do nothing else in this stage.
Otherwise, let $r$ be a \colr{smallcircle\_west} robot on $\SC_w$.
If a \colr{smallcircle\_east} robot lies on $\SC_e$, symmetrical with respect to $r$, then $r$ moves to another position $x$ on $\SC_w$ such that \emph{(i)} $r$ does not collide with other robots or change its rank in the displacement of robots, \emph{(ii)} the symmetrical position of $x$ on $\SC_e$ is robot-free, and \emph{(iii)} $r$ does not leave the safe arc of $\SC_w$.
When all the \colr{smallcircle\_west} robots have properly shifted their positions (if necessary), then they use {\BDCP} to migrate on their projection on $\SC_e$.
After the migration, we assume that all the robots on $\SC_e$ (except for the median one) are colored as \colr{smallcircle\_east}.
See \Cref{lemma:smallcircle_stage2_proof} in \Cref{appendix:smallcircle} for the correctness proof of Stage 2.

\subparagraph*{Stage 3.}
The \colr{smallcircle\_east} robots update their color alternating two colors, \colr{west} and \colr{east}, so that the closest robot to the median robot is colored as \colr{east}.
After that, each \colr{west} robot heads to the projection on $\SC_w$ of its upper-adjacent \colr{east} robot (see \Cref{fig:slice_stage3}).
See \Cref{lemma:smallcircle_stage3_proof} in \Cref{appendix:smallcircle} for the correctness proof of Stage 3.

\ifthenelse{\boolean{usefigures}}{
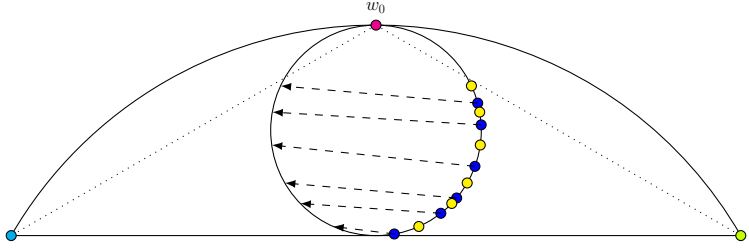
\begin{figure}[!h]
\begin{center}

		\begin{tikzpicture}[scale=0.4, transform shape, font = {\LARGE}]
			\def\r{3.5cm} 
			\draw[fill=yellow!0] (0,-\r) -- ({sqrt(12)*\r}, -\r) arc[start angle={30}, end angle={150},radius=\r*4] -- (0, -\r);
			\draw [] (0,0) circle (\r);
		
			\draw [thin, dotted] (90:\r) -- ({sqrt(12)*\r}, -\r);
			\draw [thin, dotted] (90:\r) -- (-{sqrt(12)*\r}, -\r);
						
			\rnode{\colblockL}{below right}{}{-{sqrt(12)*\r}, -\r};
			\rnode{\colblockR}{below left}{}{{sqrt(12)*\r}, -\r};
			
			\rnode{\colmedian}{above}{$w_0$}{90:\r};
			
			\foreach \e/\w in {65/75, 80/87, 98/110, 120/130, 134/142, 156/170}
				{\rnode{\coleast}{right}{}{90-\e:\r};
				\rnode{\colwest}{left}{}{90-\w:\r};
				\draw [thin, dashed, -latex] (90-\w:\r) -- (90+\e:\r);
			}
			
		\end{tikzpicture}
	
\end{center}
\caption{Stage 3 of \proc{Slice}: \colr{west} robots reach $\SC_e$.}
\label{fig:slice_stage3}
\end{figure}
}{}

\begin{lemma}[\textbf{Small Circle}]\label{lemma:smallcircle}
Starting from \conf{oddblock}, \conf{smallcircle} is reached in $O(1)$ epochs using $O(1)$ colors in the \asynch{} setting, avoiding collisions, guaranteeing robots operate within $SEC(\config_{oddblock})$. 
In \conf{smallcircle}, the robots in each odd block (guards and median robot excluded) are arranged on the same \emph{small circle} $\SC$ in a mirror-symmetric pattern where the line of symmetry lies on the median diameter $\diam$.
\end{lemma}
\subsection{\proc{Slice} ({Transition} \mycircled{7})}
\label{sec:procedure_slice}
Consider an $l$-block in \conf{smallcirle} (see \Cref{fig:slice_init_config}).
First, the guards and the median robots turn into \colr{sliceL}, \colr{sliceR} and \colr{sliceMedian}, respectively.
Let $Q$ be the center of $\SC$.
Let $e_1,\dots, e_m$ (resp., $w_1,\dots, w_m$) be the ordered sequences of \colr{east} (resp., \colr{west}) robots which lay on $\SC_e$ (resp., $\SC_w$) starting from the closest to the median robot.
Note that $e_j$ and $w_j$ lay on symmetric positions in the two half-circles, for each $1\leq j \leq m$.
We say that the median robot splits $\barc$ into the right arc $\barc_e$ and the left arc $\barc_w$.
\proc{Slice} aims at moving the robots from $\SC$ to $\barc$ to the vertices of the target regular $n$-gon.
Let us present the three main stages that constitute the entire procedure:
\begin{itemize}
	\item\textbf{Stage 1 -- Rank encoding:} Each $w_j$ encodes its rank $j$ by reaching a specific position on $\SC_w$.
	\item\textbf{Stage 2 -- Moving towards the right target vertices:} Each $e_j$ reaches its uniform position on $\barc_e$.
	\item\textbf{Stage 3 -- Moving towards the left target vertices:} Each $w_j$ reaches its uniform position on $\barc_w$.
\end{itemize}

\begin{figure}[!h]
\begin{center}

		\begin{tikzpicture}[scale=0.4, transform shape, font = {\LARGE}]
			\def\r{3cm} 
			\draw[fill=yellow!0] (0,-\r) -- ({sqrt(12)*\r}, -\r) arc[start angle={30}, end angle={150},radius=\r*4] -- (0, -\r);
			\draw [] (0,0) circle (\r);
			\draw [thin, dotted] (90:\r) -- (270:\r);
			\node (a) at (\r*0.1,\r*0.2) {$\diam$};
			
			\draw [thin, dotted] (90:\r) -- ({sqrt(12)*\r}, -\r);
			\draw [thin, dotted] (90:\r) -- (-{sqrt(12)*\r}, -\r);

			\node (a) at (2*\r,\r*0.6) {$\barc_e$};
			\node (a) at (-2*\r,\r*0.6) {$\barc_w$};

			\rnode{\colblockL}{below}{\colr{scL}}{-{sqrt(12)*\r}, -\r};
			\rnode{\colblockR}{below}{\colr{scR}}{{sqrt(12)*\r}, -\r};
			
			\rnode{\colmedian}{above}{$w_0$}{90:\r};
			
			\foreach \i/\a in {1/65, 2/80, 3/98, 4/120, 5/134, 6/156}
				{\rnode{\colwest}{left}{$w_\i$}{90+\a:\r};
				\rnode{\coleast}{right}{$e_\i$}{90-\a:\r};
			}
			
		\end{tikzpicture}
	
\end{center}
\caption{\conf{smallcirle} at the beginning of \proc{Slice}.}
\label{fig:slice_init_config}
\end{figure}
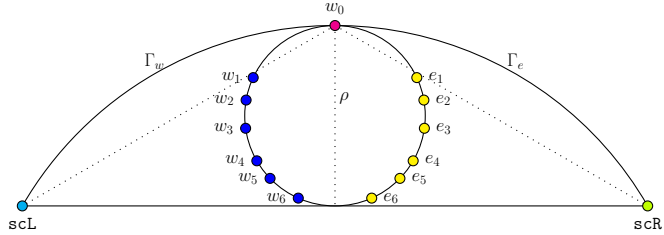

\subparagraph*{Stage 1.}
Let us denote with $w_0$ and $w_{m+1}$ the two endpoints of $\SC_w$ (the median robot laying on $w_0$).
Let $\alpha_j = \angle w_jQw_{j+1}$, for each $0\leq j \leq m$.
At the beginning of this stage, robot $w_1$ computes $\delta=\min\{\alpha_j\}_{0\leq j \leq m}$, moves to a position $w'_1$ on $\SC_w$ so that $\angle w_0Qw'_1 = \delta$ setting its color as \colr{angle} (see \Cref{subfig:delta_setting}).
In this way, the \colr{angle} robot fixes the amplitude of the slices both $\SC_e$ and $\SC_w$ will be split throughout this procedure.
After this setting, each \colr{west} robot splits $\SC_w$ into $\delta$-slices, starting from the point $w_0$ (see \Cref{subfig:slicing}).
Let $\eta_0,\dots, \eta_{\floor{\frac{\pi}{\delta}}-1}$ be the ordered sequence of such $\delta$-slices, starting from the median robot.
We ignore the possible remaining slice of amplitude $<\delta$ that can exist between $\eta_{\floor{\frac{\pi}{\delta}}-1}$ and $w_{m+1}$.
Indeed, the arc of each $\delta$-slice contains at most two robots: in the case of exactly two robots, both of them must lay on the endpoints of the slice arc.
Let $w_j$ be a \colr{west} robot and let $\eta_k$ be the slice it lays on (if $w_j$ lays on the endpoint between two slices, then it chooses the slice with the smaller rank $k$).
Then, $w_j$ moves to a new position $w'_j$ on $\SC_w$ such that $\angle w_0Ow'_j= k\delta + \frac{j\delta}{m+1}$.
Note this strategy always guarantees $w'_j$ is strictly contained in $\eta_k$, thus ensuring these movements never create collisions or collinearities among \colr{west} robots.
See \Cref{lemma:slice_stage1_proof} in \Cref{appendix:slice} for the correctness proof of Stage 1.

\begin{figure}[!h]
\centering
	\subfloat[Setting of $\delta$.]{
		\begin{tikzpicture}[scale=0.4, transform shape, font = {\LARGE}]
			\def\r{4cm} 
			\draw [] (0,0) circle (\r);
			\draw [thin, dotted] (90:\r) -- (270:\r);
			
			\node (a) at (90+7.5:\r*0.8) {$\delta$};
			\draw [thin, dotted] (0:0) -- (90+15:\r);

			\rnode{\colmedian}{above}{$w_0$}{90:\r};
			
			\foreach \i/\a in {2/80, 3/98, 4/120, 5/134, 6/156}
				{\rnode{\colwest}{left}{$w_\i$}{90+\a:\r};
				\rnode{\coleast}{right}{$e_\i$}{90-\a:\r};
			}
			
			\draw [thin, dashed, -{Latex}] (90+65:\r) -- (90+18:\r);
			\rnode{\colwest}{above left}{$w_1$}{90+65:\r};
			\rnode{\colangle}{above left}{$w'_1$}{90+14:\r};
			\rnode{\coleast}{above right}{$e_1$}{90-65:\r};

		\end{tikzpicture}
		\label{subfig:delta_setting}
	}
	\hspace{\textwidth/4}
	\subfloat[Splitting of $\SC_w$ in $\delta$-slices.]{
		\begin{tikzpicture}[scale=0.4, transform shape, font = {\LARGE}]
			\def\r{4cm} 
			\def\d{15} 
			\def\m{12} 

			\draw [] (0,0) circle (\r);
			\draw [thin, dotted] (90:\r) -- (270:\r);
			
			\node (a) at (90+7.5:\r*0.8) {$\delta$};
			\draw [thin, dotted] (0:0) -- (90+\d:\r);

			\rnode{\colmedian}{above}{$w_0$}{90:\r};
			
			\foreach \i/\a in {2/80, 3/98, 4/120, 5/134, 6/156}
				{\rnode{\colwest}{left}{$w_\i$}{90+\a:\r};
				\rnode{\coleast}{right}{$e_\i$}{90-\a:\r};
			}
			
			\rnode{\colangle}{above left}{$w_1$}{90+14:\r};
			\rnode{\coleast}{above right}{$e_1$}{90-65:\r};
			
			\foreach \i in {1, ..., \m}
				{\draw [thin, dotted] (0:0) -- (90+\d*\i:\r);}

		\end{tikzpicture}
		\label{subfig:slicing}
	}
\label{fig:procedure_slice}
\caption{Steps of Stage 1 in \proc{Slice}.}
\end{figure}
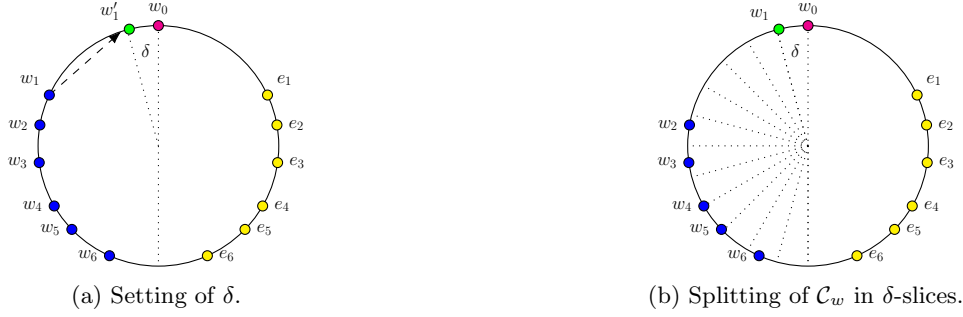

\subparagraph*{Stage 2.}
In this stage, each \colr{east} robot reaches its uniform position on $\barc_e$ in two steps.
Let $e_1,\dots,e_m$ be the sequence of the \colr{east} robots and let $U_1, \dots , U_m$ be the sequence of the uniform positions on $\barc_e$ so that $U_j$ is intended for $e_j$.
In the first step, robots $e_1,e_2$ and $e_3$ turn their color as \colr{beacon}.
Afterward, all the other \colr{east} robots $\{e_j\}_{4\leq j\leq m}$ on $\barc_e$ migrate towards $\diam$.
Specifically, each \colr{east} robot $e_j$ directly heads to the projection $e'_j$ of its position on $\SC_e$ along $\diam$, setting its color as \colr{east\_diameter}.
This migration is accomplished to make robots be aligned and so to avoid collisions and collinearities which might not allow a complete parallelism in the second step.
Finally, the three beacons $e_1,e_2$ and $e_3$ migrate on $\barc_w$ in order to fix $\barc$ and so to help the \colr{east\_diameter} robots to uniformly arrange themselves on the arc $\barc_e$ in the second step.
The beacons must safely migrate to some positions that are visible to the robots on $\diam$.
Let $u$ be the intersection point between $\barc_w$ and the line connecting $e_1$ (the first \colr{beacon}) with $w_1$ (the \colr{angle} robot).
Let $v$ be the intersection point between $\barc_w$ and the line connecting $e'_m$ with $w_2$.
We call $\wideparen{uv}$ the \emph{safe arc} of $\barc_w$.
So, $e_1,e_2$ and $e_3$ reach some deterministic and internal points of the safe arc $\wideparen{uv}$ (see \Cref{fig:slice_stage2_beacons}).

In the second step, all the \colr{east\_diameter} robots on $\diam$ compute and reach their uniform position.
Let $e$ be an \colr{east\_diameter} robot on $\diam$.
To compute its uniform position, $e$ needs to recover the arc $\barc_e$ and its rank $j$.
The rank $j$ can be obtained by slicing $\SC_w$ (resp., $\SC_e$) in $\eta_0,\dots, \eta_{\floor{\frac{\pi}{\delta}}-1}$ (resp., $\eta'_0,\dots, \eta'_{\floor{\frac{\pi}{\delta}}-1}$).
Let $\eta'_k$ be the slice where $e$ originally laid on (before its migration on $\diam$).
Since $e$ can see all the \colr{west} robots, it is sufficient for $e$ to compute the rank $j$ of the \colr{west} robot located on $\eta_k$.
After having obtained its rank $j$ and the arc $\barc_e$ (thanks to the presence of the beacons fixing $\barc$), $e$ heads to $U_j$ setting its color as \colr{regular} (see \Cref{fig:slice_stage2_east_ui}).
Finally, the beacons $e_1,e_2, e_3$ reach the missing uniform positions $U_1,U_2,U_3$.
See \Cref{lemma:slice_stage2_proof} in \Cref{appendix:slice} for the correctness proof of Stage 2.

\begin{figure}[!h]
\begin{center}
	\subfloat[Step 1: \colr{east\_diameter} robots on $\diam$ and beacons on the safe arc $\wideparen{uv}$.]{
		\begin{tikzpicture}[scale=0.4, transform shape, font = {\LARGE}]
			\def\r{4cm} 
			\def\m{6}
			\def\angle{60/7}

			\draw(0,-\r) -- ({sqrt(12)*\r}, -\r) arc[start angle={30}, end angle={150},radius=\r*4] -- (0, -\r);
			\draw [] (0,0) circle (\r);
			\draw [thin, dotted] (90:\r) -- (270:\r);

			\rnode{\colangle}{above left}{$u$}{90+14:\r};

			\foreach \i in {1,..., \m}
				{\vertx{\colvertx}{above}{$U_{\i}$}{ {4*\r*cos(90-\angle*\i)} , {4*\r*sin(90-\angle*\i) -3*\r} };}
				
			\foreach \i/\a in {1/65, 2/80, 3/98}
				{\rnode{\colbeacon}{above}{}{ {4*\r*cos(97+3*\i)} , {4*\r*sin(97+3*\i) -3*\r} };
				\rnode{\colbeacon}{right}{$e_\i$}{90-\a:\r};
				\draw [thin, dashed, -latex] (90-\a:\r) -- ({4*\r*cos(97+3*\i)} , {4*\r*sin(97+3*\i) -3*\r});
			}

			\rnode{\colblockL}{below}{\colr{sliceL}}{-{sqrt(12)*\r}, -\r};
			\rnode{\colblockR}{below}{\colr{sliceR}}{{sqrt(12)*\r}, -\r};
			
			\rnode{\colmedian}{above}{$w_0$}{90:\r};
			
			\foreach \i/\a in {2/80, 3/98, 4/120, 5/134, 6/156}
				{\rnode{\colwest}{left}{$w_\i$}{90+\a:\r};}
				
			\foreach \i/\a in {4/120, 5/134, 6/156}
				{\rnode{\coleastdiameter}{left}{}{0,{\r*sin(90-\a)}};
				\rnode{\coleast}{right}{$e_\i$}{90-\a:\r};
				\draw [thin, dashed, -latex] (90-\a:\r) -- (0,{\r*sin(90-\a)});
			}
			
			\draw [thin, dotted] (90-65:\r) -- (90+14:\r);
			\draw [thin, dotted] (0, {\r*sin(90-156)}) -- ({4*\r*cos(112)} , {4*\r*sin(112) -3*\r});
			\node (a) at ({4*\r*cos(112)} , {4*\r*sin(112) -2.9*\r}) {$v$};

		\end{tikzpicture}
		\label{fig:slice_stage2_beacons}
	}
	\vfill
	\subfloat[Step 2: the \colr{east\_diameter} robots reach their uniform positions.]{
		\begin{tikzpicture}[scale=0.4, transform shape, font = {\LARGE}]
			\def\r{4cm} 
			\def\m{6}
			\def\angle{60/7}

			\draw(0,-\r) -- ({sqrt(12)*\r}, -\r) arc[start angle={30}, end angle={150},radius=\r*4] -- (0, -\r);
			\draw [] (0,0) circle (\r);
			\draw [thin, dotted] (90:\r) -- (270:\r);

			\rnode{\colangle}{above}{}{90+14:\r};

			\foreach \i in {1,..., \m}
				{\vertx{\colvertx}{above}{$U_{\i}$}{ {4*\r*cos(90-\angle*\i)} , {4*\r*sin(90-\angle*\i) -3*\r} };}
				
			\foreach \i in {1,..., 3}
				{\rnode{\colbeacon}{above}{$e_{\i}$}{ {4*\r*cos(95+3*\i)} , {4*\r*sin(95+3*\i) -3*\r} };}
						
			\rnode{\colblockL}{below}{\colr{sliceL}}{-{sqrt(12)*\r}, -\r};
			\rnode{\colblockR}{below}{\colr{sliceR}}{{sqrt(12)*\r}, -\r};
			
			\rnode{\colmedian}{above}{$w_0$}{90:\r};
			
			\foreach \i/\a in {2/80, 3/98, 4/120, 5/134, 6/156}
				{\rnode{\colwest}{left}{$w_\i$}{90+\a:\r};}
				
			\foreach \i/\a in {4/120, 5/134, 6/156}
				{\draw [thin, dashed, -latex] (0, {\r*sin(90-\a)}) -- ({4*\r*cos(90-\angle*\i)} , {4*\r*sin(90-\angle*0.95*\i) -3.1*\r} );
				\rnode{\coleastdiameter}{left}{$e_\i$}{0,{\r*sin(90-\a)}};
			}
			
		\end{tikzpicture}
		\label{fig:slice_stage2_east_ui}
	}
	
\end{center}
\caption{Stage 2, starting with $m=6$  \colr{east} robots. }
\label{fig:slice_stage2}
\end{figure}
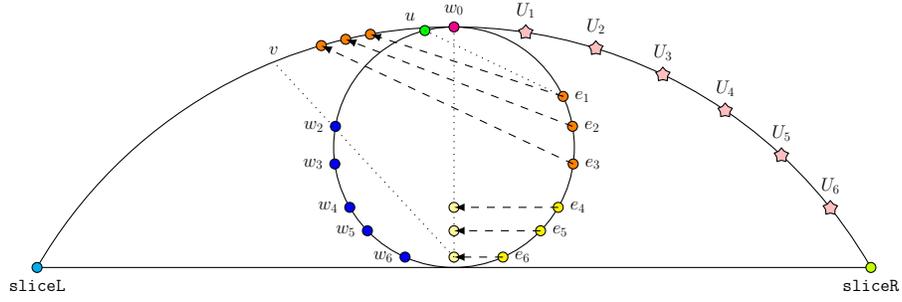
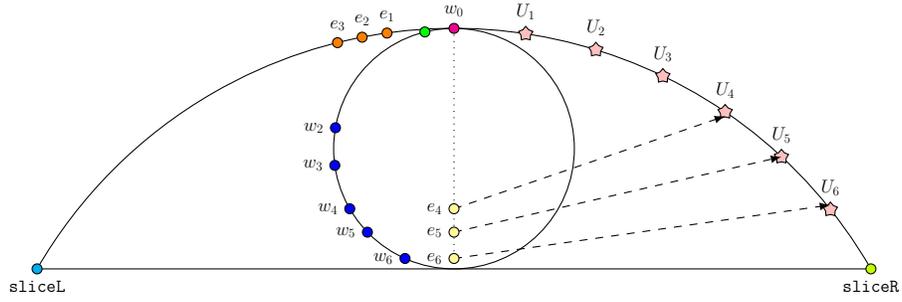

\subparagraph*{Stage 3.}\label{stage:slice_stage3}
In this stage, each \colr{west} robot reaches its uniform position on $\barc_w$ in two steps.
Let $w_2,\dots,w_m$ be the sequence of the \colr{west} robots and let $U'_2, \dots , U'_m$ be the sequence of the uniform positions on $\barc_w$ so that $U'_j$ is intended for $w_j$.
Recall that robot $w_1$ updated its color to \colr{angle} in Stage 1. This robot is destined for position $U'_1$. 

In the first step, the \colr{angle} robot moves perpendicularly to $\diam$ and reaches its symmetric position on $\SC_e$.
Then, $w_2$ reaches a position on $\SC_e$ to form an angle $\frac{\delta}{m}$ with \colr{angle}, setting its color as \colr{anglem}.
After this setting, all the \colr{west} robots migrate from $\SC_w$ to their projection on $\diam$ in constant time by implementing the \BDCP{} strategy ($w_0,w_3,w_{m-1},w_m$ as beacons) and change their color in \colr{west\_diameter}.

In the second step, all the \colr{west\_diameter} robots along $\diam$ compute and reach their uniform position.
Let $w$ be a \colr{west\_diameter} robot on $\diam$. 
It needs $\barc_w$ and its rank to compute its target uniform position.
To obtain its rank, $w$ must exploit the slicing technique of $\SC$ to decode its rank.
So, $w$ recomputes $\SC_w$ and $\SC_e$ through the presence of the \colr{angle} and \colr{anglem} robots, and the line where $\diam$ lies.
Then, $w$ splits $\SC_w$ in $\delta$-slices and determines the slice $\eta_k$ where it originally laid on.
Then, $w$ decodes its rank $j$ inverting the formula $\angle w_0Ow'= k\delta + \frac{j\delta}{m+1}$, where $w'$ is the projection of $w$ on $\SC_w$. 
Recall that $w_0$ is the position of the median robot.
Note that $w$ can determine $m$ through the angle formed by the \colr{angle} and \colr{anglem} robots.
After having obtained its rank $j$ and the arc $\barc_w$ (through the presence of the \colr{sliceR} robot and the \colr{regular} robots on $\barc_e$), $w$ heads towards $U'_j$ and sets its color as \colr{regular} (see \Cref{fig:slice_west_part}).
Lastly, the \colr{angle} and \colr{anglem} robots on $\SC_e$ reach $U'_1$ and $U'_2$ setting their color as \colr{regular}.
At the end of \proc{Slice}, all the uniform positions of the block are occupied by a robot, and no other robot lays on $\SC$ (except for the median one).
See \Cref{lemma:slice_stage3_proof} in \Cref{appendix:slice} for the correctness proof of Stage 3.

Since \proc{Slice} is repeated for each block of $Cir$, the target $n$-gon is achieved. Eventually, all robots turn into \colr{regular}.

\begin{figure}[!h]
\begin{center}

		\begin{tikzpicture}[scale=0.4, transform shape, font = {\LARGE}]
			\def\r{3cm} 
			\def\m{6}
			\def\angle{60/7}

			\draw(0,-\r) -- ({sqrt(12)*\r}, -\r) arc[start angle={30}, end angle={150},radius=\r*4] -- (0, -\r);
			\draw [] (0,0) circle (\r);
			\draw [thin, dotted] (90:\r) -- (270:\r);

			\rnode{\colangle}{above}{}{90-14:\r};
			\rnode{\colanglem}{below right}{}{90-22:\r};		
			
			\foreach \i in {1,..., \m}
				{\rnode{\colregular}{above}{}{ {4*\r*cos(90-\angle*\i)} , {4*\r*sin(90-\angle*\i) -3*\r} };
				\vertx{\colvertx}{above}{$U'_{\i}$}{ {-4*\r*cos(90-\angle*\i)} , {4*\r*sin(90-\angle*\i) -3*\r} };}

			\rnode{\colblockL}{below}{\colr{sliceL}}{-{sqrt(12)*\r}, -\r};
			\rnode{\colblockR}{below}{\colr{sliceR}}{{sqrt(12)*\r}, -\r};
			
			\rnode{\colmedian}{above}{$w_0$}{90:\r};
			
			\foreach \i/\a in { 3/98, 4/120, 5/134, 6/156}
				{\rnode{\colwestdiameter}{right}{$w_\i$}{0, {\r*sin(90+\a)}};
				\draw[thin, dashed, ->] (0, {\r*sin(90-\a)}) -- ({-4*\r*cos(90-\angle*\i)} , {4*\r*sin(90-\angle*\i) -3*\r} );}
			
		\end{tikzpicture}
	
\end{center}
\caption{Stage 3, Step 2: \colr{west\_diameter} robots reach their uniform positions (here $m=6$).}
\label{fig:slice_west_part}
\end{figure}
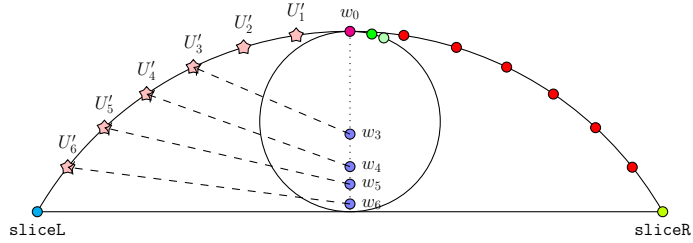

\begin{lemma}[\textbf{Slice}]
\label{lemma:slice}
Starting from \conf{smallcircle}, \conf{regular} is reached 
in $O(1)$ epochs using $O(1)$ colors in the {\asynch} setting, avoiding collisions, guaranteeing robots always operate within $SEC(Conf_{smallcircle})$. Each robot is colored \colr{regular}.
\end{lemma}

\subsection{\proc{Sequential Match} ({\bf Transition} \mycircled{8})}
\label{section:seqmatch}
Let \conf{unisect} be the output configuration of \proc{Split}, where $Cir$ is split into $k$ uniform sectors $\{\Upsilon_i\}_{0\leq i \leq k-1}$.
\proc{Sequential Match} is executed starting from \conf{unisect} when the number of robots $q$ in each uniform sector $\Upsilon_i$ (except for its boundaries) is less than 12.
Remember that, at the end of \proc{Split}, each $\Upsilon_i$ contains a \colr{left} and a \colr{right} robot at $U_1$ and $U_q$, and possibly a \colr{regular} robot at its boundaries.
We assume all the robots that are not \colr{left}-, \colr{right}- or \colr{regular}-colored set their color as \colr{unmatched}.
Now, the closest \colr{unmatched} robot to the related $U_2$ of each uniform sector sets its color as \colr{pre\_matched}.
After the setting of the \colr{pre\_matched} in each of the $k$ uniform sectors, these robots reach the corresponding uniform position $U_2$ and set their color as \colr{matched}.
This routine is repeated until each uniform position $U_j$ for $j=2,\dots, q-1$ is covered by a \colr{matched} robot.
Since the number of movements is upper-bounded by a constant, the whole procedure is performed in constant time.
Eventually, all the robots change their color to \colr{regular}, solving {\UCF}.

See \Cref{sec:correctness_procedure_seqmatch} for the correctness proofs of \proc{Sequential Match}.

\begin{lemma}[\textbf{Sequential Match}]
\label{lemma:seqmatch}
Starting from \conf{unisect}, \conf{regular} is reached 
in $O(1)$ epochs using $O(1)$ colors in the {\asynch} setting, avoiding collisions, guaranteeing all robots perform within $SEC(\config_{unisect})$. 
Each robot is colored \colr{regular}.
\end{lemma}

\section{Putting It All Together}\label{section:alltogether}
\Cref{alg:ucf} summarizes the steps of our {\UCF} algorithm.
Starting from an arbitrary initial configuration \conf{init}, robots arrange in \conf{convex} through the algorithm in \cite{SharmaVT17} which guarantees to avoid collisions and to use $O(1)$ epochs/colors (Transition \mycircled{1}). 
Moreover, no robot moves out from $SEC(\config_{init})$ while reaching $SEC(\config_{convex})$, thus not expanding the initial computational SEC.
Afterward, robots achieve \conf{circle} by reaching the boundaries of $SEC(\config_{convex})$ in $O(1)$ epochs/colors, without colliding or moving out from $SEC(\config_{convex})$ (Transition \mycircled{2}). 
Indeed, $SEC(\config_{circle}) = SEC(\config_{convex})$.
\Cref{lemma:circleformation} summarizes the results of Transitions~{\mycircled{1}-\mycircled{2}}.

If \conf{circle} is a biangular configuration, then robots arrange themselves in a regular $n$-gon through our strategy explained in \Cref{section:componentII} (Transition \mycircled{3}).
\Cref{lemma:biangular} shows that this transition occurs without collisions, in $O(1)$ epochs/colors, without expanding the computational SEC.
Otherwise, if \conf{circle} is a periodic configuration, robots perform a sequence of procedures (corresponding to Transitions \mycircled{4}-\mycircled{8}) to eventually form \conf{regular}. 
Colors are adopted to synchronize robots and make them identify the exact procedure to be executed.
Lemmas \ref{lemma:split}--\ref{lemma:seqmatch} show that such procedures are performed using $O(1)$ epochs/colors, avoiding collisions and within $SEC(\config_{circle})$.
Combining all these results, we have the following theorem, our contribution.

\begin{theorem}[\textbf{Uniform Circle Formation}]\label{theorem:optimal}
Given any \conf{init} of $n$ robots on distinct points on a plane initially colored \colr{off}, the robots reposition to \conf{regular} solving {\UCF} in $O(1)$ epochs using $O(1)$ colors under {\asynch}, avoiding collisions, always performing within $SEC(\config_{init})$.
\end{theorem}

As a corollary, our {\UCF} algorithm asymptotically optimizes the computational time (number of epochs) and the used light colors, and minimizes the computational SEC.

\section{Concluding Remarks}
\label{section:conclusion}
In this paper, we have studied the {\ucf} (\UCF) problem, which is considered an important special case of the fundamental \textsc{Geometric Pattern Formation} problem. 
Specifically, we have investigated {\UCF} under the {\asynch} scheduler in the \emph{luminous-opaque} robot model and presented a $O(1)$-time $O(1)$-color deterministic algorithm, which is asymptotically optimal in both the fundamental metrics, time and color complexities. 
Additionally, it minimizes what we have called the \emph{computational SEC}, i.e. the smallest circular area touched by the swarm during the execution of the algorithm.
The state-of-the-art solution \cite{PattanayakS2024} was optimal in either time complexity or color complexity but not both (\Cref{table:summary}) and the computational SEC was not minimized.

Although our solution optimizes the \emph{color complexity}, we have not focused on minimizing the \emph{exact number} of colors used along the whole algorithm. 
Indeed, the size of our palette can be significantly reduced by reusing some colors in multiple steps along the algorithm.
However, for the sake of clarity, instead of saving on the number of colors, we have used colors with specific, consistent, and meaningful names to help the reader understand their purposes and follow the algorithm steps.
Moreover, minimizing the palette size would result in proving that no ambiguities, deadlocks, or wrong computations can ever happen along the algorithm.
Trustingly, we leave this investigation for further work.

For other future work, it would be interesting to consider non-rigid movements of robots, so that a robot may not reach the computed destination but stop at some point along the Move trajectory. 
It would also be interesting to establish whether a $O(1)$-time solution can be obtained for {\UCF} considering oblivious robots (i.e. without lights) with both rigid and non-rigid movements of robots.

\newpage
\appendix{}

\section{Pseudo-code}\label{appendix:preudocode}

\begin{algorithm}[H]
	\caption{{\UCF} pseudocode}
        \label{alg:ucf}
	 \footnotesize		
        \SetKwInput{Input}{Input}
        \SetKwInput{Output}{Output}
        \SetKwInput{Result}{Final Result}

        \SetKwBlock{CV}{{\sc Complete Visibility}}{end}
        \SetKwBlock{SEC}{{\sc Circle Formation}}{end}
        \SetKwBlock{UT}{{\sc Uniform Transformation}}{end}
        \SetKwBlock{Split}{\proc{Split} [\Cref{sec:procedure_split}]}{end}
        \SetKwBlock{OddBlock}{\proc{Odd Block} [\Cref{sec:procedure_oddblock}]}{end}
        \SetKwBlock{SmallCircle}{\proc{Small Circle} [\Cref{sec:procedure_smallcircle}]}{end}
        \SetKwBlock{Slice}{\proc{Slice} [\Cref{sec:procedure_slice}]}{end}
        \SetKwBlock{Seq}{\proc{Sequential Match} [\Cref{section:seqmatch}]}{end}

        \CV{
            \Input{\conf{init} with $n$ robots in distinct positions on $\reals$}
            Use \cite{SharmaVT17} to arrange robots into a convex pattern\;
            \Output{\conf{convex}}
        }
        \SEC{
            \Input{\conf{convex}}
            $Cir \gets SEC(Conf_{convex})$\;
            Make all robots move to $Cir$ radially\;
        \Output{\conf{circle}}
        }
        
        \UT{
            \Input{\conf{circle}}
            \uIf{\conf{circle} is \conf{regular}}{
            \Result{\conf{regular}}
            } 
            \uElseIf{\conf{circle} is \conf{bingular}}{
                Robots slide along the edges of the exogenous polygon EP [\textbf{\Cref{section:componentII}}]\;
            \Result{\conf{regular}}
            }\Else{
                \conf{circle} is \conf{periodic}\;
                \Split{
                    \Input{\conf{periodic}}
                    Split $Cir$ into $k$ uniform sectors $\Upsilon_0, \dots, \Upsilon_{k-1}$\;
                    Set the chirality of each $\Upsilon_i$ through the \colr{left} and \colr{right} robots\;
                    \Output{\conf{unisect}}
                }
                $q\gets$ number of robots in each $\Upsilon_i$\;
                \uIf{$q\geq 12$}{
                    \OddBlock{
                        \Input{\conf{unisect}}
                        In each $\Upsilon_i$ form an odd-block with the \emph{left} and \emph{right} guards\;
                        Use {\BDCP} to make robots on $\Upsilon_i$ migrate on the block chord $\chord$\;
                        \Output{\conf{oddblock}}
                    }
                    \SmallCircle{
                        \Input{\conf{oddblock}}
                        In each odd-block, spot the inscribed small circle $\SC$\; 
                        Robots migrate from $\chord$ to $\SC$\;
                        Robots on $\SC$ equally distribute on the two halves $\SC_e$ and $\SC_w$\;
                        \Output{\conf{smallcircle}}
                    }
                    \Slice{
                        \Input{\conf{smallcircle}}
                        Robots on $\SC_w$ encode their rank\;
                        Robots on $\SC_e$ and $\SC_w$ reach their target uniform positions\;
                        \Result{\conf{regular}}
                    }
                }\Else{
                    \Seq{
                        \Input{\conf{unisect}}
                        Robots in each $\Upsilon_i$ reach their target uniform position using a sequential scheme\;
                         \Result{\conf{regular}}
                    }
                 }
               
            }
        }	
\end{algorithm}

\section{Correctness Proofs}\label{appendix:correctness_proofs}

\subsection{\proc{Split}}\label{sec:correctness_procedure_split}
We show here all the correctness proofs for \proc{Split} (\Cref{sec:procedure_split}).

\begin{lemma}[\cite{feletti2023journal}]\label{lemma:diam_splitting}
Let us consider a configuration where all $n$ robots lay on the same circle $Cir$.
\begin{itemize}
    \item For {\em odd} $n$, there exists a diameter passing through a robot and dividing $Cir$ into two half-circles, each having $\frac{n-1}{2}$ robots. 
    \item For \emph{even} $n$, there exists a diameter dividing $Cir$ into two halves such that
    $(i)$ either the diameter passes through just one robot and the half-circles have $\frac{n}{2}$ and $\frac{n}{2}-1$ robots, or $(ii)$ the diameter passes through two antipodal robots and the half-circles have $\frac{n}{2}-1$ robots each.
\end{itemize}
\end{lemma}

\begin{proof}
Let $A,B$ be two distinct points on $Cir$, and let $D_A,D_B$ be the respective antipodal points on $Cir$.
We define the \emph{overlap arcs} of $A$ and $B$ as the two equal-length arcs $\wideparen{AD_B}$ and $\wideparen{BD_A}$.
Note that overlap arcs are null iff $A,B$ are antipodal.
Let us assume $n$ is odd, and reason by induction on $n$. 
For $n=1$, the property follows straightforwardly. 
Now, assume the property true for an odd $n>1$, and let us prove the property for $n+2$. 
Let us remove two robots, say $r_1$ and $r_2$, having minimal overlap arcs $\theta$ and $\theta'$, so that $n$ robots are left on $Cir$. 
So, by the induction hypothesis, there exists a diameter $d$ passing through a robot and dividing $Cir$ into two halves, each having $\frac{n-1}{2}$ robots. 
Notice that the only way for $d$ to leave $r_1$ and $r_2$ in the same half-$Cir$ is to pass through the overlap arcs $\theta$ and $\theta'$.
Yet, since $\theta$ and $\theta'$ are minimal, no other robot can sit on these overlap arcs (except for $r_1$ and $r_2$).
So, since $d$ passes through a robot, it clearly cannot leave $r_1$ and $r_2$ on the same half-$Cir$. 
Thus, the result follows.
For even $n$, the proof is similar.
\end{proof}

\begin{lemma}\label{lemma:phi1}
If $|\Phi| =1$, then the configuration is either asymmetric or mirror-symmetric with just one symmetry axis.
\end{lemma}
\begin{proof}
If $|\Phi| =1$, the smallest angular sequence $\hat{\mu}$ starts from a unique robot, say $r$.
There can exist just two cases: either $\hat{\mu}$ can be read just in one orientation (e.g. the clockwise one) or $\hat{\mu}$ can be read in both the orientations.
In the first case, the configuration is asymmetric. 
In the second case, $r$ lies on a mirror-symmetry axis. Since $|\Phi| =1$, no other axes exist in the configuration.
\end{proof}

\begin{lemma}
If $|\Phi| =1$, then we can split $Cir$ into two uniform sectors $\Upsilon_0,\Upsilon_1$.
\end{lemma}
\begin{proof}
By \Cref{lemma:phi1}, the configuration is asymmetric or mirror-symmetric with just one symmetry axis.
In the first case, we can use \Cref{lemma:diam_splitting} to elect the leader robot and the diameter splitting the $Cir$ into two uniform sectors $\Upsilon_0,\Upsilon_1$. 
Indeed, such $Cir$-halves are asymmetric and so chiral.
In the second case, the axis of symmetry splits $Cir$ into two mirrored halves $\Upsilon_0,\Upsilon_1$. Since there is just one axis of symmetry, the two halves are asymmetric, thus guaranteeing the two sectors are chiral.
\end{proof}

\begin{lemma}\label{lemma:leaders_no_axes}
    In \conf{biperiodic}, \colr{leader} robots (i.e. the robots from $\Phi$) do not lay on the axes of symmetry.
\end{lemma}
\begin{proof}
    By contradiction, let $r$ be a \colr{leader} robot lying on an axis of symmetry of the configuration.
    Since $r$ belongs to $\Phi$, the minimal angular sequence $\hat{\mu}$ starts from $r$, in both directions (being $r$ on a mirror-symmetry axis).
    Let $\hat{\mu} = \mu_0\mu_1\mu_2\dots $, and let $r'$ be the left (w.l.o.g.) adjacent robot to $r$, so that from $r'$ starts the angular sequence $\dot{\mu} = \mu_0\mu_0\mu_1\dots$ in the clockwise direction.
    Let $\mu_i$ be the first angle in $\hat{\mu}$ which is different from $\mu_0$ (since the configuration is not regular, such an angle always exists).
    So we can conclude that $\dot{\mu} = \mu_0\mu_0\dots\mu_i\dots$ is lexicographically smaller than $\hat{\mu}=\mu_0\dots\mu_i\dots $, and so $r$ does not belong to $\Phi$.
    Contradiction.    
\end{proof}

\begin{observation}\label{obs:collinearity_circle}
    Consider a circle configuration where the robots lay on the same circle $Cir$.
    If a robot $r$ on position $X$ has to reach a new position $Y$ on $Cir$ traveling along the chord $\overline{XY}$, then $r$ can create a collinearity $\overline{arb}$ with just two robots $a,b$ at a given time during its movement.
    Indeed, $a$ and $b$ must lay on the two distinct arcs cut by the chord $\overline{XY}$, each on a different arc.
\end{observation}

\begin{figure}[h!]
  	\centering
	\subfloat[\colr{pre\_left} (here \emph{light \colblockL}) and \colr{pre\_right} (here \emph{light \colblockR}) robots have to reach the uniform positions and thus to fix the uniform sectors.]{
			\centering
                \begin{tikzpicture}[scale=0.5, transform shape, font = {\LARGE}]
                    \def\r{\ray*1.2}
                    \def\p{1.5}
                    \def\a{360/24}
                    \draw [draw=none] (0,0) circle (1.1*\r);
                    \draw [very thick] (0,0) circle (\r);
                    \draw[dotted] (90:\r) -- (270:\r); 
            
                    \foreach \i in {0,1} 
                    {
                        \draw[] (90*\i+\a*0.5:\r) -- (180 + 90*\i+\a*0.5:\r); 
                        \draw[] (90*\i-\a*0.5:\r) -- (180 + 90*\i-\a*0.5:\r); 
                    }
            
                    \node (a) at (90:\r*0.7) {$\frac{2\pi}{n}$};
            
                    \rnode{\collead}{above left}{}{90+\p*\a:\r};
                    \rnode{\collead}{above right}{}{90-\p*\a:\r};
            
                    \rnode{\collead}{above left}{}{270+\p*\a:\r};
                    \rnode{\collead}{above right}{}{270-\p*\a:\r};

                    \draw[dashed,-latex] (180-35:\r) -- (90+\a*0.5:\r);
                    \rnode{\colblockL}{below left}{}{{90+\p*\a}:\r*0.9};
            
                    \foreach \i[evaluate=\i as \x using {Mod(\i,2)},
                                evaluate=\i as \z using {2*\x-1}] in {0,1,2,3} 
                    {
                        \rnode{\coldefault}{below right}{}{{90*(\i+\x)+30*\z}:\r};
                        \rnode{\coldefault}{below right}{}{{90*(\i+\x)+20*\z}:\r};
                        \rnode{\colblockL!50}{below left}{}{{90*(\i+\x)+37*\z}:\r};
                        \rnode{\colblockR!50}{below left}{}{{90*(\i+\x)+13*\z}:\r};
                        \draw[dotted] (90*\i:\r) -- (0:0); 
                    }
                    \node (a) at (-38:\r*1.15) {$r$};

                \end{tikzpicture}
                
        }
    \hspace{\textwidth/4}
	\subfloat[\colr{left} (\emph{\colblockL}) and \colr{right} (\emph{\colblockR}) robots fix the uniform sectors.]{

			\begin{tikzpicture}[scale=0.5, transform shape, font = {\LARGE}]
                    \def\r{\ray*1.2}
                    \def\p{1.5}
                    \def\a{360/24}
                    \draw [draw=none] (0,0) circle (1.1*\r);
                    \draw [very thick] (0,0) circle (\r);
                    \draw[dotted] (90:\r) -- (270:\r); 
            
                    \foreach \i in {0,1} 
                    {
                        \draw[] (90*\i+\a*0.5:\r) -- (180 + 90*\i+\a*0.5:\r); 
                        \draw[] (90*\i-\a*0.5:\r) -- (180 + 90*\i-\a*0.5:\r); 
                    }
            
                    \node (a) at (90:\r*0.7) {$\frac{2\pi}{n}$};
            
                    \rnode{\collead}{above left}{}{90+\p*\a:\r};
                    \rnode{\collead}{above right}{}{90-\p*\a:\r};
            
                    \rnode{\collead}{above left}{}{270+\p*\a:\r};
                    \rnode{\collead}{above right}{}{270-\p*\a:\r};
            
                    \foreach \i[evaluate=\i as \x using {Mod(\i,2)},
                                evaluate=\i as \z using {2*\x-1}] in {0,1,2,3} 
                    {
                        \rnode{\coldefault}{below right}{}{{90*(\i+\x)+30*\z}:\r};
                        \rnode{\coldefault}{below right}{}{{90*(\i+\x)+20*\z}:\r};
                        \rnode{\colblockR}{below left}{}{{90*(\i+\x)+\a*0.5*\z}:\r};
                        \rnode{\colblockL}{below left}{}{{90*(1+\i+\x)-\a*0.5*\z}:\r};
                        \draw[dotted] (90*\i:\r) -- (0:0); 
                    }

                \end{tikzpicture}
        }
\caption{\proc{Split} in a biperiodic configuration with $|\Phi|=4$ uniform sectors to be fixed. The moving \colr{to\_left} robot hides the \colr{leader} (here \emph{\collead}) robot from $r$. However, $r$ can compute its target point even in the case of collinearities, thus fixing the uniform sectors in constant time.}
\label{fig:split_collinearity}
\end{figure}
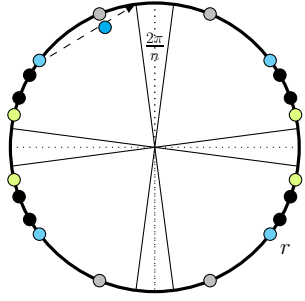
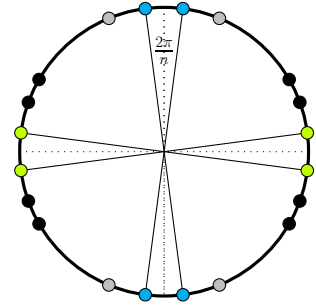

\begin{lemma}\label{lemma:to_left_compute_destination}
    In the case of a biperiodic configuration, a \colr{pre\_left} (\colr{pre\_right}, resp.) robot can always recover its destination point even if other \colr{to\_left} (\colr{to\_right}, resp.) robots are moving and hiding \colr{leader} robots.
\end{lemma}
\begin{proof}
    Let $r$ be a \colr{pre\_left} robot on the uniform sector $\Upsilon_i$, activated during the movement of a \colr{to\_left} robot (see \Cref{fig:split_collinearity}).
    Since no robot moves in $\Upsilon_i$, $r$ can see at least a \colr{leader} and a \colr{pre\_right} robot.
    If the uniform sectors are also defined by \colr{regular} robots (\Cref{subfig:unisplit_biper_odd}), then $r$ can see all of them (in fact, no robot moves out from its original uniform sector, and so no robot hides \colr{regular} robots from $r$).
    Thus, in presence of \colr{regular} robots, $r$ can determine the boundaries of the uniform sectors (i.e. the points where \colr{regular} robots lay on and the middle points on $Cir$ between any pair of consecutive \colr{regular} robots) and the cardinality $n$ of the swarm by multiplying the number of robots inside its sector $\Upsilon_i$ with the number of uniform sectors (note that $r$ has complete visibility of $\Upsilon_i$ since no robot is moving within it).
    
    If $r$ cannot see any \colr{regular} robots, then the configuration originally was biperiodic and no robots lay on the boundaries of the uniform sectors (\Cref{subfig:unisplit_biper}).
    In this case, $r$ can spot the \colr{pre\_right} and the \colr{leader} robot belonging to $\Upsilon_i$ by simply selecting the closest \colr{pre\_right} and the closest \colr{leader} robot to $r$ (note that \colr{pre\_right} and \colr{leader} robots lay on symmetrical positions in adjacent sectors).
    Following the \colr{leader}-\colr{pre\_right} orientation of $\Upsilon_i$, let us consider the adjacent uniform sector, say $\Upsilon_{i+1}$.
    Let $g$ be the closest \colr{leader} or \colr{pre\_right} robot which $r$ sees following this direction.
    Of course, $g$ must belong to $\Upsilon_{i+1}$. 
    In fact, since within $\Upsilon_{i+1}$ at most one robot is moving, at least one \{\colr{leader}, \colr{pre\_right}\}-colored robot cannot be hidden from $r$ (by \Cref{obs:collinearity_circle}).
    If $g$ is \colr{pre\_right}, then $r$ can determine the boundary between $\Upsilon_i$ and $\Upsilon_{i+1}$ as the middle point on $Cir$ between $g$ and the position of the \colr{pre\_right} robot belonging to $\Upsilon_i$.
    Otherwise, it means that the \colr{pre\_right} robot of $\Upsilon_{i+1}$ is hidden by a moving robot which is visible to $r$.
    Thus, $r$ can compute the exact position of the hidden \colr{pre\_right} robot and so determine the boundary between $\Upsilon_i$ and $\Upsilon_{i+1}$.
    The same strategy is applied searching the \colr{leader} robot in $\Upsilon_{i-1}$.
    Being the uniform sectors equal in amplitude and number of robots, $r$ can easily reconstruct the boundaries of all the other sectors and so the number of robots in the swarm.
    
    The proof for the \colr{pre\_right} robots is similar.
\end{proof}

\subsection{\proc{Odd Block}}\label{sec:correctness_procedure_oddblock}
We show here all the correctness proofs for \proc{Odd Block} (\Cref{sec:procedure_oddblock}).

\begin{lemma}\label{lemma:to_x_compute_destination}
    Consider a configuration \conf{unisect} with $k$ uniform sectors.
    Let $r$ be a \colr{pre\_x} robot (being $\colr{x}\in\{\colr{mid, median, blockR}\}$) lying on the arc of a sector $\Upsilon_i$.
    Within $\Upsilon_i$, no robot is moving.
    Then, $r$ can always determine the boundaries of all the uniform sectors and the cardinality of the swarm $n$ even if at most one robot \colr{to\_x} is moving within any other sector $\Upsilon_{j\neq i}$ to reach a new position on $Cir$.
\end{lemma}
\begin{proof}
    The proof is similar to that of \Cref{lemma:to_left_compute_destination}.
    Since no robot moves in $\Upsilon_i$, $r$ can see at least a \colr{left} and a \colr{right} robot.
    If \conf{unisect} is also defined by \colr{regular} robots, then $r$ can see all of them (in fact, no robot moves out from its original uniform sector, and so hides \colr{regular} robots from $r$).
    According to the number of \colr{regular} robots and the number of \{\colr{left}, \colr{right}\}-colored robots between two \colr{regular} robots, $r$ can determine if the configuration originally was uniperiodic (\Cref{subfig:unisplit_uniper}) or biperiodic (\Cref{subfig:unisplit_odd,subfig:unisplit_biper,subfig:unisplit_biper_odd}).
    If there are \colr{regular} robots, $r$ can determine $n$ by the central angle $2\pi/n$ between the \colr{regular} and the related \colr{left} or \colr{right} robot, and the positions of all the uniform sectors.
  
    If $r$ cannot see any \colr{regular} robots, then the configuration originally was biperiodic and no robots lay on the boundaries of the uniform sectors (\Cref{subfig:unisplit_biper_odd}).
    In this case, $r$ can spot the pair of \colr{left}-\colr{right} robots belonging to $\Upsilon_i$.
    In fact, $r$ selects the closest \colr{left}-\colr{right} pair which cuts a sector where no robot is moving.
    Following the \colr{left}-\colr{right} orientation of $\Upsilon_i$, let us consider the adjacent uniform sector $\Upsilon_{i+1}$, which is located on the right w.r.t. $\Upsilon_i$.
    Let $g$ be the closest \colr{left} or \colr{right} robot which $r$ sees following the right direction.
    Of course, $g$ must belong to $\Upsilon_{i+1}$. 
    In fact, since within $\Upsilon_{i+1}$ at most one \colr{to\_x} robot is moving, at least one \{\colr{left}, \colr{right}\}-colored robot cannot be hidden from $r$ (by \Cref{obs:collinearity_circle}).
    If $g$ is \colr{right}, then $r$ can determine the boundary between $\Upsilon_i$ and $\Upsilon_{i+1}$ as the middle point on $Cir$ between $g$ and the position of the \colr{right} robot belonging to $\Upsilon_i$.
    Otherwise, it means that the \colr{right} robot of $\Upsilon_{i+1}$ is hidden by a moving robot which is visible to $r$.
    Thus, $r$ can compute the exact position of the hidden \colr{right} robot and so determine the boundary between $\Upsilon_i$ and $\Upsilon_{i+1}$.
    The same strategy is applied searching the \colr{left} robot in $\Upsilon_{i-1}$.
    Being the uniform sectors equal in amplitude and number of robots, $r$ can easily reconstruct the boundaries of all the other sectors and so the number of robots in the swarm.
\end{proof}

\begin{corollary}
    In \proc{Odd Block} the setting of the \colr{blockR}, \colr{median}, and \colr{mid} takes $O(1)$ epochs.
\end{corollary}
\begin{proof}
    By \Cref{lemma:to_x_compute_destination}, we know that each \colr{pre\_x} robot $r$ (being $\colr{x}\in\{\colr{mid, median, blockR}\}$) computes its target position (where it will assume the color \colr{x} once stopped there and reactivated) even if other \colr{to\_x} robots are moving to reach their target positions in the other uniform sectors.
    It follows that the time taken for setting the \colr{x} robots is independent of the number of uniform sectors.
\end{proof}

\begin{lemma}
Given an odd block built inside a uniform sector $\Upsilon_i$, the midpoint $M$ of the arc of $\Upsilon_i$ lies on the block arc $\Gamma$.
\end{lemma}
\begin{proof}
Let $U_1$ and $U_{2l+1}$ be the uniform positions where the guards of the block lay on (i.e. they are the endpoints of $\Gamma$), and let $B_i$ and $B_{i+1}$ be the boundaries of $\Upsilon_i$. 
To prove the statement, we need to prove that $\angle U_1O U_{2l+1} \geq \frac{\angle B_iOB_{i+1}}{2}$ (where $O$ is the center of $Cir$), i.e. the amplitude of the odd block is always not smaller than the amplitude of half $\Upsilon_i$.
By construction we know that $\angle B_iOU_1 \leq \frac{2\pi}{n}$ (for the left guard), while $\angle U_{2l+1}OB_{i+1} \leq \frac{2\pi}{n} + \frac{2\pi}{n}$ (for the right guard and for the padding robot).
Thus, $\angle U_1O U_{2l+1} \geq \angle B_iOB_{i+1} - \frac{6\pi}{n}$.
To prove that $\angle U_1O U_{2l+1} \geq \frac{\angle B_iOB_{i+1}}{2}$, it is sufficient to verify that 
$$\angle B_iOB_{i+1} - \frac{6\pi}{n} \stackrel{?}{\geq} \frac{\angle B_iOB_{i+1}}{2}$$
$$\angle B_iOB_{i+1}  \stackrel{?}{\geq} \frac{12\pi}{n} = 6\cdot\frac{2\pi}{n}.$$
Since each $\Upsilon$ contains $q\geq12$ robots (laying not on its boundaries), we know that $\angle B_iOB_{i+1} \geq 12\cdot\frac{2\pi}{n}$, thus proving our statement.
\end{proof}

\begin{observation}\label{obs:chords}
Two chords of the same circle can \emph{(i)} have no points in common, \emph{(ii)} intersect in only one point, or \emph{(iii)} be coincident.
\end{observation}

\begin{lemma}\label{lemma:odd_block_phase2}
During Stage 2 of \proc{Odd Block}, all the \colr{out\_chord} robots can reach $\chord$ in $O(1)$ epochs using {\BDCP}.
\end{lemma}
\begin{proof}
    We verify that all the input assumptions listed in \Cref{alg:bdcp} hold to correctly implement {\BDCP} in $O(1)$ time.
    \begin{enumerate}
	\item $\chord$ is the target $2$-curve;
        \item 4 beacons are originally set to fix $\chord$ (the guards and the two \textit{ad hoc} beacons);
        \item the \colr{out\_chord} robots are the waiting robots which have to reach $\chord$, so that
            \begin{enumerate}
                \item they lay on the arcs $\wideparen{B_iU_1}$ and $\wideparen{U_{2l+1}B_{i+1}}$ (except for their endpoints), so they are external to $\chord$, belonging to the same half-plane delimited by $\chord$;
                \item at the beginning of Stage 2, no robot is moving in $\Upsilon_i$, and all robots of the uniform sector lay on the arc of $\Upsilon_i$ or on $\chord$. 
	           So, any \colr{out\_chord} robot can see the 4 original beacons on $\chord$.
                \item they have to reach $\chord$
                    \begin{itemize}
                        \item by construction, the two internal beacons lay on $\chord$ so that the target points of each \colr{out\_chord} robots are located between them;
                        \item each \colr{out\_chord} $r$ robot must reach $\chord$ traveling along the trajectory $\overline{rM}$ where $M$ is the position of the \colr{mid} robot. Of course, the trajectories do not intersect each other; 
                        \item such trajectories intersect $\chord$ in just one point (see Observation \ref{obs:chords});
                        \item by construction, the beacons have been placed on $\chord$ so that they do not create collinearities with the \colr{mid} robot and the \colr{out\_chord} robots.	This guarantees that any \colr{out\_chord} robot can see the \colr{mid} robot at the beginning of Stage 2. Moreover, throughout the whole execution of {\BDCP}, \colr{out\_chord} robots travel along distinct concurrent trajectories that are generated by the origin (the \colr{mid} robot), they never create collinearities with the \colr{mid} robot. 
                        As follows, any \colr{out\_chord} robot can reconstruct its path towards $\chord$.
                        For the same reason, the moving robots never collide since each \colr{out\_chord} trajectory does not intersect with any other trajectory. 
                    \end{itemize}
            \item each \colr{out\_chord} robot sets its color as \colr{chord} once on $\chord$, playing the role of a new beacon fixing the chord $\chord$. 
            \end{enumerate}
    \end{enumerate}
    All these assumptions guarantee all \colr{out\_chord} robots reach $\chord$ in $O(\log{2}) = O(1)$ epochs without colliding.
\end{proof}

\begin{lemma}
During Phase 3 of \proc{Odd Block}, all the \colr{in\_chord} robots can reach $\chord$ in $O(1)$ epochs using {\BDCP}.
\end{lemma}
\begin{proof}
    The proof is similar to the proof of \Cref{lemma:odd_block_phase2}, so we report just the differences.
	The presence of the original 4 beacons, plus the \colr{chord} robots, fix the target curve $\chord$.
    If no such beacons exist, then the robots closest to $U_1$ and $U_{2l+1}$ move to the $\chord$ to act as beacons.
	By construction, the beacons have been set in order to be on the left and on the right of the target positions of the \colr{in\_chord} robots on $\chord$.
	Let $r$ be a \colr{in\_chord} which sees at least 2 robots on $\chord$, and let $F'$ be its projection on $\chord$.
	Let $Z$ be the position of the closest \colr{in\_chord} robot (still on $\Gamma$ or in transit towards $\chord$) and let $Z'$ be its projection on $\chord$.
	Then $r$ always can compute its target position $F''$ on $\chord$ and so its path towards $\chord$ ($F''=F'$ if no robot lies on $F'$, otherwise $|\overline{F'F''}| =|\overline{F'Z'}|/ 3$).
	According to this rule, two adjacent \colr{in\_chord} robots will always compute two paths without colliding.
	The other assumptions hold for the same reasons as in \Cref{lemma:odd_block_phase2}.
\end{proof}

\subsection{\proc{Small Circle}}
\label{appendix:smallcircle}
We show here all the correctness proofs for each step in \proc{Small Circle} (\Cref{sec:procedure_smallcircle}).

\begin{lemma}\label{lemma:safe_unsafe_arcs}
When all the robots have migrated from $\chord$ to $\SC$, then they can always see the right guard (if they lay on $\SC_e$) or the left guard (if they lay on $\SC_w$).
\end{lemma}
\begin{proof}
Let $l$ be the line passing through the median robot and the right guard.
This line splits $\SC_e$ into two arcs: the upper one (which we call \emph{the unsafe} arc) and the lower one (which we call \emph{the safe} arc).
According to Stage 1 of \proc{Small Circle}, all the robots on $\SC_e$ are located in the safe arc cut by line $l$ (see \Cref{fig:slice_init_config}).
Let $l'$ be any other line passing through the right guard and a point of $\SC_e$ (different from the median robot).
If $l'$ intersects the unsafe arc, then it intersects it into two points.
Otherwise, $l'$ intersects the safe arc into only one point.
This means that no robot on $\SC_e$ creates collinearity with another robot on $\SC_e$ and the right guard.
The same holds considering the left guard and $\SC_w$.
\end{proof}

\begin{lemma}\label{lemma:smallcircle_stage1_proof}
During Stage 1 of \proc{Small Circle}, the following statements hold: 
\begin{enumerate}
	\item each \colr{chord} robot on $\chord$ can see the \colr{scMedian} robot even when other robots are moving to $\SC$.
	Moreover, the movements at Stage 1 are collision-less.
	\item If the right guard (resp. left guard) sets its color as \colr{scR\_complete} (resp. \colr{scL\_complete}), no \colr{chord} or \colr{to\_smallcircle} robot exists in the right (resp. left) half-block.
\end{enumerate}
\end{lemma}

\begin{proof}
We prove each statement:
\begin{enumerate}
	\item At the beginning of Stage 1, all the robots (except for the \colr{scMedian} one) are located on the block chord $\chord$ (and not inside the block), so no robot obstructs the visibility of the \colr{chord} robot.
Moreover, since the \colr{chord} robots travel along distinct concurrent trajectories that are generated by the origin (the \colr{scMedian} robot), they never create collinearities with the \colr{scMedian} robot. 
For the same reason, the moving robots never collide.

	\item Let us prove for $\SC_e$ (the same holds for $\SC_w$).
		By \Cref{lemma:safe_unsafe_arcs}, robots on $\SC_e$ do not create collinearities with the right guard.
		Moreover, a \colr{to\_smallcircle} robot cannot be hidden by other robots on $\SC_e$ from the right guard.
		So, if a \colr{to\_smallcircle} robot exists, the right guard can see it.
		At the same time, if at least a \colr{chord} robot exists on $\chord$, then the right guard can see.
	
\end{enumerate}
\end{proof}

\begin{lemma}\label{lemma:smallcircle_stage2_proof}
During Stage 2 of \proc{Small Circle}, the following statements hold: 
\begin{enumerate}
	\item Each \colr{smallcircle\_complete} robot on $\SC$ can correctly detect if it belongs on $\SC_e$ or $\SC_w$ and so set its color as \colr{smallcircle\_east} or \colr{smallcircle\_west}, respectively.
	\item The position shift of the \colr{smallcircle\_west} robots can be executed without colliding and in $O(1)$ time.
        \item The \colr{smallcircle\_west} robots migrate from $\SC_w$ to $\SC_e$ in $O(1)$ using {\BDCP}.
\end{enumerate}
\end{lemma}
\begin{proof}
We prove each statement:
\begin{enumerate}
	\item By \Cref{lemma:safe_unsafe_arcs}, all the \colr{smallcircle\_complete} robots on $\SC_e$ (resp. $\SC_w$) see the related guard. 
	Thus, they can detect if they belong to the right/left side of $\SC$, and update their color properly.
	\item The new positions of \colr{smallcircle\_west} robots can be easily computed so that each robot shifts on $\SC_w$ of a little distance downward so that it neither collides with nor goes beyond the adjacent robot. 
	Such movements always guarantee complete visibility among the robots within $\SC$.
	Note that a \colr{smallcircle\_west} robot can always recompute correctly $\SC$ even when the other robots are shifting; in fact, there are always at least three robots still on $\SC$: the robot itself, the median robot, and at least one \colr{smallcircle\_east} robot on $\SC_e$ (otherwise, no shift is needed).
	\item  We verify that all the input assumptions listed in \Cref{alg:bdcp} hold to correctly implement {\BDCP} in $O(1)$ time.
    \begin{enumerate}
	\item $\SC_e$ is the target $3$-curve;
        \item 6 robots (selected among the \colr{smallcircle\_east} or \colr{smallcircle\_west} robots) can be set on $\SC_e$ so that they play the role of the left and right beacons (using a proper color). This setting can be achieved in constant time;
        \item the \colr{smallcircle\_west} robots are the waiting robots which have to reach $\SC_e$. We can guarantee that:
            \begin{enumerate}
                \item being on $\SC_w$, they are external to $\SC_e$, and they lay on the convex region delimited by $\SC_e$;
                \item at the beginning of this step, no robot is moving within $\SC$, so the waiting robot on $\SC_w$ can see all the robots on $\SC_e$;
                \item they have to reach $\SC_e$
                    \begin{itemize}
                        \item to their projections on $\SC_e$. Being the beacons originally the ``leftmost'' and the ``rightmost'' robots on $\SC_e$, the robots on $\SC_w$ will migrate between the two groups of beacons;
                        \item being parallel, the robots' trajectories do not intersect each other; 
                        \item such trajectories are chords connecting $\SC_w$ to $\SC_e$, thus they intersect $\SC_e$ in just one point;
                        \item if a \colr{smallcircle\_west} robot on $\SC_w$ can reconstruct $\SC$, then it can reconstruct its trajectory (looking at the left guard, a \colr{smallcircle\_west} robot can detect the position of $\diam$ and so its trajectory towards $\SC_e$);
                    \end{itemize}
            \item each \colr{smallcircle\_west} robot becomes a new beacon (setting its color as \colr{smallcircle\_east}) once on $\SC_e$.
            \end{enumerate}
    \end{enumerate}
    All these assumptions guarantee all \colr{smallcircle\_west} robots reach $\SC_e$ in $O(\log{3}) = O(1)$ epochs without colliding.
\end{enumerate}
\end{proof}

\begin{lemma}\label{lemma:smallcircle_stage3_proof}
During Stage 3 of \proc{Small Circle}, the following statements hold: 
\begin{enumerate}
	\item Each \colr{west} robot on $\SC_e$ correctly detects its upper-adjacent \colr{east} robot, even if other robots are moving.
	\item Each \colr{west} robot on $\SC_e$ correctly and safely reaches its target position on $\SC_w$.

\end{enumerate}
\end{lemma}

\begin{proof}
We prove each statement:
\begin{enumerate}
	\item Let $w$ be a \colr{west} robot on $\SC_e$. 
		Thanks to the presence of the right guard (always visible to $w$), $w$ understands it has to migrate on the other side.
		Moreover, $w$ can always see its two adjacent \colr{east} robots even when other robots are migrating to $\SC_w$ (the first robot on $\SC_e$ has an \colr{east} robot and the median robot as adjacent).
		So $w$ recomputes $\SC$.
		Now, let us suppose the median robot is hidden from $w$. 
            Thus, $w$ must detect which of its two \colr{east} adjacent robots is the upper-adjacent one.
		From the right guard, two chords can be elected as $\chord$. 
		However, by the presence of other robots of other blocks, $w$ can always detect which is the correct block chord, and so it can establish the upper-down orientation of the block.
		
	\item Again, let $w$ be a \colr{west} robot on $\SC_e$. 
		As proved above, $w$ computes $\SC$ by the presence of at least three robots on it, and $\diam$ by the presence of the right guard.
		Once obtained $\SC_w$ and once selected its upper-adjacent robot $u$, $w$ heads to the projection of $u$ on $\SC_w$.
		All the trajectories are disjoint and so no collision can occur.
\end{enumerate}
\end{proof}

\subsection{\proc{Slice}}
\label{appendix:slice}
We show here all the correctness proofs for each step in \proc{Slice} (\Cref{sec:procedure_slice}).

\begin{lemma}\label{lemma:slice_stage1_proof}
During Stage 1 of \proc{Slice}, the following statements hold: 
	\begin{enumerate}
		\item\label{item:slice_phase1_most_two} 
		At the beginning of Stage 1, each slice arc on $\SC_w$ contains at most two robots: in the case of exactly two robots, both of them must lay on the endpoints of the slice arc;
		
		\item\label{item:slice_phase1_chosen_one} 
		Each slice is chosen by at most one \colr{west} robot where it will encode its rank;
		
		\item\label{item:slice_phase1_no_cross} 
		During the movements in Stage 1, each \colr{west} robot $w_j$ will move to a new point $w'_j$ which is strictly contained in its original slice;
		
		\item\label{item:slice_phase1_end_phase} 
		At the end of Stage 1, each slice arc of $\SC_w$ contains at most one robot, except for the slice $\eta_0$ which contains two robots (the \colr{sliceMedian} robot and the \colr{angle} one) on its endpoints.
		
		\item\label{item:slice_phase1_no_collisions} 
		Movements at Stage 1 never create collisions or collinearities among robots on $\SC$.
	\end{enumerate}
\end{lemma}

\begin{proof}
We prove each statement:
	\begin{enumerate}
		\item
			By construction, each slice has an amplitude of $\delta$, where $\delta$ is the minimum angle between any two robots on the circle $\SC_w$.
			Let $\eta$ be a $\delta$-slice of $\SC_w$.
			If three robots lay on the arc of $\eta$, then they define a smaller angle than $\delta$ (contradiction).
			The same contradiction results in the case of two robots where one of them does not lay on the endpoint of the $\eta$ arc.
			Note that the remaining slice of amplitude $<\delta$ which can exist between the last $\delta$-slice and the lower endpoint of $\SC_w$ does not contain any robot (otherwise $\delta$ would not be the minimum angle).
		\item
			If a slice does not contain any robots, then it will be chosen by no robots.
			If a slice contains just one \colr{west} robot, then it will choose its slice to encode its rank.
			If a slice $\eta_k$ contain two \colr{west} robots, say $w_j,w_{j+1}$, one per each arc endpoint, then $w_j$ chooses $\eta_{k-1}$ whereas $w_{j+1}$ chooses $\eta_k$.
			Since $\eta_{k-1}$ contains either just $w_j$ or both $w_{j-1}$ and $w_j$ (by \Cref{item:slice_phase1_most_two}), $\eta_{k-1}$ will be chosen just by $w_j$ (in fact, if $w_{j-1}$ is not \colr{angle}, then it will choose $\eta_{k-2}$).
		\item 
			Let $\eta_k$ be the chosen slice for the rank encoding of $w_j$ (i.e. the slice where $w_j$ lays on, possibly on one of its endpoints).
			Since $w_j$ reaches a new point $w'_j$ such that $\angle w_0Ow'_j= k\delta + \frac{j\delta}{m+1}$, it is sufficient to prove that $w'_j$ belongs to $\eta_k$ (endpoints excluded).
			Such a property is easily proved since $\frac{j\delta}{m+1}< \delta$ (in fact, $j<m+1$) and since $j\neq 0$ (\colr{west} robots have rank $j=2,\dots, m$).
		\item 
			By \Cref{item:slice_phase1_most_two}, each \colr{west} robot travels in a new position which is strictly contained in its chosen slice.
			By \Cref{item:slice_phase1_chosen_one}, each slice is chosen by at most one robot.
			These two results ensure each slice arc of $\SC_w$ contains at most one robot.
			To conclude, $\eta_0$ is the only slice in $\SC_w$ containing two robots at the end of Stage 1.
		\item 
			Movements of \colr{west} robots are performed in separated slices, thus never creating collisions with the other robots.
			Let $w_j$ be a \colr{west} robot which has to travel in the slice $\eta_k$ to reach the point $w'_j$.
			This robot creates collinearity with another robot on $\SC$ only if both the arcs cut by the trajectory $\overline{w_jw'_j}$ contain at least a robot.
			However, since in the minor arc $\wideparen{w_jw'_j}$ (i.e. the arc strictly contained in $\eta_k$) no robots are contained except for $w_j$ itself, no collinearities can occur.
\end{enumerate}
\end{proof}

\begin{lemma}\label{lemma:slice_stage2_proof}
During Stage 2, the following statements hold: 
	\begin{enumerate}
	\item
			{\bf [Step 1]} The migration of the \colr{east} robots $\{e_j\}_{4\leq j \leq m}$ on $\diam$ is always possible in one epoch without collisions.
			
	\item
			{\bf [Step 1]} The beacons $e_1,e_2$ and $e_3$ correctly and safely migrate on the safe arc $\wideparen{uv}$.

		\item{}
			{\bf [Step 2]} An \colr{east\_diameter} robot on $\diam$ can recompute the slices on $\SC_e$ and $\SC_w$.
			
		\item{}
			{\bf [Step 2]} The rank $j$ computed by an \colr{east\_diameter} robot on $\diam$ coincides with its original rank on $\SC_e$, for $2\leq j \leq m $.
			
			\item{}\label{item:slice_phase2_see_beacons}
				{\bf [Step 2]} Every \colr{east\_diameter} robot on $\diam$ can see the three beacons.
			
		\item{}
			{\bf [Step 2]} An \colr{east\_diameter} robot $e$ on $\diam$ can instantly recompute its uniform position $U_j$ on $\barc_e$.
			
\item\label{item:slice_phase2_beacons_arc}
			{\bf [Step 2]} Beacons $e_1,e_2$ and $e_3$ on $\barc_w$ safely reach their target positions on $\barc_e$ in constant time.

	\end{enumerate}
\end{lemma}
\begin{proof}
We prove each statement:
	\begin{enumerate}
		\item
		Let $e$ be an \colr{east} robot on $\SC_e$ at the beginning of Stage 2 (after the setting of the beacons on $\SC_e$).
		At most $m-4$ \colr{to\_east\_diameter} robots are moving to $\diam$, thus possibly obstructing the visibility of $e$. 
		However, since each moving robot can hide only one other robot on $\SC$ from $e$, and since there are at least $m+5$ robots on $\SC$ (included $e$), then $e$ can always reconstruct $\SC$.
		By \Cref{lemma:safe_unsafe_arcs}, $e$ can see \colr{sliceR} and so it can recompute the supporting line of $chord$, and so $\diam$.
		Thus, $e$ directly travels perpendicularly towards $\diam$, without waiting. 
		Since all the trajectories are perpendicular to the same line, the migration of the \colr{east} robots is accomplished without colliding.
		\item 
		Let $b$ be a \colr{beacon} robot on $\SC_e$ which sees all the \colr{east\_diameter} robots on $\diam$.
		Thanks to the presence of the \colr{west} robots, $e$ can recompute $\SC$, and so it can establish whether a \colr{beacon} is still moving or if the configuration is static.
		For the sake of simplicity, we make our beacons reach $\barc_w$ following a sequential scheme: thus, firstly $e_3$ reaches a deterministic point on $\wideparen{uv}$, while $e_1$ is the last beacon to move.
		As soon as $b$ sees a static configuration (within its block) and understands it is its turn, it can easily compute the safe arc $\wideparen{uv}$ and move to a deterministic position on $\wideparen{uv}$, thus avoiding collisions with the other robots.
Since the number of beacons for a block is constant, their migration takes $O(1)$ time.

		\item
			Since at Stage 2 no robot moves within the west half-circle given by $\SC_w$, an \colr{east\_diameter} robot $e$ on $\diam$ can always see all the \colr{west} robots and the \colr{angle} one, thus allowing it to determine $\SC$ and the angle $\delta$.
			Moreover, $e$ can easily recompute $\diam$ by the presence of the adjacent \colr{east\_diameter} or \colr{sliceMedian} robot.
			So, $e$ can split $\SC_w$ and $\SC_e$ in $\delta$-slices and determine the slice $\eta'_k$ and $\eta_k$ which contain its projection on $\SC$.
		\item
			Thanks to \Cref{lemma:slice_stage1_proof}, each \colr{west} robot $w_j$ (except for $w_1$ which became the \colr{angle} robot) has moved in its own original slice during Stage 1, thus maintaining the one-to-one rank association with the \colr{east} side.
			
		\item
			See \Cref{fig:slice_stage2_beacons}.  
			The beacons lay on the safe arc $\wideparen{uv}$, which was defined by the two lines: the first one which connects the initial position of $e_1$ and $w_1$ (the \colr{angle} robot), and the second one which connects the position of $e_m$ and $w_2$.
			Let $\diam'$ be the segment of $\diam$ cut by such lines.
			By construction, all the \colr{east\_diameter} lay on $\diam'$ ($e_m$ lay on its lower endpoint).
			Let $x$ be an internal point of the safe arc $\wideparen{uv}$ and let $y$ be a point of $\diam'$.
			Then, the line connecting $x$ with $y$ crosses the arc $\wideparen{w_1w_2}$ (endpoints excluded).
			Since no robot lays within the arc $\wideparen{w_1w_2}$, any robot on $\wideparen{uv}$ can see any robot on $\diam'$.

		\item
			The robot $e$ can always recompute its rank $j$ (according to the previous Items), the circle $\SC$, and the diameter $\diam$.
			Thanks to \Cref{item:slice_phase2_see_beacons}, $e$ can see the three beacons on $\wideparen{uv}$ and so it can reconstruct $\barc_e$.
		Since $e$ has complete visibility of $\SC_w$, it counts the $m$ \colr{west} robots and so it can define the positions of all the uniform positions on $\barc_e$.
		Thus, it computes its target position $U_j$.

		\item
		See \Cref{fig:slice_stage2_east_ui}.  
		Let $b\in\{e_1,e_2,e_3\}$ be a \colr{beacon} robot on $\barc_w$ that sees no \colr{east\_diameter} robots or moving robots within the block.
		Note that, in this situation, $b$ has complete visibility of $\barc_e$ except for at most one robot which can be hidden by the \colr{angle} robot.
		So, following a sequential scheme ($b=e_3,e_2,e_1$), $b$ recomputes $\barc_e$, and the angle of each uniform arc $\wideparen{U_jU_{j+1}}$ and so its target position.
		Thus, $b$ reaches its target uniform position and sets its color as \colr{regular}.
		In particular, if \colr{angle} lays on the trajectory between $b$ and its target position $U_j$, then $b$ reaches $U_j$ in two moves: firstly, it stops in a position on $\diam$ which is visible to the other \colr{beacon} on $\barc_w$, secondly, it reaches $U_j$.
		Note that at most one beacon has to split its trajectories into two steps (in fact, \colr{angle} can be collinear with just one $e_j$ and the related $U_j$).
		So, this task is accomplished in constant time.
	\end{enumerate}
\end{proof}

\begin{lemma}\label{lemma:slice_stage3_proof}
During Stage 3, the following statements hold: 
	\begin{enumerate}
		\item{}{\bf [Step 1]}
		Robots $w_1$ and $w_2$ can safely and correctly travel to $\SC_e$ fixing the angle $\delta$ and the cardinality $m$.

            \item{}{\bf [Step 1]}
		If a \colr{west} robot on $\SC_w$ sees a \colr{west\_diameter} or a \colr{sliceMedian} robot, then it is always able to determine whether such robots belong to its block or not. 

		\item{}{\bf [Step 1]}\label{item:slice_stage3_beacons}
			The \colr{west} robots $w_3,w_{m-1},w_m$ migrate to $\rho$ in constant time in order to play the role of the beacons (together with $w_0$, the median robot) for the other \colr{west} robots;
			
		\item{}{\bf [Step 2]}\label{item:step3_compute_C}
		A \colr{west\_diameter} robot $w$ on $\diam$ can correctly and instantly compute $\SC_e$ and $\SC_w$ and its original slice $\eta_k$. So, it correctly decodes its rank $j$. 
		
		\item{}{\bf [Step 2]}
		A \colr{west\_diameter} robot $w$ on $\diam$ can correctly spot $\barc_w$ and its target uniform position $U'_j$.

		\item{}{\bf [Step 2]}
		The \colr{angle} and \colr{anglem} robots reach $U'_1$ and $U'_2$ without colliding.
	\end{enumerate}
\end{lemma}
\begin{proof}
We prove each statement:
	\begin{enumerate}
	
		\item
		The \colr{angle} robot ($w_1$) can detect Stage 2 has ended since it can see all the uniform positions on $\barc_e$ are covered by \colr{regular} robots.
		So, it can safely migrate on its projection on $\SC_e$, maintaining the angle $\delta$.
		Once $w_1$ has switched half-circle, $w_2$ enjoys complete visibility of all the $m+1$ robots on $\SC$ and so it can correctly migrate on $\SC_w$ to fix the angle $\frac{\delta}{m}$.

            \item
		Let us suppose that a \colr{west} robot $w$ on $\SC_w$ sees some robots (\colr{west\_diameter} or \colr{sliceMedian}) of another median diameter $\diam'$, and it cannot see the robots on its $\diam$ (due to the movements of the other robots).
		This situation never creates ambiguous snapshots since $w$ can always understand $\diam'$ does not belong to its block.
		In fact, the moving robots in the block of $w$ can never hide all the \colr{regular} robots on $\barc_e$: the visibility of $w$ can be obstructed by at most $m-6$ robots which are moving inside $\SC$ (where the 6 still robots are: $w$ itself, the \colr{angle} and \colr{anglem} robots, and the three beacons $w_3,w_{m-1},w_m$). 
            So $w$ can see at least 6 \colr{regular} robots on its block.
		Such \colr{regular} robots are sufficient to prevent $w$ from selecting the wrong median diameter.

		\item 
            Before this step, $w_0$ (the median robot) is the only robot on $\diam$.
		Robots $w_3$ (as the ``left'' beacon) and $w_{m-1},w_m$ (as ``right'' beacons) can easily and safely elect themselves and reach their projections on $\diam$, by following a sequential scheme (i.e. in constant time). 
        Traveling along parallel trajectories, they do not collide.
		
		\item{}
		Let $w$ be a \colr{west\_diameter} robot on $\diam$.
		It can always recompute the supporting line of $\diam$ ($w$ can always see another \colr{west\_diameter} robot or the \colr{sliceMedian} robot).
		Moreover, it can always see the robots \colr{angle} and \colr{anglem}.
		Thus, $w$ can recompute $\SC$ (the center of $\SC$ is the intersection point between the bisector of \colr{angle}-\colr{anglem} and the supporting line of $\diam$) and its projection $w'$ on $\SC_w$.
		Then, $w$ splits $\SC_e$ and $\SC_w$ in $\delta$-slices ($\delta$ being fixed by the \colr{angle} robot) and computes the slice $\eta_k$ where $w'$ lays on (i.e. where $w$ originally laid before the migration on $\diam$).
	
		 After having recomputed $\eta_k$, $w$ obtains its original rank $j$ by inverting the formula $\angle w_0Ow'= k\delta + \frac{j\delta}{m+1}$, where $m$ is fixed by \colr{angle} and \colr{anglem}.
		 
		 \item{}
			See \Cref{fig:slice_west_part}. 
		 Thanks to \Cref{item:step3_compute_C}, we know that $w$ can compute $\SC_e$, $\SC_w$, and its rank $j$, even when other robots are moving to reach their uniform positions on $\barc_w$.
		 Moreover, $w$ can see at least three robots on $\barc_e$ (one of them is \colr{sliceR}) and so it recomputes $\barc_w$.
	The presence of the \colr{angle} and \colr{anglem} robots allows $w$ to recompute $m$ and so the positions of each uniform position. 
Thus, it can easily spot the position of its target $U'_j$.

	\item{}
		For the sake of simplicity, we make \colr{angle} and \colr{anglem} move sequentially.
		Let $a$ be the \colr{angle} or \colr{anglem} robot.
		When $a$ sees no \colr{west\_diameter} robots, $a$ easily recomputes the (possible) missing uniform positions $U'_1$ or $U'_2$ (by the presence of the \colr{sliceMedian} robot and the robots on $\barc_{\{e,w\}}$).
	Then, it heads to its $U'_j$ without colliding with any other robots (due to the absence of robots on $\diam$ and $\SC_w$, except for the \colr{sliceMedian} robot).
	\end{enumerate}
\end{proof}

\begin{lemma}\label{lemma:slice_bdcp}
During Stage 3 of \proc{Slice}, all the \colr{west} robots on $\SC_w$ can reach $\diam$ in $O(1)$ epochs by implementing {\BDCP}.
\end{lemma}
\begin{proof}
    We verify that all the input assumptions listed in \Cref{alg:bdcp} hold to correctly implement {\BDCP} in $O(1)$ time.
    \begin{enumerate}
	\item $\diam$ is the target $2$-curve;
        \item 4 beacons ($w_3,w_{m-1},w_m$ together with the median robot $w_0$) are originally set to fix $\diam$ (see \Cref{lemma:slice_stage3_proof}, \Cref{item:slice_stage3_beacons});
        \item the \colr{west} robots are the waiting robots which have to reach $\diam$, so that
            \begin{enumerate}
                \item since they start from $\SC_w$ (excluded its endpoints), they are external to $\diam$, belonging to the same half-plane delimited by $\diam$;
                \item at the beginning of this step, no robot is moving within $\SC$, so each \colr{west} robot can see all the robots on $\diam$;
                \item they have to reach $\diam$
                    \begin{itemize}
                        \item to their projections on $\diam$. Being the beacons originally the ``leftmost'' and the ``rightmost'' \colr{west} robots on $\SC_w$, the \colr{west} robots will migrate between the two groups of beacons;
                        \item being parallel, the \colr{west} robots' trajectories do not intersect each other; 
                        \item such trajectories are perpendicular to $\diam$, so they intersect each other in just one point;
                        \item if a \colr{west} robot on $\SC_w$ can see its target line $\diam$, then it can reconstruct its trajectory (the projection on $\diam$);
                    \end{itemize}
            \item each \colr{west} robot sets its color as \colr{west\_diameter} once on $\diam$, playing the role of a new beacon fixing the chord $\diam$. 
            \end{enumerate}
    \end{enumerate}
    All these assumptions guarantee all \colr{west} robots reach $\diam$ in $O(\log{2}) = O(1)$ epochs without colliding.
\end{proof}

\subsection{\proc{Sequential Match}}\label{sec:correctness_procedure_seqmatch}
We show here all the correctness proofs for \proc{Sequential Match} (\Cref{section:seqmatch}).

\begin{lemma}
Given \conf{unisect} with $q<12$ robots in each uniform sector (boundaries excluded), a \colr{pre\_matched} robot can compute and reach its target uniform position $U_j$ even if other \colr{to\_matched} robots are moving and hiding \colr{right} or \colr{left} robots.
\end{lemma}
\begin{proof}
The proof is similar to the one for \Cref{lemma:to_left_compute_destination,lemma:to_x_compute_destination}.
In particular, a \colr{pre\_matched} robot can always spot the \colr{left} and \colr{right} robots of its uniform sector $\Upsilon_i$.
Using the same strategy as \Cref{lemma:to_left_compute_destination,lemma:to_x_compute_destination}, it recomputes the boundaries of the uniform sectors and the central angle $2\pi/n$ formed by two consecutive uniform positions (occupied by \colr{left}, \colr{right} or \colr{regular} robots).
Since a \colr{pre\_matched} robot has complete visibility of its uniform sector (in fact, no robot is moving within it), it counts the \colr{matched} robots in $\Upsilon_i$ and, so, it recomputes the current index $j$ of the uniform position to be matched.
Thus, it heads to the uniform position $U_j$ of $\Upsilon_i$ and sets its color as \colr{matched}.
\end{proof}

\end{document}